%% file: main.tex
\pdfoutput=1

\documentclass[sigconf,nonacm]{acmart}
\setcopyright{none}
\input{commands}

\usepackage[subtle, lists=tight, mathspacing=normal, mathdisplays=normal]{savetrees}

\begin{document}

\title{Longest Chain Consensus Under Bandwidth Constraint}

\author{Joachim Neu}
\email{jneu@stanford.edu}
\affiliation{\country{}}
\author{Srivatsan Sridhar}
\email{svatsan@stanford.edu}
\affiliation{\country{}}
\author{Lei Yang}
\email{leiy@csail.mit.edu}
\affiliation{\country{}}
\author{David Tse}
\email{dntse@stanford.edu}
\affiliation{\country{}}
\author{Mohammad Alizadeh}
\email{alizadeh@csail.mit.edu}
\affiliation{\country{}}

\thanks{JN, SS and LY contributed equally and are listed alphabetically.}

\begin{abstract}
Spamming attacks are a serious concern for consensus protocols, as witnessed by recent outages of a major blockchain, Solana. They cause congestion and excessive message delays in a \emph{real} network due to its bandwidth constraints. In contrast, longest chain (LC), an important family of consensus protocols, has previously only been proven secure assuming an \emph{idealized} network model in which all messages are delivered within bounded delay. This model-reality mismatch is further aggravated for Proof-of-Stake (PoS) LC where the adversary can spam the network with equivocating blocks. Hence, we extend the network model to capture bandwidth constraints, under which nodes now need to choose carefully which blocks to spend their limited download budget on. To illustrate this point, we show that `download along the longest header chain', a natural download rule for Proof-of-Work (PoW) LC, is insecure for PoS LC. We propose a simple rule `download towards the freshest block', formalize two common heuristics `not downloading equivocations' and `blocklisting', and prove in a unified framework that PoS LC with any one of these download rules is secure in bandwidth-constrained networks. In experiments, we validate our claims and showcase the behavior of these download rules under attack. By composing multiple instances of a PoS LC protocol with a suitable download rule in parallel, we obtain a PoS consensus protocol that achieves a constant fraction of the network's throughput limit even under worst-case adversarial strategies.
\end{abstract}

\maketitle

\input{1_introduction}
\input{23_modelprotocol}

\input{4_analysis_high_level}

\input{5_analysis_details}
\input{6_experiments}
\input{8_equivocation_avoidance}
\input{7_parallel}

\section*{Acknowledgment}
We
thank Dan Boneh and Ertem Nusret Tas for fruitful discussions.
JN is supported by the Protocol Labs PhD Fellowship,
a gift from Ethereum Foundation,
and the Reed-Hodgson Stanford Graduate Fellowship.

\bibliographystyle{ACM-Reference-Format}
\bibliography{references}

\appendix
\input{appendix_pseudocode}

\input{appendix_experiments}

\input{appendix_parallel}
\input{appendix_proofs}
\input{appendix_proofs_parallel}
\input{appendix_conflux}

\end{document}

%% file: commands.tex
\usepackage[utf8]{inputenc}
\usepackage{hyperref}
\usepackage{graphicx}
\usepackage{xcolor}

\usepackage{amsmath}
\usepackage{amsthm}
\usepackage{amsfonts}
\usepackage{bbm}

\usepackage{datetime2}

\usepackage{IEEEtrantools}
\allowdisplaybreaks

\usepackage{tikz}
\usetikzlibrary{calc}
\usetikzlibrary{arrows}
\usetikzlibrary{arrows.meta}
\usetikzlibrary{patterns}
\usetikzlibrary{positioning}
\usetikzlibrary{decorations.pathreplacing}
\usetikzlibrary{shapes.misc}
\usetikzlibrary{spy}

\usepackage{pgfplots}
\pgfplotsset{compat=1.14}

\definecolor{myParula01Blue}{RGB}{0,114,189}
\definecolor{myParula02Orange}{RGB}{217,83,25}
\definecolor{myParula03Yellow}{RGB}{237,177,32}
\definecolor{myParula04Purple}{RGB}{126,47,142}
\definecolor{myParula05Green}{RGB}{119,172,48}
\definecolor{myParula06LightBlue}{RGB}{77,190,238}
\definecolor{myParula07Red}{RGB}{162,20,47}

\tikzset{myparula11/.style={color=myParula01Blue,solid,mark=+,mark options={solid}}}
\tikzset{myparula12/.style={color=myParula01Blue,densely dashed,mark=x,mark options={solid}}}
\tikzset{myparula13/.style={color=myParula01Blue,densely dotted,mark=o,mark options={solid}}}
\tikzset{myparula14/.style={color=myParula01Blue,dashdotted,mark=triangle,mark options={solid}}}
\tikzset{myparula15/.style={color=myParula01Blue,dashdotdotted,mark=square,mark options={solid}}}

\tikzset{myparula21/.style={color=myParula02Orange,solid,mark=+,mark options={solid}}}
\tikzset{myparula22/.style={color=myParula02Orange,densely dashed,mark=x,mark options={solid}}}
\tikzset{myparula23/.style={color=myParula02Orange,densely dotted,mark=o,mark options={solid}}}
\tikzset{myparula24/.style={color=myParula02Orange,dashdotted,mark=triangle,mark options={solid}}}
\tikzset{myparula25/.style={color=myParula02Orange,dashdotdotted,mark=square,mark options={solid}}}

\tikzset{myparula31/.style={color=myParula03Yellow,solid,mark=+,mark options={solid}}}
\tikzset{myparula32/.style={color=myParula03Yellow,densely dashed,mark=x,mark options={solid}}}
\tikzset{myparula33/.style={color=myParula03Yellow,densely dotted,mark=o,mark options={solid}}}
\tikzset{myparula34/.style={color=myParula03Yellow,dashdotted,mark=triangle,mark options={solid}}}
\tikzset{myparula35/.style={color=myParula03Yellow,dashdotdotted,mark=square,mark options={solid}}}

\tikzset{myparula41/.style={color=myParula04Purple,solid,mark=+,mark options={solid}}}
\tikzset{myparula42/.style={color=myParula04Purple,densely dashed,mark=x,mark options={solid}}}
\tikzset{myparula43/.style={color=myParula04Purple,densely dotted,mark=o,mark options={solid}}}
\tikzset{myparula44/.style={color=myParula04Purple,dashdotted,mark=triangle,mark options={solid}}}
\tikzset{myparula45/.style={color=myParula04Purple,dashdotdotted,mark=square,mark options={solid}}}

\tikzset{myparula51/.style={color=myParula05Green,solid,mark=+,mark options={solid}}}
\tikzset{myparula52/.style={color=myParula05Green,densely dashed,mark=x,mark options={solid}}}
\tikzset{myparula53/.style={color=myParula05Green,densely dotted,mark=o,mark options={solid}}}
\tikzset{myparula54/.style={color=myParula05Green,dashdotted,mark=triangle,mark options={solid}}}
\tikzset{myparula55/.style={color=myParula05Green,dashdotdotted,mark=square,mark options={solid}}}

\tikzset{myparula61/.style={color=myParula06LightBlue,solid,mark=+,mark options={solid}}}
\tikzset{myparula62/.style={color=myParula06LightBlue,densely dashed,mark=x,mark options={solid}}}
\tikzset{myparula63/.style={color=myParula06LightBlue,densely dotted,mark=o,mark options={solid}}}
\tikzset{myparula64/.style={color=myParula06LightBlue,dashdotted,mark=triangle,mark options={solid}}}
\tikzset{myparula65/.style={color=myParula06LightBlue,dashdotdotted,mark=square,mark options={solid}}}

\tikzset{myparula71/.style={color=myParula07Red,solid,mark=+,mark options={solid}}}
\tikzset{myparula72/.style={color=myParula07Red,densely dashed,mark=x,mark options={solid}}}
\tikzset{myparula73/.style={color=myParula07Red,densely dotted,mark=o,mark options={solid}}}
\tikzset{myparula74/.style={color=myParula07Red,dashdotted,mark=triangle,mark options={solid}}}
\tikzset{myparula75/.style={color=myParula07Red,dashdotdotted,mark=square,mark options={solid}}}

\pgfplotsset{
    mysimpleplot/.style = {
        every axis plot/.prefix style={thick},
        width=1.05\linewidth,
        height=0.75\linewidth,
        title style={font=\footnotesize,align=center},
        legend cell align=left,
        legend style={font=\footnotesize},
        legend columns=3,
        legend style={
            at={(0.5,1)},
            yshift=0.3em,
            anchor=south,
            draw=none,
            /tikz/every even column/.append style={
                column sep=0.3em
            },
            cells={
                align=left
            }
        },
        grid=both,
        minor tick num=4,
        major grid style={solid,draw=gray!50},
        minor grid style={densely dotted,draw=gray!50},
        label style={font=\footnotesize,align=center},
        tick label style={font=\footnotesize},
    },
}

\usepackage{multirow,booktabs}
\usepackage{makecell}

\usepackage{algorithmicx}
\usepackage[noend]{algpseudocode}
\usepackage{algorithm}

\algnewcommand{\algorithmicswitch}{\textbf{switch}}
\algdef{SE}[SWITCH]{Switch}{EndSwitch}[1]{\algorithmicswitch\ #1\ \algorithmicdo}{\algorithmicend\ \algorithmicswitch}%
\algtext*{EndSwitch}%

\algnewcommand{\algorithmiccase}{\textbf{case}}
\algdef{SE}[CASE]{Case}{EndCase}[1]{\algorithmiccase\ #1}{\algorithmicend\ \algorithmiccase}%
\algtext*{EndCase}%

\algnewcommand{\algorithmicon}{\textbf{on}}
\algdef{SE}[ON]{On}{EndOn}[1]{\algorithmicon\ #1\ \algorithmicdo}{\algorithmicend\ \algorithmicon}%
\algtext*{EndOn}%

\algnewcommand{\algorithmicrealfunction}{\textbf{function}}
\algdef{SE}[REALFUNCTION]{RealFunction}{EndRealFunction}[1]{\algorithmicrealfunction\ #1\ \algorithmicdo}{\algorithmicend\ \algorithmicrealfunction}%
\algtext*{EndRealFunction}%

\algrenewcommand{\algorithmicdo}{}
\algrenewcommand{\algorithmicthen}{}

\algnewcommand{\algorithmicgoto}{\textbf{goto}}%
\algnewcommand{\Goto}[1]{\algorithmicgoto~\ref{#1}}%

\algnewcommand{\algorithmicassert}{\textbf{assert}}%
\algnewcommand{\Assert}[1]{\algorithmicassert~{#1}}%

\algnewcommand{\algorithmicbreak}{\textbf{break}}%
\algnewcommand{\Break}[0]{\algorithmicbreak}%

\algnewcommand{\algorithmicwaiton}{\textbf{wait on}}%
\algnewcommand{\WaitOn}[1]{\algorithmicwaiton~{#1}}%

\algnewcommand{\LineComment}[1]{\State \(\triangleright\) \textit{#1}}
\algnewcommand{\InlineRequire}[1]{\textbf{require} {#1}}

\theoremstyle{plain}
\newtheorem{theorem}{Theorem}
\newtheorem*{theorem*}{Theorem}
\newtheorem{corollary}{Corollary}
\newtheorem{lemma}{Lemma}
\newtheorem{proposition}{Proposition}

\theoremstyle{definition}
\newtheorem{definition}{Definition}

\newcommand{\Prob}[1]{\ensuremath{\operatorname{Pr}\left[#1\right]}}

\newcommand{\poly}{\ensuremath{\mathrm{poly}}}
\newcommand{\negl}{\ensuremath{\mathrm{negl}}}

\usepackage{xspace}
\newcommand{\cf}[0]{cf.\xspace}
\newcommand{\ie}[0]{\emph{i.e.}\xspace}
\newcommand{\eg}[0]{\emph{e.g.}\xspace}

\usepackage{dsfont}
\newcommand{\ind}[1]{\ensuremath{\mathds{1}\!\left\{#1\right\}}}
\newcommand{\TPivot}[1]{\mathsf{Pivot}(#1)}

\newcommand{\Adverse}[1]{\mathsf{Adv}(#1)}
\newcommand{\Unique}[1]{\mathsf{Unique}(#1)}
\newcommand{\WeakPivot}[1]{\mathsf{WeakPivot}_{w}(#1)}
\newcommand{\grprom}[1]{\ensuremath{\mathsf{MaxDL}_{\bwslot,#1}}}
\newcommand{\grpromempty}{\ensuremath{\mathsf{MaxDL}_{\bwslot}}}
\newcommand{\grpromemptycustom}[1]{\ensuremath{\mathsf{MaxDL}_{#1}}}
\newcommand{\Bint}[1]{\ensuremath{\mathcal{B}\!\left(#1\right]}}
\newcommand{\Lint}[1]{\ensuremath{\mathcal{U}\!\left(#1\right]}}
\newcommand{\Nint}[1]{\ensuremath{\mathcal{A}\!\left(#1\right]}}
\newcommand{\exec}{\mathcal{E}^{\rho,\beta,\Th}}
\newcommand{\len}[1]{\ensuremath{|#1|}}
\newcommand{\trunc}[1]{\ensuremath{^{\lceil #1}}}
\newcommand{\dlrule}{\ensuremath{\mathcal{D}}}
\newcommand{\dlrulefresh}{\ensuremath{\mathcal{D}_{\mathsf{fresh}}}}
\newcommand{\dlruleeqavoid}{\ensuremath{\mathcal{D}_{\mathsf{lhc-ea}}}}
\newcommand{\dlruleblock}{\ensuremath{\mathcal{D}_{\mathsf{blist}}}}

\newcommand{\pa}{\ensuremath{p_\mathrm{A}}}
\newcommand{\pu}{\ensuremath{p_\mathrm{U}}}
\newcommand{\Th}[0]{\ensuremath{T_{\mathrm{h}}}}
\newcommand{\PivotCondition}[0]{\ensuremath{\mathsf{PivotCondition}}}
\newcommand{\SecurityCondition}[0]{\ensuremath{\pu=\frac12 p(1+\epsilon_1)}\xspace}
\newcommand{\SecurityConditionFull}{\ensuremath{(1-\beta)\rho e^{-\rho} = \frac{1-e^{-\rho}}{2}(1+\epsilon_1)}}
\newcommand{\DeltaHeader}{\ensuremath{\Delta_{\mathrm{h}}}}

\newcommand{\Tp}[0]{\ensuremath{\gamma}}
\newcommand{\FrequentPivots}{\mathsf{FreqPivots}_{\gamma}}
\newcommand{\FewBlockOpps}{\mathsf{ShortPrefixes}_{\bwslot}}
\newcommand{\BoundedEq}{\mathsf{FewLongChains}_{\downloadLimit}}
\newcommand{\Td}{\ensuremath{T}}
\newcommand{\Tb}{\ensuremath{T_{\mathrm{b}}}}
\newcommand{\greatblock}{great block\xspace}

\newcommand{\node}[1]{\ensuremath{#1}}
\newcommand{\bwslot}{\ensuremath{K}}
\newcommand{\bwtime}{\ensuremath{C}}

\newcommand{\Tconf}{\ensuremath{T_{\mathrm{conf}}}}

\newcommand{\Tlive}{\ensuremath{T_{\mathrm{live}}}}
\newcommand{\protocol}{\ensuremath{\Pi^{\rho,\tau,\Tconf}}}
\newcommand{\protocolpc}{\ensuremath{\Pi^{\rho,\tau,\Tconf,m}_{\mathrm{pc}}}}

\newcommand{\calC}{\ensuremath{\mathcal{C}}}
\newcommand{\TP}{\ensuremath{\mathrm{TP}}}
\newcommand{\idleslots}{\ensuremath{\phi_{\mathrm{idle}}}}
\newcommand{\bwpassive}{\ensuremath{\phi_{\mathrm{p}}}}
\newcommand{\tpinc}{\ensuremath{\theta_{\mathrm{inc}}}}

\newcommand{\constFrequentPivots}{\ensuremath{\alpha_1}}
\newcommand{\constPivotRace}{\ensuremath{\alpha_1'}}
\newcommand{\constPivotInterval}{\ensuremath{\alpha_1''}}
\newcommand{\constNumBlocks}{\ensuremath{\alpha_2}}
\newcommand{\constBoundedEq}{\ensuremath{\alpha_3}}

\DeclareMathOperator*{\argmax}{arg\,max}
\DeclareMathOperator*{\argmin}{arg\,min}

\newcommand{\Env}{\mathcal Z}
\newcommand{\Adv}{\mathcal A}

\newcommand{\TRUE}{\ensuremath{\mathtt{true}}}
\newcommand{\FALSE}{\ensuremath{\mathtt{false}}}
\newcommand{\txs}{\ensuremath{\mathsf{txs}}}

\newcommand{\LOG}[2]{%
    \ensuremath{\mathsf{LOG}_{#1}^{#2}}
}

\newcommand{\Ftree}{\mathcal F_\mathrm{headertree}^{\rho}}
\newcommand{\Fparallel}{\mathcal F_\mathrm{parallel}^{\rho,m}}

\newcommand{\NIL}{\mathtt{unknown}}
\newcommand{\INVALID}{\mathtt{invalid}}

\newcommand{\Tree}{\mathcal{T}}
\newcommand{\Chain}{\mathcal C}

\newcommand{\HeaderTree}{\mathrm{h}\mathcal{T}}
\newcommand{\DownloadedChain}{\mathrm{d}\Chain}
\newcommand{\dC}{\DownloadedChain}
\newcommand{\TxsMap}{\mathrm{blkTxs}}
\newcommand{\downloadLimit}{\ensuremath{\bwslot}}

\newcommand{\Lottery}{\mathrm{lottery}}
\renewcommand{\index}{\mathsf{idx}}
\newcommand{\primindex}{\mathsf{pri}}
\newcommand{\maxt}{\mathsf{tmax}}

\tikzset{blockchain/.style={
        x=1.5cm,
        y=0.6cm,
        node distance=0.5cm,
        block/.style = {
            minimum width=0.25cm,
            minimum height=0.25cm,
            draw,
            shade,
            top color=white,
            bottom color=black!10,
        },
        block-adv1/.style = {
            block,
            bottom color=myParula01Blue!50,
            draw=myParula01Blue!50!black
        },
        block-adv2/.style = {
            block,
            bottom color=myParula07Red!50,
            draw=myParula07Red!50!black,
        },
        block-adv3/.style = {
            block,
            bottom color=myParula05Green!50,
            draw=myParula05Green!50!black,
        },
        block-green/.style = {
            block,
            bottom color=myParula05Green!50,
            draw=myParula05Green!50!black,
        },
        block-red/.style = {
            block,
            bottom color=myParula07Red!50,
            draw=myParula07Red!70!black,
        },
        block-gray/.style = {
            block,
            bottom color=black!30,
        },
        block-big/.style = {
            minimum width=0.7cm,
            minimum height=0.7cm,
            draw,
            shade,
            top color=white,
            bottom color=black!10,
        },
        branch/.style = {
            minimum width=0.1cm,
            minimum height=0.1cm,
            draw,
            circle,
            inner sep=0,
            fill=black,
        },
        link/.style = {
            -latex,
        },
        hiddenlink/.style = {
            dashed,
        },
        hiddenlink-adv1/.style = {
            hiddenlink,
            draw=myParula01Blue!50!black,
        },
        hiddenlink-adv2/.style = {
            hiddenlink,
            draw=myParula07Red!50!black,
        },
        hiddenlink-adv3/.style = {
            hiddenlink,
            draw=myParula05Green!50!black,
        },
        label/.style = {
        },
        label-adv1/.style = {
            label,
            text=myParula01Blue!50!black,
        },
        label-adv2/.style = {
            label,
            text=myParula07Red!50!black,
        },
        label-adv3/.style = {
            label,
            text=myParula05Green!50!black,
        },
    }
}

\pgfplotsset{
    discard if not/.style 2 args={
        x filter/.code={
            \edef\tempa{\thisrow{#1}}
            \edef\tempb{#2}
            \ifx\tempa\tempb
            \else
                
            \fi
        }
    },
    myresultsplot01/.style={
        legend style={
            at={(0.5,1.15)},
            anchor=north,
            legend columns=-1,
            draw=none,
            /tikz/every even column/.append style={column sep=1em},
            cells={align=left},
            font=\small,
        },
        enlarge x limits=0.15,
        ybar,
        bar width=6mm,
        xtick=data,
        width=\linewidth,
        height=0.6\linewidth,
        ymin=0,
        grid=both,
        minor y tick num=4,
        minor x tick num=1,
        major grid style={solid,draw=gray!50},
        minor grid style={densely dotted,draw=gray!50},
        label style={font=\small},
        tick label style={font=\small},
    },
}

\usepackage{pifont}%
\newcommand{\xmark}{\ding{56}}%

\newcommand*\circledScriptGray[1]{\tikz[baseline=(char.base)]{
            \node[shape=circle,draw=black,text=black,inner sep=2pt,fill=black!10] (char) {\scriptsize #1};}}
\newcommand*\circledScriptGraySlim[1]{\tikz[baseline=(char.base)]{
            \node[shape=circle,draw=black,text=black,inner sep=1pt,fill=black!10] (char) {\scriptsize #1};}}

%% file: 1_introduction.tex
\vspace{-.7em}
\section{Introduction}
\label{sec:introduction}

\paragraph{Consensus}
In the state machine replication (SMR) formulation
of the consensus problem,
a group of \emph{nodes}
aim to order \emph{transactions}
received from the environment into a common \emph{ledger}.
For this purpose, nodes exchange messages and perform computations
as prescribed by the consensus protocol.
Consensus is made non-trivial
by an adversary
who has some control over message delays,
controls a certain fraction of nodes,
and can cause them to deviate from the protocol
in an arbitrary (\emph{Byzantine}) manner
in a concerted effort to disturb consensus.
\emph{Secure} consensus is achieved if the resulting transaction
ledgers across different honest nodes and points in time
are \emph{consistent}
so that it is meaningful to speak of
\emph{the} single common ledger (which is \emph{safe}),
and if that ledger is \emph{live} in the sense that
every transaction gets assigned a position in the ledger
soon after it is input to honest nodes for the first time.

\paragraph{Nakamoto's Longest Chain Protocol}
In the seminal Bitcoin whitepaper \cite{nakamoto_paper},
Satoshi Nakamoto describes the \emph{longest chain} (LC) consensus protocol.
In this protocol,
honest nodes broadcast blocks to each other.
A block contains a list of transactions,
a nonce, and a reference to a parent block,
resulting in chains of blocks
up to a root genesis block that is common knowledge.
A block is \emph{valid} if a cryptographic hash of it is smaller than a certain
fixed threshold, and if the transactions it contains have been legitimized by
the owners of the affected assets
and are consistent with respect
to transactions preceding it as ordered in the same block and its ancestor blocks.
Every node adds valid blocks it receives to its local copy of the block tree.
Nodes also aim to produce new blocks.
For this purpose they bundle recently received transactions
together with a reference to the block at the tip of the longest chain
in their local block tree
and use brute force search to determine
a nonce such that the resulting block is valid (\ie, the hash inequality is satisfied).
Newfound valid blocks are broadcast to other nodes, completing the process.
Each node outputs as ledger the transactions as ordered in the prefix
of the block that is $k$-deep in the longest chain
of its local block tree.

Besides being remarkably simple,
Nakamoto's LC consensus protocol has two outstanding properties.
First, it enables consensus in a \emph{permissionless} setting
by using \emph{proof-of-work} (PoW) puzzles as a Sybil resistance mechanism \cite{pow1,pow2}.
The bottleneck to block production is finding nonces which
lead to valid blocks which satisfy the hash inequality,
and
as long as the majority of hash power at every point in time
is controlled by honest nodes, honest nodes output a secure ledger 
\cite{backbone,pss16}.
Second, the LC can tolerate \emph{dynamic participation} in the sense
that the ledger remains secure
even as the total hash power participating in the protocol as well as its distribution among participants varies over time.

\paragraph{Proof-of-Stake Longest Chain}
A drawback of Nakamoto's PoW LC is the high electricity consumption
and as a result a tendency for centralization of nodes at places of relatively low electricity cost.
To overcome the drawbacks of PoW LC while retaining its advantages,
protocols such as Ouroboros \cite{kiayias2017ouroboros,david2018ouroboros,badertscher2018ouroboros}
and Sleepy Consensus \cite{sleepy,snowwhite} preserve the operating principle
of the LC
but
replace PoW with \emph{proof-of-stake} (PoS) lotteries, where a party is assigned
random block production
opportunities in proportion to the amount of stake it holds in the system, effectively
substituting `one CPU, one vote' by `one coin, one vote'.
For this purpose,
nodes use synchronized clocks
to count time slots of a predetermined duration.
For every time slot, nodes evaluate
a block production lottery associated
with their cryptographic identity.
For instance in \cite{david2018ouroboros,badertscher2018ouroboros},
nodes get to produce a new valid block
if the output of a \emph{verifiable random function} (VRF)
is below a threshold proportional to
the node's stake.

\paragraph{Proof-of-Stake Longest Chain Under Bandwidth Constraint}
While PoS LC behaves in some aspects similar to PoW LC,
it differs drastically in others.
For instance, in PoS, block production opportunities can be `reused'
in the sense that when a node is eligible to produce a block in a certain time slot,
it can in fact create many equivocating but equally valid blocks for the same time slot,
each potentially with a different set of transactions
and/or attached to a different parent block.
This problem arises because block production `lottery tickets'
in PoS
can not depend on the proposed block's transactions.
Otherwise an adversary could increase
its chances to produce a block by trying
various sets of transactions (\emph{grinding}). 
Similarly, the PoS lotteries can not depend on the parent block, as the adversary could extend several chains at once to increase their chance of block production (\emph{nothing-at-stake} attack \cite{pos_paper}).
In PoW however, each block production opportunity corresponds to a unique block (a combination of transaction set, parent block, and nonce), 
thus the rate of block production opportunities simultaneously bounds
the rate at which new valid blocks can be created.

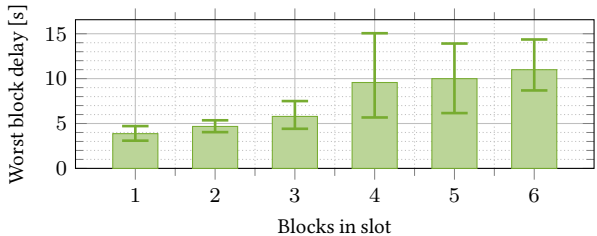
\begin{figure}
    \centering
    \begin{tikzpicture}
        \begin{axis}[
                myresultsplot01,
                ylabel={Worst block delay [s]},
                xlabel={Blocks in slot},
                height=0.42\linewidth,
            ]

            \addplot+ [
                color=myParula05Green,
                fill=myParula05Green!50,
                error bars/.cd,
                y explicit,
                y dir=both,
                error bar style={
                    line width=1pt,
                },
                error mark options={
                    rotate=90,
                    mark size=5pt,
                    line width=1pt
                },
            ] table [
                x=blocks,
                y expr=\thisrow{mean},
                y error minus expr=\thisrow{mean}-\thisrow{p10},
                y error plus expr=\thisrow{p90}-\thisrow{mean},
            ] {figures/network_delay_vs_block_rate.txt};
            
            \legend{};

        \end{axis}
    \end{tikzpicture}%
    \vspace{-0.3em}%
    \caption{Time taken ($10$-th percentile, mean, $90$-th percentile) for all nodes to download all blocks mined in a slot, when different number of new blocks are produced and broadcast in a slot.
    The delay
    increases as the number of blocks 
    is increased,
    showing that network delay 
    is not
    independent of network load. We use Cardano's Ouroboros implementation. Details of the experimental setup are given in Appendix~\ref{sec:appendix-experiment-details-fig1}.}
    \label{fig:networkdelayblockrate}
\end{figure}

Previous analysis \cite{david2018ouroboros,sleepy,dem20} shows that this difference is immaterial
in the synchronous network model where the message propagation delay between honest nodes
is controlled by the adversary, but below a known upper bound $\Delta$. Under such a network model, PoS LC and PoW LC
behave the same in terms of security, transaction throughput and confirmation latency.
This model, however, is over-idealized in that it assumes a fixed
delay upper bound for every single message, 
even when
many
messages are transmitted simultaneously (which may be under normal execution or due to adversarial actions). The model does not capture notions of capacity
and congestion which have a significant impact on the behavior of real networks.
In fact, an increase in network delay with increasing network load
(via increased block size) has been demonstrated previously for Bitcoin \cite{btcnetworkdelay}.
Similarly, increasing the network load (via increasing the number of blocks per slot)
leads to increased network delay in
our experiments (see Figure~\ref{fig:networkdelayblockrate})
with
Cardano's Ouroboros implementation---a PoS protocol.
Once we enrich the network model to capture such phenomena,
the difference in the behavior of PoW LC and PoS LC with respect to
reuse of block production opportunities strikes.
The possibility of producing (infinitely) many equivocating valid blocks per
block production opportunity opens up new adversarial strategies in which
the adversary aims to exhaust limited network resources with useless spam
in an attempt to disturb consensus.
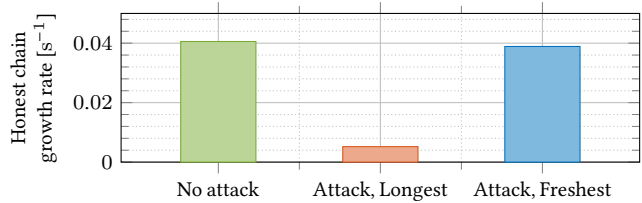
\begin{figure}
    \centering
    \begin{tikzpicture}
        \begin{axis}[
                myresultsplot01,
                ylabel={Honest chain\\growth rate [$\mathrm{s}^{-1}$]},
                xtick={1,2,3},%
                xticklabel style={align=center},
                xticklabels={No attack, {Attack, Longest}, {Attack, Freshest}},
                enlarge x limits=0.3,
                bar width=10mm,
                legend style={
                    at={(0,1)},
                    anchor=north west,
                    legend columns=-1,
                    draw=black,
                    /tikz/every even column/.append style={column sep=1em},
                    cells={align=center},
                    xshift=0.5em,
                    yshift=-0.5em,
                },
                ymin=0.00,
                ymax=0.05,
                yticklabel style={
                    /pgf/number format/fixed,
                    /pgf/number format/precision=2
                },
                scaled y ticks=false,
                height=0.42\linewidth,
                label style={
                    align=center,
                },
            ]

            \addplot+ [
                bar shift=0pt,
                color=myParula05Green,
                fill=myParula05Green!50,
            ] coordinates {
                (1,0.0405555)
            };

            \addplot+ [
                bar shift=0pt,
                color=myParula02Orange,
                fill=myParula02Orange!50,
            ] coordinates {
                (2,0.00519444)
            };

            \addplot+ [
                bar shift=0pt,
                color=myParula01Blue,
                fill=myParula01Blue!50,
            ] coordinates {
                (3,0.0388888)
            };

        \end{axis}
    \end{tikzpicture}%
    \vspace{-0.5em}%
    \caption{The honest chain growth rate in three scenarios: without spamming attack; under attack while downloading the longest header chain first (priority rule in Cardano's block download logic); under attack while downloading the freshest block first (introduced in this work). Details of the experimental setup are given in Section~\ref{sec:spam-demo}. 
    For a 
    trace of the chain growth 
    in the same experiment,
    see Figure~\ref{fig:spam-trace}.}
    \label{fig:networkdelayspamming}
\end{figure}
This protrudes
in another experiment (see Figure~\ref{fig:networkdelayspamming})
where
nodes run PoS LC with
our implementation of Cardano's block
download logic as per \cite{cardano-network}.
Adversarial spamming (through block equivocations) causes significant network traffic at the victim nodes, leaving insufficient bandwidth for the victims to download honest blocks. As a result, block production on the honest chain stalls, and the victim node can be easily fooled by a longer chain from the adversary,
potentially resulting in a safety violation.
\paragraph{Modelling Bandwidth Constraints}
We model a bandwidth constrained network as follows.
Recall that blocks in Nakamoto consensus consist of
a list of transactions as \emph{block content}, and the information pertaining to
the PoS/PoW lottery and the block tree structure (reference to parent block)
as \emph{block header}.
Since a block's header is small compared to its content,
we assume that block headers propagate with a known delay upper bound $\DeltaHeader$
between honest nodes.
At any point after obtaining a block header, a node can request the corresponding
block content from the network.
Since a block's content is large, every honest node can only
download a limited number of blocks' contents per time slot.
This model is inspired by the peer-to-peer network designs
used for blockchain protocols. For instance, in the Cardano network
\cite{cardano-network-protocols,cardano-network-spec},
each node advertises its block header chain to its peers,
which in turn decide based on the block headers
which block contents to fetch.
Without a carefully designed \emph{download rule} for the protocol
to determine which blocks
honest nodes should spend their scarce bandwidth on, 
we will see that
consensus cannot be achieved with PoS LC.

\begin{figure}
    \centering
    \begin{tikzpicture}
        \footnotesize
        
        \begin{scope}[blockchain,x=1cm,y=0.8cm]

            \coordinate (G) at (-1,0);
            
            \node [block-gray] (b0) at (0,0) {};
            \node [anchor=south east,yshift=0.6em] at (b0) {$b_0$};
            
            \node [block-red] (bA1b) at (1,2.0) {};
            \node [anchor=south,yshift=0.6em,xshift=0em] at (bA1b) {$b^{(1)}$};
            \node [anchor=north west,inner sep=0] at (bA1b) {\normalsize\textcolor{myParula07Red}{\xmark{}}};
            \node [block-red] (bA1B) at (2,2.0) {};
            \node (advLabelTop) [anchor=south,yshift=0.6em] at (bA1B) {$B^{(1)}$};
            
            \node [block-red] (bA2b) at (1,1.0) {};
            \node [anchor=south,yshift=0.6em,xshift=0em] at (bA2b) {$b^{(2)}$};
            \node [anchor=north west,inner sep=0] at (bA2b) {\normalsize\textcolor{myParula07Red}{\xmark{}}};
            \node [block-red] (bA2B) at (2,1.0) {};
            \node [anchor=south,yshift=0.6em] at (bA2B) {$B^{(2)}$};
            
            \node (bA3b) at (1,0.5) {$\vdots$};
            \node (bA3B) at (2,0.5) {$\vdots$};
            
            \node [block-green] (bH1) at (3,-0.5) {};
            \node (honLabelTop) [anchor=south,yshift=0.6em,xshift=0em] at (bH1) {$b_1$};
            \node [block-green] (bH2) at (4,-1.5) {};
            \node (honLabelBottom) [anchor=south,yshift=0.6em,xshift=0em] at (bH2) {$b_2$};
            
            \node (bH3) at (4,-2) {$\vdots$};

            \draw [link] (b0) -- (G);
            \draw [link] (bA1B) -- (bA1b);
            \draw [link] (bA1b) -- (b0);
            \draw [link] (bA2B) -- (bA2b);
            \draw [link] (bA2b) -- (b0);
            \draw [link] (bH1) -- (b0);
            \draw [link] (bH2) -- (b0);

            \draw [decorate,decoration={brace,amplitude=4pt}]
                (advLabelTop.north east) -- (bA3B.south east -| advLabelTop.north east)
                node [midway,anchor=west,xshift=4pt,align=left]
                {Adversarial equivocations luring\\ honest nodes away from down-\\ loading honest blocks (\textcolor{myParula07Red}{\xmark{}}: blocks\\ with invalid content)};

            \draw [decorate,decoration={brace,amplitude=4pt}]
                (honLabelTop.north east -| honLabelBottom.north east) -- (bH3.south east -| honLabelBottom.north east)
                node [midway,anchor=west,xshift=4pt,align=left]
                {Honest blocks are\\ never downloaded\\because of shorter\\ header chain};
                
            \node (labelb0) [align=center] at (-0.5,-1.75) {Last valid block\\downloaded by\\all honest nodes};
            \draw [link,dashed,shorten >=0.75em,shorten <=0.75em,black!40] (labelb0) -- (b0);
            
            \begin{scope}%
                \draw [Latex-] (6,3) -- (-1.5,3) node [above right] {\emph{Time slots}};
                
                \draw [] (0,3) ++(0,-0.2) -- ++(0,0.4) node [above] {$t_0$};
                \draw [] (1,3) ++(0,-0.2) -- ++(0,0.4) node [above] {$t$};
                \draw [] (2,3) ++(0,-0.2) -- ++(0,0.4) node [above] {$t'$};
                \draw [] (3,3) ++(0,-0.2) -- ++(0,0.4) node [above] {$t_1$};
                \draw [] (4,3) ++(0,-0.2) -- ++(0,0.4) node [above] {$t_2$};
            \end{scope}

        \end{scope}
    \end{tikzpicture}
    \caption{In PoS LC with `download along the longest header chain' rule,
    an adversary can stall consensus indefinitely if it
    has two consecutive block production opportunities $t < t'$ at which it creates
    infinitely many equivocating chains $b_0 \leftarrow b^{(i)} \leftarrow B^{(i)}$
    where $b^{(i)}$ have invalid content.
    The blocks of later honest block production opportunities
    $... > t_2 > t_1 > t' > t$ are never downloaded
    by other honest nodes, because they prioritize the longer adversarial
    header chains, wasting their bandwidth downloading each $b^{(i)}$ only
    to discard it immediately thereafter because of invalid content.}
    \label{fig:poslclongestheaderchainattack}
\end{figure}
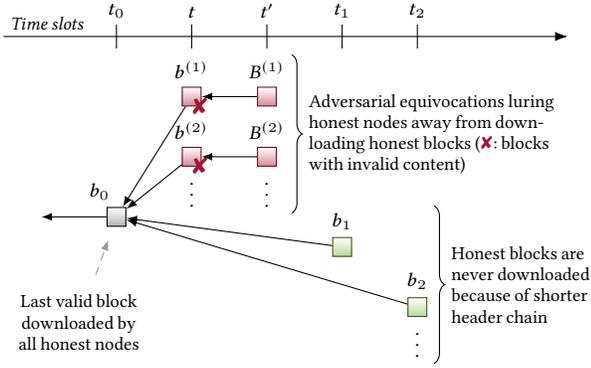

\paragraph{The `Download Along The Longest Header Chain' Rule}
Given that in LC, honest nodes extend the
longest chain,
a natural download rule is `download along the longest header chain',
\ie, based on the block tree structure obtained from block headers,
a node identifies the longest (header) chain, and prioritizes downloading the blocks
along that chain.
Indeed, Bitcoin does exactly that
\cite{btcdevp2pnetworkheadersfirst}.
Cardano's Ouroboros implementation also follows this paradigm in broad strokes
\cite{cardanoforkchoice1,cardano-network-spec,cardano-network-protocols}
for chain selection \cite{cardanoforkchoice2} and block downloads \cite{cardanoforkchoice3}.
As long as the block production rate is low relative to the download bandwidth,
this (and other rules that ensure that nodes download a block at most once) work well for PoW LC,
simply because the 
number of distinct blocks is limited by the number of block production opportunities.

Unfortunately, as illustrated in Figure~\ref{fig:poslclongestheaderchainattack},
this download rule fails for PoS LC in that the resulting protocol is not secure, even if the block production rate is low and the adversary controls a small minority of the stake.
The reason is that the adversary can use consecutive adversarial
block production opportunities (at $t$ and $t'$ in Figure~\ref{fig:poslclongestheaderchainattack})
to produce infinitely many equivocating chains
($b_0 \leftarrow b^{(i)} \leftarrow B^{(i)}$ in Figure~\ref{fig:poslclongestheaderchainattack}).
To avoid honest nodes building on these equivocating chains, the adversary fills
$b^{(i)}$ with invalid content, which honest nodes can only
detect after they have already wasted their scarce bandwidth to download it.
As a result, honest nodes produce blocks off $b_0$
in their block production opportunities
($b_1, b_2, ...$ at $t_1 < t_2 < ...$ in Figure~\ref{fig:poslclongestheaderchainattack}),
but these are never downloaded by other honest nodes because
the adversarial header chains are longer
and thus of higher download priority.
This is clearly an attack on liveness but it also implies an attack on safety because the adversary could now build a chain longer than the honest parties (who are stalled) even though the adversary owns very little stake.
The impact of this attack is seen in our experiments with a PoS LC node implementing this download rule (Figure~\ref{fig:networkdelayspamming}).

The above attack suggests that securing PoS LC under bandwidth constraints requires a carefully designed download rule.
In practice, protocols follow various heuristics
to attempt to mitigate
spamming/equivocation attacks.
However, a rigorous analysis is usually missing.
Our goal in this work is to identify simple download rules that can be proven secure in the bandwidth constrained model.

In the attack
in Figure~\ref{fig:poslclongestheaderchainattack},
we observe that even though new honest blocks are being proposed, 
the download rule prioritizes older adversarial equivocating blocks. 
If honest nodes downloaded the `fresher' blocks proposed in more recent time slots $t_1, t_2, ...$ instead, then this attack would not succeed. 
This intuition extends beyond the specific attack of Figure~\ref{fig:poslclongestheaderchainattack}.
We would like that
whenever an honest node proposes a block, other honest nodes download that block and its prefix `soon'.
This way,
honest nodes have a chance to produce blocks extending it,
and to align their block production efforts toward a particular chain.
This is arguably the key to LC security and central to prior
security analysis \cite{dem20,sleepy}
on which we build.
This insight naturally motivates the following simple download rule.

\paragraph{The `Download Towards The Freshest Block' Rule}
We propose a simple download rule for PoS LC,
`download towards the freshest block',
\ie, in every time slot an honest node identifies the block proposed in the most
recent time slot based on the header information,
and downloads any missing blocks in its prefix, including
that freshest block.
Thus,
when an honest node proposes a block, within the same time slot, other honest nodes prioritize downloading that block and its prefix.
The length of the prefix cannot be too long since valid chains cannot contain equivocations.
By making the time slot long enough to allow downloading the whole prefix, this rule directly satisfies our desired property that
honest nodes download honestly proposed blocks `soon' (within the same time slot).
This property is the key step in prior security analysis, thus allowing us to use prior techniques to prove the security of PoS LC with this download rule.
In particular, this download rule
avoids the attack of Figure~\ref{fig:poslclongestheaderchainattack}
and the honest chain's growth rate remains unaffected by this spamming attack (Figure~\ref{fig:networkdelayspamming}).
Importantly, note that the freshest block rule is only a download priority rule. Honest nodes still propose blocks extending their longest valid downloaded chain.

\paragraph{Other Download Rules}
More generally, we identify other download rules with the property that whenever an honest block is proposed, all honest nodes download that block and its prefix `soon'. Thus, we develop a unified framework to prove security of PoS LC with any download rule with this property. We consider the following two commonly proposed heuristics against equivocations, formalize them and give a rigorous proof of security.
\begin{enumerate}
    \item `Equivocation avoidance': We modify `download along the longest header chain' such that an honest node refrains from downloading a chain whose tip is an equivocating block header (\ie, it has seen another block header from the same slot and validator).
    A rule of this kind can be seen used in PoS Ethereum \cite{eth2-spec-p2p}.
    \item `Blocklisting': An honest node avoids downloading any chain whose tip is proposed by a validator that has equivocated before (in its view of block headers). Note that this notion of blocklisting only affects the download priority rule. It does not invalidate a block, as doing so independently at each node risks introducing split views, and doing so consistently would require consensus 
    in the first place.
\end{enumerate}
Due to the simplicity and efficiency of the `download towards the freshest block' rule, and because it directly satisfies the key property that enables our security proofs, we use this rule as a running example to illustrate our model and the analysis. We then extend this analysis to the other two rules.

\paragraph{Our Contributions}
By means of experiments and a concrete attack strategy, we show that the bounded delay model fails to capture network congestion and spamming attacks.
We show that using suitable download rules, we can provably secure PoS LC 
in networks with bandwidth constraint in which
the adversary can 
(\emph{inter alia})
spam the network with equivocating blocks
at an arbitrary rate,
withhold blocks,
and
release blocks with invalid content
that honest nodes discard after downloading.
We identify a key property of a download rule that enables it to secure PoS LC. 
We use this to develop a unified framework to prove the security of PoS LC with any download rule that satisfies this property.
We propose a simple rule `download towards the freshest block' that satisfies this property. We also formalize heuristics in the form of the `equivocation avoidance' and `blocklisting' rules for which we provide a rigorous security proof using our 
framework.
We show that parallel composition of multiple instances of PoS LC with a secure download rule
(inspired by \cite{near-optimal-thruput})
yields a consensus protocol that achieves a constant fraction of the
network's throughput limit even in the worst case.

\paragraph{Related Work}
Network-level attacks on Bitcoin have been studied in \cite{revisit-network-level-attacks,bitcoin_hijack}.
Eclipse attacks on peer-to-peer networks, where an adversary uses several IP addresses to occupy all connections maintained by a victim node
and thus cut said node off from the network, have been studied in \cite{eclipse,eclipse2,eclipse3,eclipse4} and in the context of Bitcoin in \cite{eclipse1}.
The authors of \cite{bftgossip} show that if one can connect with consensus validators that are pseudo-randomly chosen every few slots based on their stake,
then one can secure PoS LC against Sybil attacks and eclipse attacks on the network layer.
These earlier works 
share their focus on network topology, an important aspect not captured
by the bounded delay network model.
Our work instead focuses on bandwidth constraints,
an orthogonal feature of real networks
not captured by the bounded delay model.
However, our works share the philosophy of co-designing consensus and network layer protocols.

The impact of spamming was seen recently in the temporary shutdown of a PoS protocol Solana \cite{solana} on multiple occasions in 2021-2022 \cite{solanastop,solanastop2,solanastopagain}.
These shutdowns were 
reportedly 
due to an increase in the transaction load in the network, 
and the ``lack of prioritization of network-critical messaging caused the network to start forking'' \cite{solanastop}.
These incidents indicate that
messages that are critical for consensus among honest nodes (\eg, blocks)
must be appropriately prioritized during periods of congestion.
Consensus-critical blocks are easily
prioritized at the network level over
less critical
transaction requests,
as the two are different kinds.
Thus, this work focuses instead on 
the
design of a download rule with which
the consensus protocol
assists the network
in prioritizing
consensus-critical blocks
over 
similarly looking
spam blocks.

In practice, implementations show awareness of
and attempt to mitigate
equivocation-based spamming attacks using
various heuristics.
However, their efficacy and side effects
are often not fully understood.
For instance, Cardano's Ouroboros implementation
disconnects from peers once they propagate invalid or equivocating blocks
\cite{cardanoforkchoice1,cardano-network-spec,cardano-network-protocols}.
However,
an adversary can boost the impact of its attack 
by creating more Sybil network peers
(recall that there is no relation between consensus validators
and peers in the underlying communication network),
so that disconnected peers are likely replaced by
new adversarial peers, ready to waste more of the honest node's resources
\cite{eclipse2,eclipse3,eclipse4}.

Slashing is routinely proposed as a solution to mitigate spamming with equivocations, as such attacks can be attributed to specific validators \cite{casper,weaksubjectivity,forensics,availability-accountability}. 
Typical crypto-economic guarantees are of the form ``if human intervention is needed to recover from a safety attack, then 33\% of stake is slashable`` \cite{gasper,forensics,availability-accountability}.
However, the attack in Figure~\ref{fig:poslclongestheaderchainattack} only requires two consecutive block production opportunities, which can be obtained by an adversary with a very small fraction of stake. 
Hence in this case, slashing would impose a very small penalty for an attack that violates security and potentially incurs large costs due to human intervention and other losses.
Instead, we take the approach of preventing attacks in the first place by using download rules that are proven secure. Once security is proven, slashing can be employed as an additional measure to disincentivize equivocation-based spamming.

The need for careful modelling of bandwidth constraints in the context of high throughput protocols was identified in \cite{prism,near-optimal-thruput}.
Earlier works \cite{prism,ghost,btcnetworkdelay} note that the network delay 
increases with the 
message size
(\ie, block size in this case).
In this model, it is assumed that as long as the network load is less than the bandwidth,
every message is downloaded within a given delay bound which depends on the message size but is independent of total network load. 

In the PoS context, \cite{near-optimal-thruput} captures congestion due to increased network load by modelling the inbox of each node as a queue. Each message undergoes a propagation delay before being added to the recipient's inbox queue. The recipient can retrieve messages from their queue at a rate limited by their bandwidth,
resulting in a queuing delay.
However, the security result \cite[Theorem~1]{near-optimal-thruput}
still assumes a bounded (propagation+queuing) delay.
This assumption is only shown to hold under honest executions when the adversary does not corrupt any nodes and does not send or delay any messages \cite[Theorem~3]{near-optimal-thruput}, and therefore the security claim does not hold for all adversarial strategies.
In particular, this excludes adversaries that can spam the network using equivocating blocks and cause attacks such as in Figure~\ref{fig:poslclongestheaderchainattack}.
The model we use is a variant of that in \cite{near-optimal-thruput} with the difference that nodes can inspect a small segment (block header) at the beginning of every message in their queue and decide based on that which message (block content) to prioritize for download (subject to the bandwidth constraint). This modification allows us to prove security against a general adversary, even with unbounded equivocations.

Although our work is the first to prove PoS LC secure under bandwidth constraints, our analysis builds on tools from several years of security analysis for LC protocols
\cite{backbone,pss16,ren,sleepy,david2018ouroboros,pos_paper,dem20,tight_bitcoin}, 
particularly the concept of pivots \cite{sleepy}
(\cf Nakamoto blocks \cite{dem20}).

\paragraph{Outline}
We state the PoS LC protocol augmented with a download rule and introduce our formal
model for bandwidth constrained networks in Section~\ref{sec:modelprotocol}.
In Section~\ref{sec:securityargument}, we provide a high-level description of our unified framework for proving security of PoS LC with different download rules 
under bandwidth constraints.
In Section~\ref{sec:analysis}, we show the key steps towards this unified proof, and analyze the `download towards the freshest block' rule.
We present experimental evidence for the robustness
and superior performance
of
the `freshest block' rule
in Section~\ref{sec:experiments}.
We formalize and analyze the `equivocation avoidance' and `blocklisting' rules in Section~\ref{sec:other_download_rules}.
Finally, we sketch in Section~\ref{sec:parallel}
how to use PoS LC with a suitable download rule as a building block
to obtain a consensus protocol with a constant fraction of the network's throughput limit in the worst case.

%% file: 23_modelprotocol.tex
\section{Protocol and Model}
\label{sec:modelprotocol}

\begin{algorithm}[t]
    \caption{Idealized PoS LC consensus protocol $\protocol$ with a download rule
    (helper functions: Appendix~\ref{sec:appendix-helperfunctions-pseudocode}, $\Ftree$: Algorithm~\ref{algo:pseudocode-Ftree}, $\Env$: Appendix~\ref{sec:appendix-environment})}
    \label{algo:pseudocode-pos-lc}
    \begin{algorithmic}[1]
        \small
        \On{$\Call{init}{\mathsf{genesisHeaderChain}, \mathsf{genesisTxs}}$}
        \LineComment{Initialize header tree $\HeaderTree$, longest downloaded chain $\DownloadedChain$, and mapping from block headers
        to contents (lists of transactions) $\TxsMap$}
            \State $\HeaderTree, \DownloadedChain \gets \{\mathsf{genesisHeaderChain}\}, \mathsf{genesisHeaderChain}$
            \State $\TxsMap[\DownloadedChain] \gets \mathsf{genesisTxs}$
            \Comment{Unset entries of $\TxsMap$ are $\NIL$}
        \EndOn
        \On{$\Call{receivedHeaderChain}{\Chain}$} \Comment{Called by $\Env$ or $\Adv$}
            \label{loc:pseudocode-pos-lc-receiveheader}
            \State \Assert{$\Ftree.\Call{verify}{\Chain}$}
            \Comment{Validate header chain (Algorithm~\ref{algo:pseudocode-Ftree})}
            \State $\HeaderTree \gets \HeaderTree \cup \operatorname{prefixChainsOf}(\Chain)$ \Comment{Add $\Chain$ and its prefixes to $\HeaderTree$}
            \State $\Env.\Call{broadcastHeaderChain}{\Chain}$
        \EndOn
        \On{$\Call{receivedContent}{\Chain, \txs}$} \Comment{Called by $\Env$ or $\Adv$}
            \label{loc:pseudocode-pos-lc-receivecontent}
            \LineComment{Defer processing the content until we received the corresponding header chain $C$, and its prefixes' 
            contents
            are downloaded and valid}
            \State \textbf{defer until} $\Chain \in \HeaderTree$
            \State \textbf{defer until} $\forall \Chain' \prec \Chain\colon \TxsMap[\Chain'] \not\in \{\NIL, \INVALID\}$
            \State \Assert{$\Chain.\mathsf{txsHash} = \operatorname{Hash}(\txs)$}
            \If{$\operatorname{txsAreSemanticallyValidWrtPrefixesOf}(\Chain, \txs)$}
                \State $\TxsMap[\Chain] \gets \txs$
                \State $\Env.\Call{uploadContent}{\Chain, \txs}$
            \Else
                \State $\TxsMap[\Chain] \gets \INVALID$
            \EndIf
            \LineComment{Update the longest downloaded chain among downloaded valid chains}
            \State $\Tree' \gets \HeaderTree \setminus \{ \Chain' \in \HeaderTree \mid \TxsMap[\Chain'] \in \{\NIL,\INVALID\} \}$
            \State $\DownloadedChain \gets \argmax_{\Chain \in \Tree'} \len{\Chain}$
            
        \EndOn
        \On{$\Call{scheduleContentDownload}{\null}$}
            \label{loc:pseudocode-pos-lc-downloadrule}
            \LineComment{Pick next block to download according to download rule (\cf Algs.~\ref{algo:freshest-block-rule}, \ref{algo:longest-header-chain-equivoc-avoid-rule})}
            \If{$\Chain \neq \bot$ \textbf{with} $\Chain \gets \operatorname{downloadRule}(\HeaderTree,\TxsMap)$} 
                \State $\Env.\Call{requestContent}{\Chain}$
            \EndIf
            \LineComment{$\Call{receivedContent}{}$ will be triggered by $\Env$ on successful download}
        \EndOn

        \For{time slots $t \gets 1,...,\Th$ of duration $\tau$}
            \Comment{PoS LC protocol main loop}
            \label{loc:pseudocode-pos-lc-mainloop}
            \State $\txs \gets \Env.\Call{receivePendingTxsSemanticallyValidWrt}{\DownloadedChain}$
            \LineComment{Produce and disseminate a new block if eligible, see Alg.~\ref{algo:pseudocode-Ftree}}
            \If{$\Chain' \neq \bot$ \textbf{with} $\Chain' \gets \Ftree.\Call{extend}{t,\DownloadedChain,\txs}$}
                \State $\Env.\Call{uploadContent}{\Chain', \txs}$
                \State $\Env.\Call{broadcastHeaderChain}{\Chain'}$
            \EndIf
            \LineComment{Download block contents (starting after $\DeltaHeader$ time into the slot)}
            \While{end of current time slot $t$ not reached}
                \State $\Call{scheduleContentDownload}{ }$
            \EndWhile
            \State $\Env.\Call{outputLedger}{\DownloadedChain^{\lceil \Tconf}}$
            \Comment{Ledger of node $i$ at time $t$: $\LOG{i}{t}$}
        \EndFor
    \end{algorithmic}
\end{algorithm}

\begin{algorithm}[t]
    \caption{`Freshest block' download rule}
    \label{algo:freshest-block-rule}
    \begin{algorithmic}[1]
        \small
        \RealFunction{$\operatorname{downloadFreshestBlock}(\HeaderTree, \TxsMap)$}
            \label{loc:freshest-block-rule-lines}
            \LineComment{Ignore chains with invalid content in any block}
            \State $\Tree \gets \{ \Chain \in \HeaderTree \mid \forall \Chain' \preceq \Chain\colon \TxsMap[\Chain'] \neq \INVALID \}$
            \LineComment{Find the chain ending in the freshest block (\ie, from most recent slot)}
            \State $\Chain \gets \argmax_{\Chain' \in \Tree} \Chain'.\mathsf{time}$
            \LineComment{Find the first not downloaded block on that chain
            (if non-existent: $\bot$)}
            \label{loc:pseudocode-pos-lc-empty-chain}
            \State $\Chain \gets \argmin_{\Chain' \preceq \Chain\colon \TxsMap[\Chain'] = \NIL} |\Chain'|$
            \State \Return{$\Chain$}
        \EndRealFunction
    \end{algorithmic}
\end{algorithm}

\begin{algorithm}[t]
    \caption{Idealized functionality $\Ftree$: block production lottery and header block chain structure
    (\cf Appendix~\ref{sec:appendix-helperfunctions-pseudocode})}
    \label{algo:pseudocode-Ftree}
    \begin{algorithmic}[1]
        \small
        \On{$\Call{init}{\mathsf{genesisHeaderChain},\mathsf{numNodes}}$}
            \State $N \gets \mathsf{numNodes}$
            \State $\Tree \gets \{\mathsf{genesisHeaderChain}\}$
            \Comment{Global set of valid header chains}
        \EndOn
        \On{$\Call{isLeader}{P,t}$ \textbf{from} $\mathcal A$ (only for adversarial $P$) or $\Ftree$}
        \LineComment{Abstraction of proof-of-stake lottery: each node is chosen leader in each slot with probability $\rho/N$ independent of other nodes and slots}
            \If{$\Lottery[P,t] = \bot$}
                \label{loc:rho_definition_in_alg}
                \State $\Lottery[P,t] \overset{\$}{\gets} (\TRUE \text{ with probability $\rho/N$, else } \FALSE)$
                \label{loc:pseudocode-Ftree-blockproductionlottery}
            \EndIf
            \State \Return{$\Lottery[P,t]$}
        \EndOn
        \On{$\Call{extend}{t',\Chain,\txs}$ \textbf{from} node $P$ (possibly adversarial) \textbf{at} time slot $t$}
            \label{loc:pseudocode-Ftree-binding}
            \LineComment{New header chain is valid if parent chain $\Chain$ is valid, $P$ is leader for slot $t'$, and $t'$ is later than the tip of $\Chain$ and is not in the future}
            \If{$(\Chain\in\Tree) \land \Call{isLeader}{P,t'} \land (\Chain.\mathsf{time} < t' \leq t)$} 
            \LineComment{Produce a new block header extending $\Chain$}
            \State $\Chain' \gets \Chain\|\operatorname{newBlock}(\mathsf{time}\colon t',\mathsf{node}\colon P,\mathsf{txsHash}\colon \operatorname{Hash}(\txs))$ 
                \State $\Tree \gets \Tree \cup \{\Chain'\}$ \Comment{Register  new header chain in header tree}
                \State \Return{$\Chain'$}
            \EndIf
            \State \Return{$\bot$}
        \EndOn
        \On{$\Call{verify}{\Chain}$}
            \State \Return{$\Chain \in \Tree$} \Comment{Header chain is valid if 
            previously added to header tree}
        \EndOn
    \end{algorithmic}
\end{algorithm}

\begin{figure}
    \centering
    \begin{tikzpicture}[x=1.8cm]
        \footnotesize
        
        \node [align=left,anchor=west] at (-2.5,0) {\textsc{Participants}};
        \node [align=left,anchor=west] at (-2.5,-2) {\textsc{Content}\\\textsc{Repository}};
        \node [align=left,anchor=west] at (-2.5,1) {\textsc{Header}\\\textsc{Broadcasting}};
        
        \node [inner sep=0] (u1) at (-1,0) {\includegraphics[width=0.75cm]{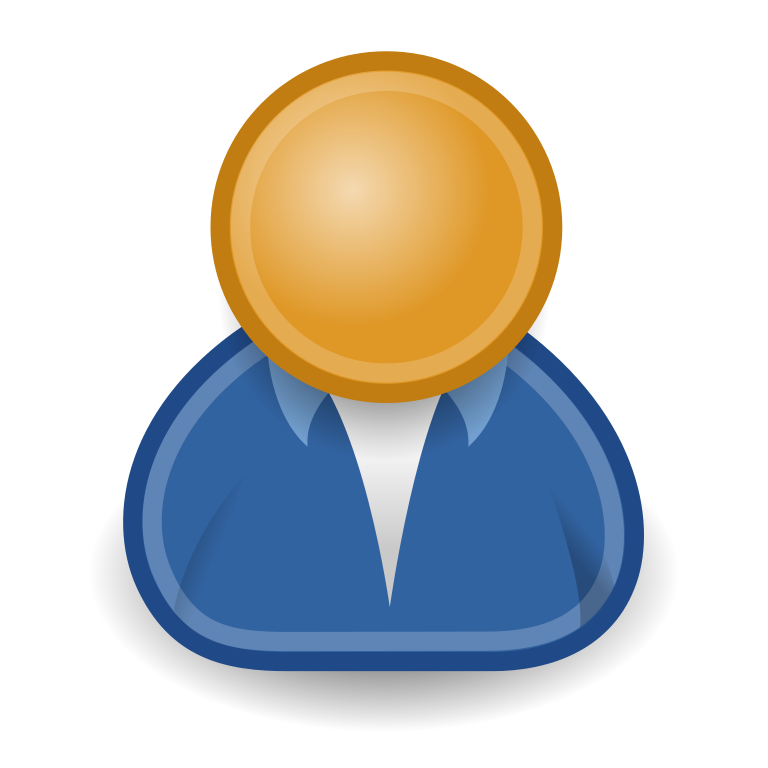}};
        \node [inner sep=0] (u2) at (0,0) {\includegraphics[width=0.75cm]{figures/icon-honest.png}};
        \node [inner sep=0] (u3) at (1,0) {\includegraphics[width=0.75cm]{figures/icon-honest.png}};
        \node [inner sep=0] (a) at (2,0) {\includegraphics[width=0.75cm]{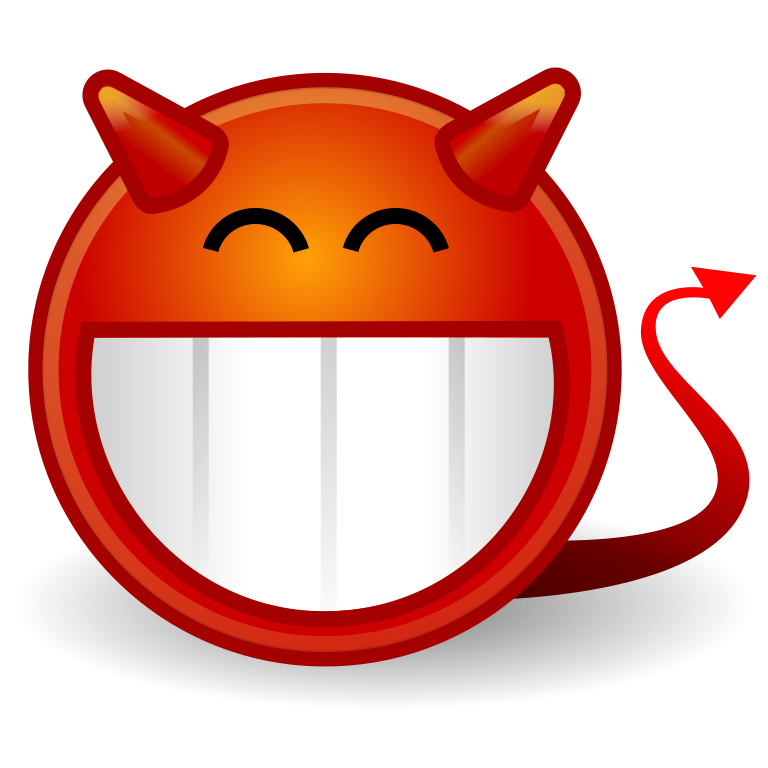}};
        
        \def\TMPBLOCKCONTENT{\node [anchor=center] at (0.5,-0.65) {\tiny$\txs_2$};}
        
        \node (blkNew) at (-1.3,-0.7) {\tikz[x=10pt,y=10pt]{ \draw [fill=gray!80] (0,0) rectangle (1,-0.3); \draw [fill=gray!20] (0,-0.3) rectangle (1,-1); \TMPBLOCKCONTENT }};
        \node (blkNewSplit) at (-0.7,-0.7) {\tikz[x=10pt,y=10pt]{ \draw [fill=gray!80] (0,0.2) rectangle (1,-0.1); \draw [fill=gray!20] (0,-0.3) rectangle (1,-1); \TMPBLOCKCONTENT }};
        
        \draw [-latex] (blkNew) -- (blkNewSplit);
        \draw [-latex] (blkNewSplit) |- (0,1) node [above] {Delay $\leq\DeltaHeader$} -- (u2) node [midway,right] {\tikz[x=10pt,y=10pt]{ \draw [fill=gray!80] (0,0) rectangle (1,-0.3); }};
        \draw [-latex] (blkNewSplit) |- (1,1) -- (u3) node [midway,left] {\tikz[x=10pt,y=10pt]{ \draw [fill=gray!80] (0,0) rectangle (1,-0.3); }};
        
        \draw [-latex,myParula07Red] [bend right] (a) to (u3);
        \node [yshift=1.5em,anchor=south] at (1.5,0) {\tikz[x=10pt,y=10pt]{ \draw [fill=gray!80] (0,0) rectangle (1,-0.3); }};
        \draw [-latex,myParula07Red] [bend left] (a) to (u3);
        \node [yshift=-1.5em,anchor=north] at (1.5,0) {\tikz[x=10pt,y=10pt]{ \draw [fill=gray!20] (0,-0.3) rectangle (1,-1); }};
        
        \begin{scope}[yshift=-2em]
            \node (repoHeader) at (0.25,-1.25) {\scriptsize\textbf{Header}};
            \node (repoContent) at (0.75,-1.25) {\scriptsize\textbf{Content}};
            \node at (0.25,-1.5) {\scriptsize$\Chain_1$};
            \node at (0.75,-1.5) {\scriptsize$\txs_1$};
            \node at (0.25,-1.75) {\scriptsize$\Chain_2$};
            \node at (0.75,-1.75) {\scriptsize$\txs_2$};
            \node at (0.25,-2) {\scriptsize$\vdots$};
            \node (repoDots) at (0.75,-2) {\scriptsize$\vdots$};
        \end{scope}
        
        \draw [-latex] (blkNewSplit) -- (repoHeader) node [midway,below left,xshift=2pt,yshift=2pt] {\tikz[x=10pt,y=10pt]{ \draw [fill=gray!20] (0,-0.3) rectangle (1,-1); \TMPBLOCKCONTENT }};
        
        \draw [-latex] ([xshift=-2pt]u2.south) -- ([xshift=-2pt]repoHeader.north) node [pos=0.3,left] {\scriptsize Req. $\Chain_2$};
        \draw [-latex] ([xshift=2pt]repoHeader.north) -- ([xshift=2pt]u2.south) node [pos=0.7,above right,xshift=-2pt,yshift=-2pt] {\tikz[x=10pt,y=10pt]{ \draw [fill=gray!20] (0,-0.3) rectangle (1,-1); \TMPBLOCKCONTENT }};
        
        \draw [-latex] ([xshift=-2pt]u3.south) -- ([xshift=-2pt]repoContent.north) node [pos=0.3,left] {\scriptsize Req. $\Chain_2$};
        \draw [-latex] ([xshift=2pt]repoContent.north) -- ([xshift=2pt]u3.south) node [pos=0.7,below right,xshift=-2pt,yshift=2pt] {\tikz[x=10pt,y=10pt]{ \draw [fill=gray!20] (0,-0.3) rectangle (1,-1); \TMPBLOCKCONTENT }};
        
        \draw [myParula02Orange] (repoHeader.north) ++ (45:0.25) arc(45:135:0.25);
        \draw [myParula02Orange] (repoContent.north) ++ (45:0.25) arc(45:135:0.25);
        \node [myParula02Orange] at (0.5,-1.15) {\tiny Bandwidth};
        
        \draw ($(repoHeader.north west)+(0,-0.05)$) rectangle (repoDots.south east -| repoContent.north east);
        
        \node [] at (-1,-1.0) {\circledScriptGray{1}};
        \node [] at (-0.5,-1.8) {\circledScriptGray{3}};
        \node [] at (0.5,1.3) {\circledScriptGray{2}};
        \node [] at (0.60,-0.3) {\circledScriptGray{4}};
        \node [] at (1.5,0) {\circledScriptGray{5}};
        
    \end{tikzpicture}
    \caption[]{In our model, block headers are propagated with a known delay upper bound $\DeltaHeader$,
    while block content is subject to a bandwidth constraint.
    \circledScriptGraySlim{1} An honest node produces a new valid block, consisting of header and content.
    \circledScriptGraySlim{2} Block headers are broadcast ($\Env.\Call{broadcastHeaderChain}{}$) and arrive
    at honest nodes ($\protocol.\Call{receivedHeaderChain}{}$) within at most $\DeltaHeader$ delay.
    \circledScriptGraySlim{3} Block content is submitted to an idealized `repository' ($\Env.\Call{uploadContent}{}$). A hash of the corresponding block content is included in the block header. 
    \circledScriptGraySlim{4} Upon request ($\Env.\Call{requestContent}{}$), the content of a certain block is obtained
    from the `repository' ($\protocol.\Call{receivedContent}{}$),
    subject to a constraint on the rate of downloaded block contents.
    \circledScriptGraySlim{5} An adversary can push block headers and block content to honest nodes
    independent of delay and bandwidth constraints. 
    See Appendix~\ref{sec:appendix-environment} for details on $\Env$.}
    \label{fig:model-network-illustration}
\end{figure}

\paragraph{Model Main Features}
For ease of exposition,
we assume a static set
of $N$ active \emph{nodes},
each with a cryptographic identity corresponding to one unit of stake.
Our analysis can be 
extended to the 
setting
of heterogeneous and dynamic stake using tools from \cite{david2018ouroboros,snowwhite}.
Nodes' cryptographic identities are common knowledge.
We are interested in the large system regime $N\to\infty$.
A \emph{static} adversary $\Adv$ chooses a set of nodes
(up to a fraction $\beta$ of all nodes,
where $\beta$ is common knowledge) to corrupt
before the randomness of the protocol is drawn
and the execution commences.
Uncorrupted \emph{honest} nodes follow the protocol as specified
at all times,
corrupted \emph{adversarial} nodes deviate from the protocol
in an arbitrary \emph{Byzantine} manner coordinated by the adversary
in an attempt to inhibit consensus.
For simplicity, we have assumed that all nodes are always \emph{awake}.
Our analysis builds on techniques from 
\cite{sleepy} and the refined machinery therein
can be used to extend our analysis to the setting of
asleep/awake honest nodes.

\paragraph{Protocol Main Features}
Pseudocode of an idealized PoS LC Nakamoto consensus protocol
parameterized by
a download rule
is provided in Algorithm~\ref{algo:pseudocode-pos-lc}
(\cf~\cite[Figure~3]{sleepy}).
The `download towards the freshest block' rule is given in Algorithm~\ref{algo:freshest-block-rule}.
Implementation details of the block production lottery
and the handling of the blockchain data structure
are abstracted away
in the idealized functionality $\Ftree$
provided in Algorithm~\ref{algo:pseudocode-Ftree}
(\cf~\cite[Figure~2]{sleepy}).
An index of the helper functions used in the pseudocode
is provided in Appendix~\ref{sec:appendix-helperfunctions-pseudocode}.
With specific implementations of $\Ftree$,
a variety of PoS LC protocols can be modelled
such as protocols from the Ouroboros family
\cite{kiayias2017ouroboros,david2018ouroboros,badertscher2018ouroboros}
and the Sleepy Consensus \cite{sleepy,snowwhite} family.
A more formal description
of the environment $\Env$
(idealized functionality modeling the 
network) is given in 
Appendix~\ref{sec:appendix-helperfunctions-pseudocode}.
In the main loop of the PoS LC protocol
(Algorithm~\ref{algo:pseudocode-pos-lc},
lines~\ref{loc:pseudocode-pos-lc-mainloop}ff.)
the node attempts in every time slot
(which is of duration $\tau$)
to produce a new block containing
transactions $\txs$
and to extend the longest downloaded chain (denoted $\DownloadedChain$)
in the node's local view.
If successful, the block content $\txs$
and the resulting new block header chain $\Chain'$
are provided to the environment $\Env$
for dissemination to all nodes.

\paragraph{Dissemination of Block Headers and Contents}
The network model and dissemination of headers and contents is
illustrated in Figure~\ref{fig:model-network-illustration}.
Block header chains broadcast via
$\Env.\Call{broadcastHeaderChain}{}$
are delivered by the environment $\Env$
to every node with a delay determined by $\Adv$,
up to a delay upper bound $\DeltaHeader$ that is common knowledge.
Once an honest node receives a new valid header chain
(Algorithm~\ref{algo:pseudocode-pos-lc},
lines~\ref{loc:pseudocode-pos-lc-receiveheader}ff.),
the node adds it to its local
header tree
$\HeaderTree$.
Block content uploaded via
$\Env.\Call{uploadContent}{}$
is kept by $\Env$ in an idealized
repository. 
In every time slot, honest nodes use a download rule to select block headers for which they wish to request the content (Algorithm~\ref{algo:pseudocode-pos-lc}, line~\ref{loc:pseudocode-pos-lc-downloadrule}).
Honest nodes can request the content
for a particular header via
$\Env.\Call{requestContent}{}$.
If available, the content requested from the repository will be
delivered by $\Env$ to the honest node
by triggering the callback $\Call{receivedContent}{}$
(Algorithm~\ref{algo:pseudocode-pos-lc},
lines~\ref{loc:pseudocode-pos-lc-receivecontent}ff.).
We set the slot duration as 
$\tau = \DeltaHeader + \frac{\bwslot}{\bwtime}$ such that 
all honest nodes receive block headers proposed at the start of the current slot, and thereafter
$\Env$ delivers at most $\bwslot$ block contents
requested from the repository
to each honest node per time slot, thereby constraining the bandwidth to $\bwtime$ blocks per second.%
\footnote{Unlike \cite{near-optimal-thruput}, we do not model the upload bandwidth because honest nodes only send very few messages in our protocol.}
Upon verifying that the content matches the hash
in the block header and that the $\txs$ are
valid with respect to the ledger determined
by the block's prefix, the node adds $\txs$
to its local view. Otherwise, the block is marked
as $\INVALID$, to prevent downloading it or any
of its descendants in the future.
Finally, the node updates its longest downloaded chain.

\paragraph{`Download Towards The Freshest Block' Rule}
Motivated by the earlier arguments in
Section~\ref{sec:introduction},
we introduce the
`download towards the freshest block' download rule
(Algorithm~\ref{algo:freshest-block-rule}).
In this download rule, first the header tree $\HeaderTree$
is pruned by invalid blocks and their descendants.
Then, the first $\NIL$ block in the prefix of the
freshest block is requested.
Ties are broken by the adversary.

\paragraph{Adversarial Strategies And Powers}
Adversarial strategies and powers include
but are not limited to:
reusing block production opportunities
to produce multiple blocks (\emph{equivocations},
by calling $\Ftree.\Call{extend}{}$ multiple times,
each with a different $\txs$ or a different parent chain $\Chain$);
extending any chain using past block production opportunities
as long as the purported block production time slots
along any chain
are strictly increasing;
releasing block headers late or selectively to honest nodes;
proactively pushing block headers or block content to honest nodes irrespective of delay or bandwidth constraints (by triggering the node's respective
$\Call{receivedHeaderChain}{}$ or $\Call{receivedContent}{}$ callback);
withholding the content of blocks;
including invalid $\txs$ in blocks;
breaking ties in chain selection and the download rule.

\paragraph{Reality Check}
Note that in practice the prioritization of blocks 
according to some download rule
does not have to
take place only at the endpoints of the network or be limited to block content.
Rather, 
it
can also be applied to block headers
and by intermediary nodes of the underlying communication
or peer-to-peer gossip overlay network
as they forward blocks.
This effectively shifts 
the download rule
from the edge into the network.
Honest participants focus their resources on
what the scheduling logic determines as `high importance' traffic,
and 
save it from being drowned out by adversarial spam.
The result is that
headers of
the blocks
which might be of interest to an honest node based on the prioritization
stipulated by the download rule
will be made available to that honest node
by the network within reasonable delay despite adversarial interference.
Because of this, we believe that our model
leads to protocols that can fare well
under bandwidth constraints and spamming in practice.

Various constructions are used to realize $\Ftree$ in 
real-world protocols,
depending on the desired properties.
The block production lottery
(Algorithm~\ref{algo:pseudocode-Ftree},
line~\ref{loc:pseudocode-Ftree-blockproductionlottery})
is typically implemented by checking whether the output
of a random function is below a certain threshold.
Against static adversaries, a collision resistant hash function suffices \cite{sleepy};
against adaptive adversaries, a verifiable random function (VRF) is used \cite{kiayias2017ouroboros}.
Although the ideal functionality $\Ftree$ relies on the knowledge of $N$ to tune the threshold $\rho/N$, in PoS realizations such as in \cite{david2018ouroboros} the factor $1/N$ is replaced by the fraction of the total stake owned by the node as per the confirmed segment of the blockchain.%
\footnote{
In our simplified model, each node owns one unit of stake, which is the same as $1/N$ fraction of the total stake where $N$ is the number of nodes.
}
The binding between a block and the production opportunity
it stems from
(Algorithm~\ref{algo:pseudocode-Ftree},
line~\ref{loc:pseudocode-Ftree-binding})
is established using digital signatures.

The idealized repository maintained by $\Env$ is just a way to abstract block dissemination in 
a peer-to-peer gossip network. In reality, each node
requests their peers for the block content
(using information from the block header),
and honest peers respond with the content.
Correspondingly,
the idealized repository indexes block content by the block header, and delivers it upon request, if available.
Note that block content
associated with a particular block header
may be unavailable when requested by some honest node
at one point in time (\eg, if the adversary
did not make it available),
but available when requested by another honest node
at a later time (\eg, if the adversary made it
available in the meantime).
Thus, the block header does not ensure data availability
or consistency among honest nodes' download
attempts,
unlike in stronger primitives such as
verifiable information dispersal \cite{avid,avidfp,dispersedledger,semiavidpr}.
By modelling the network as an idealized repository, we abstract out details such as the network topology and data availability that are orthogonal to the issue being considered here: the bandwidth constraint.

%% file: 4_analysis_high_level.tex
\section{High Level Security Argument}
\label{sec:securityargument}

Our proofs follow the techniques of \cite{sleepy}
and \cite{dem20}.
The key difference between these techniques and our proof is that the former assume that the block propagation delay is always bounded by a constant $\Delta$. In our case, we first prove that
under `suitable' download rules and protocol parameters,
with overwhelming probability,
a large fraction of honestly proposed blocks are downloaded by all honest nodes within bounded delay.

To this effect, we consider \emph{uniquely successful} time slots, in which there is exactly one honest block proposal (any other slots with block proposals are called \emph{adversarial}). 
For a given download rule and protocol execution, 
we define a property $\grpromempty$ (shorthand for Maximum Download) by which
under all adversarial strategies
and throughout the execution,
the block proposed in a uniquely successful slot is downloaded 
by all honest nodes within 
the first $\bwslot$ blocks downloaded since the beginning of that slot (Definition~\ref{def:grprom-definition}).
If the time slot is long enough to allow downloading $\bwslot$ blocks, then the block proposed in
any uniquely successful slot will be downloaded by all honest nodes before the
end of the same time slot.
Then, each block proposed in a uniquely successful slot 
increases the minimum length of all honest nodes' longest downloaded chains.
(Lemma~\ref{prop:chain_growth}).
This is the key stepping stone of earlier security proofs of LC \cite{sleepy,dem20}.
We then employ techniques of \cite{sleepy} to prove that 
PoS LC 
with the right parameters
and a download rule such that $\grpromempty$ holds with overwhelming probability
is secure
(Theorem~\ref{thm:security}).
This gives a general framework where 
in order to prove security of PoS LC with a specific download rule, one only needs to show that the download rule satisfies $\grpromempty$ with overwhelming probability.

The property $\grpromempty$ suggests a natural download rule.
In a uniquely successful slot, the block proposed in that slot can be identified as the unique freshest block. Thus, downloading towards the freshest block allows an honest node to download the block proposed in that slot most straightforwardly.
If the prefix of the freshest block contains less than $\bwslot$ blocks that have not been downloaded yet, 
then $\grpromempty$ will be satisfied.
Thus, for a suitably chosen $\bwslot$,
the block proposed in a uniquely successful slot will be downloaded within the same time slot
with overwhelming probability
(Lemmas~\ref{lem:loner_download},~\ref{lem:num_blocks_bound}).
In Section~\ref{sec:other_download_rules}, we apply a similar analysis to two other download rules, `equivocation avoidance' and `blocklisting',
to show that they too satisfy $\grpromempty$ with overwhelming probability for suitable $\bwslot$ (Lemmas~\ref{prop:unique_download_longest_header_chain_rule},~\ref{lem:bounded_equivocations}).

In Corollary~\ref{cor:security-freshest-block}, we identify the parameter values under which the protocol $\protocol$
with the freshest block download rule
is secure for a given desired resilience $\beta$
(similarly for `equivocation avoidance' and `blocklisting' in Corollary~\ref{thm:security_eq_avoid}).
For the rate of uniquely successful slots to exceed the rate of adversarial slots, we require that 
the rate of block production per slot,
$\rho$, be bounded as a function of $\beta$, so that most slots with honest block proposals are also uniquely successful slots.
A similar constraint exists in the synchronous model \cite{sleepy,dem20,david2018ouroboros} where the product of the block production rate and network delay $\Delta$ is bounded by a function of $\beta$.
Next, we require that the per slot bandwidth $\bwslot = \Omega(\kappa)$ (where $\kappa$ is the security parameter) so that 
$\grpromempty$
is satisfied 
with overwhelming probability.
This implies that the time slot $\tau = \DeltaHeader + \frac{\bwslot}{\bwtime} = \Omega(\kappa)$. 
This is similar to \cite{near-optimal-thruput} where under a bandwidth constrained model, the probabilistic delay bound increases with the security parameter.
Finally, the confirmation time $\Tconf = \Omega(\kappa^2)$ slots,
similar to that in the synchronous model \cite{sleepy}.

%% file: 5_analysis_details.tex
\section{Security Proof}
\label{sec:analysis}

\subsection{Definitions}
\label{sec:analysis-definitions}

The PoS LC protocol $\protocol$
has three parameters. The length of each time slot is $\tau$ seconds, the average number of nodes eligible to propose a block per time slot is $\rho$, and the confirmation latency is $\Tconf$ slots. 
The network has the following additional parameters. 
Each honest node has a download bandwidth of $\bwtime$ block contents per second 
(for convenience, we fix the size of the block content).
Henceforth, we fix $\tau=\DeltaHeader + \frac{\bwslot}{\bwtime}$ such that
each honest node can download the content of $\bwslot$ blocks in one time slot after receiving the headers proposed in that slot.
The adversary controls $\beta$ fraction of the stake. 
We denote by $\kappa$ the security parameter.
An event $E_{\kappa}$ will be said to occur with \emph{overwhelming} probability if $\Prob{E_{\kappa}} \geq 1 - \negl(\kappa)$.
Here, a function $f(\kappa)$ is said to be negligible or $\negl(\kappa)$ if for all $n>0$, there exists $\kappa_n^*$ such that for all $\kappa > \kappa_n^*$, $f(\kappa) < \frac{1}{\kappa^n}$.

Define the random variables $H_t$ and $A_t$ for $t=1,2,...$ to be the number of honest and adversarial nodes eligible to propose a block in slot $t$, respectively.
We consider the regime where the number of nodes $N\to\infty$ and each of them holds an equal fraction of the total stake. In this setting,
by the Poisson approximation to a binomial random variable,
we have $H_t\overset{\text{i.i.d.}}{\sim}\mathrm{Poisson}((1-\beta)\rho)$
and $A_t\overset{\text{i.i.d.}}{\sim}\mathrm{Poisson}(\beta\rho)$, all independent of each other.
An execution $\exec$ of 
time horizon $\Th$
is defined as the sequence $\{H_t,A_t\}_{t\leq\Th}$.
Denote by $\dC_i(t)$ the longest fully downloaded chain of an honest node $\node{i}$ at the end of slot $t$. Let $\len{b}$ denote the height of a block $b$. We will also use the same notation $\len{\calC}$ to denote the length of a chain $\calC$. Define $L_i(t)=\len{\dC_i(t)}$ and $L_{\min}(t) = \min_i L_i(t)$.
At the end of each slot, honest node $\node{i}$ outputs the ledger $\LOG{i}{t} = \dC_i(t)\trunc{\Tconf}$, which consists of a list of transactions as ordered in all blocks in $\dC_i(t)$ with time slot up to $t-\Tconf$.

For a given execution of a consensus protocol, we define
the following two properties:
\begin{itemize}
    \item 
\emph{Safety:}
For all adversarial strategies, for all time slots $t,t'$ and honest nodes $\node{i}, \node{j}$, $\LOG{i}{t}\preceq\LOG{j}{t'}$ or $\LOG{j}{t'}\preceq\LOG{i}{t}$.
\item
\emph{Liveness with parameter $\Tlive$:}
For all adversarial strategies, if a transaction {\sf tx} is received by all honest nodes 
before slot $t$,
then for all honest nodes $\node{i}$ and slots $t'\geq t+\Tlive$, $\mathsf{tx}\in\LOG{i}{t'}$.
\end{itemize}
A consensus protocol is \emph{secure} 
over a time horizon $\Th$
with parameter $\Tlive$ 
if it satisfies safety and liveness 
with parameter $\Tlive$
with overwhelming probability over executions
of time horizon $\Th$.

\begin{definition}
A slot $t$ is called 
\emph{successful} if $A_t+H_t>0$,
\emph{uniquely successful} if $A_t=0$ and $H_t=1$,
and adversarial if 
it is successful but not uniquely successful.
Define the predicates $\Unique{t}$ as true iff slot $t$ is uniquely successful and $\Adverse{t}$ as true iff slot $t$ is adversarial.
\end{definition}

For $s>r$, denote by $\Bint{r,s}$, $\Lint{r,s}$ and $\Nint{r,s}$, the number of successful, uniquely successful, and adversarial slots in the interval $(r,s]$ respectively.
\begin{IEEEeqnarray}{C}
    \Lint{r,s} \triangleq \sum_{t=r+1}^s \ind{\Unique{t}}, \quad
    \Nint{r,s} \triangleq \sum_{t=r+1}^s \ind{\Adverse{t}} \IEEEeqnarraynumspace
\end{IEEEeqnarray}
and $\Bint{r,s} = \Lint{r,s}+\Nint{r,s}$. 
When $r = s$, then $(r,s] = \emptyset$ and thus $\Bint{r,s} = \Lint{r,s} = \Nint{r,s} = 0$.
We define the following constants:
\begin{IEEEeqnarray}{rClCl}
    p &\triangleq& \Prob{A_t+H_t>0} &=& 1-e^{-\rho},\\
    \pu &\triangleq& \Prob{\Unique{t}} &=& (1-\beta)\rho e^{-\rho},\\
    \pa &\triangleq& \Prob{\Adverse{t}} &=& p - \pu
\end{IEEEeqnarray}

\begin{definition}
\label{def:grprom-definition}
    For a given download rule $\dlrule$, execution $\exec$ and $r<s \leq \Th$,  $\grprom{(r,s]}(\exec, \dlrule)$ holds iff
    for all adversarial strategies,
    for all uniquely successful slots in $(r,s]$, the block proposed in that slot is downloaded by all honest nodes 
    within the first $\bwslot$ blocks downloaded in that slot.
\end{definition}
We abbreviate $\grprom{(0,\Th]}(\exec,\dlrule)$ as $\grpromempty(\exec,\dlrule)$.
The inputs $\exec$ and $\dlrule$ to predicates are omitted where obvious.

\subsection{General Proof Overview}
\label{sec:analysis-overview}

\begin{lemma}
\label{prop:chain_growth}
    Let
    a download rule $\dlrule$,
    an execution $\exec$ and
    $t_0 < s \leq \Th$ 
    be
    such that $\grprom{(t_0,s]}(\exec,\dlrule)$ holds.
    Let $t_1, ..., t_m$ be the uniquely successful slots in $(t_0,s]$. Then, 
    \begin{enumerate}
        \item \label{item:inc_height} For all $j \geq 1$, $\len{b_j} > \len{b_{j-1}}$, where $b_j$ is the block proposed in $t_j$.
        \item For all $0\leq j \leq m$ and $t_j \leq t\leq s$,
        \begin{IEEEeqnarray}{C}
            L_{\min}(t)-L_{\min}(t_j) \geq \Lint{t_j,t}
        \end{IEEEeqnarray}
    \end{enumerate}
\end{lemma}
\begin{proof}
    Part (\ref{item:inc_height}) is easily seen by the fact that honest nodes propose on their longest valid downloaded chain, $b_{j-1}$ is downloaded before $b_j$ is proposed, and is valid because it was proposed by an honest node.
    Now, fix a $j$ such that $0\leq j \leq m$. 
    If $j=m$, then $\Lint{t_m,t}=0$ and $L_{\min}(t) \geq L_{\min}(t_m)$ for all $t_m \leq t \leq s$ because $L_{\min}$ is non-decreasing.
    For $j<m$, since honest nodes propose on their longest downloaded chain, $\len{b_{j+1}}\geq L_{\min}(t_{j+1}-1)+1 \geq L_{\min}(t_j)+1$.
    From part (\ref{item:inc_height}) and that 
    the blocks from uniquely successful slots in $(t_j,t]$ are downloaded before the end of their respective slots, we conclude that $L_{\min}(t)\geq \len{b_{j+1}} + \Lint{t_j,t}-1 \geq L_{\min}(t_j) + \Lint{t_j,t}$.
\end{proof}

\begin{theorem}
\label{thm:security}
For all $\bwslot \in \mathbb{N}$
and download rules $\dlrule$ such that
\begin{IEEEeqnarray}{c}
    \Prob{\exec \colon \lnot \grpromempty(\exec,\dlrule)} \leq \negl(\kappa), \IEEEeqnarraynumspace
\end{IEEEeqnarray}
if
$\SecurityConditionFull$
for some $\epsilon_1 \in (0,1)$,
$\tau = \DeltaHeader + \frac{\bwslot}{\bwtime}$ and 
$\Tconf=\Omega\left((\kappa+\ln\Th)^2\right)$,
then
the protocol $\protocol$ with download rule $\dlrule$
is secure with
$\Tlive=\Omega\left((\kappa+\ln\Th)^2\right)$.
\end{theorem}

Theorem~\ref{thm:security} is proved in Appendix~\ref{sec:appendix-proof-security-theorem} using techniques similar to \cite{sleepy}.

\subsection{`Download Towards the Freshest Block' Rule}
\label{sec:analysis-freshest-block}

\begin{definition}
\label{def:short_prefixes}
For an execution $\exec$, 
$\FewBlockOpps(\exec)$ holds iff 
\begin{IEEEeqnarray}{C}
    \forall t \leq \Th\colon \max_{r<t\colon \Unique{r} \land (\Nint{r,t} \geq \Lint{r,t})} \Nint{r,t} < \bwslot.
    \IEEEeqnarraynumspace
\end{IEEEeqnarray}
\end{definition}

\begin{lemma}
\label{lem:loner_download}
For an execution $\exec$
and the freshest block download rule $\dlrulefresh$ (Algorithm~\ref{algo:freshest-block-rule}),
\begin{IEEEeqnarray}{c}
    \FewBlockOpps(\exec) \implies \grpromempty(\exec,\dlrulefresh) \IEEEeqnarraynumspace
\end{IEEEeqnarray}
\end{lemma}

\begin{proof}
Let $t_1,...,t_m$ be the uniquely successful slots in $(0,\Th]$.
Let $b_j$ be the 
block from $t_j$ for some $1\leq j\leq m$.
The header of $b_j$ is received by all honest nodes within $\DeltaHeader$ time after the beginning of slot $t_j$. Due to the downloading rule, during slot $t_j$ all honest nodes download the chain containing $b_j$.
Furthermore, since $b_j$ is an honest block and honest nodes only propose
on their downloaded chain,
the prefix of $b_j$ can be downloaded
(\ie, does not contain invalid or missing blocks).
Thus, we only need to show that the prefix of $b_j$ contains at most $\bwslot$ blocks whose contents have not been downloaded.

For induction, assume that $\grprom{(0,t_j-1]}$ holds. Using this, we will show that $\grprom{(0,t_{j+1}-1]}$ holds. For the base case, this is true for $j=1$ since $t_1$ is the first uniquely successful slot by definition.
Note that the block $b_j$, being honest, is proposed on the tip of $\dC_i(t_j-1)$ for some $i$.
Let $r_j$ be the last unique time slot such that the block $b_j'$ from that time slot is in $\dC_i(t_j-1)$.
Clearly, $r_j \leq t_j-1$.
Then, 
\begin{IEEEeqnarray}{C}
\label{eq:cgeq1}
    \len{\dC_i(t_j-1)} \leq \len{b_j'} + \Nint{r_j,t_j-1}
\end{IEEEeqnarray}
since blocks after $b_j'$ are from adversarial slots by definition of $r_j$. 
From 
$\grprom{(0,t_j-1]}$ and part~(\ref{item:inc_height}) of Lemma~\ref{prop:chain_growth}, 
\begin{IEEEeqnarray}{C}
\label{eq:cgeq2}
    \len{b_{j-1}} \geq \len{b_j'} + \Lint{r_j,t_j-1}.
\end{IEEEeqnarray}
Since $b_{j-1}$ is downloaded by the end of slot $t_{j-1}$ and $t_j -1 \geq t_{j-1}$, $\len{\dC_i(t_j-1)} \geq \len{b_{j-1}}$, and this would imply from \eqref{eq:cgeq1} and \eqref{eq:cgeq2} that $\Nint{r_j,t_j-1} \geq \Lint{r_j,t_j-1}$.
Note that time slots of blocks in a valid chain must be strictly increasing. Since $b_j'$ is already downloaded, the number of blocks in $\dC_i(t_j-1)$ whose content is not downloaded is at most $\Nint{r_j,t_j-1}$. Since $b_j$ extends $\dC_i(t_j-1)$, the number of block contents to be downloaded including the prefix of $b_j$ is at most $\Nint{r_j,t_j-1}+1$. As per $\FewBlockOpps$, this is at most $\bwslot$ (note that $r_j \leq t_j-1$). Therefore, $b_j$ is downloaded within one slot. Since there are no more uniquely successful slots in $(t_j,t_{j+1})$, this completes the induction step by showing that $\grprom{(0,t_{j+1}-1]}$.
For $j=m$, we would conclude with $\grpromempty$ as required.
\end{proof}

\begin{lemma}
\label{lem:num_blocks_bound}
If $\SecurityConditionFull$ for some $\epsilon_1 \in (0,1)$ and $\bwslot = \pa \Td(1+\epsilon_2)$ for some $\epsilon_2>0$ 
where $\Td = \frac{\Omega(\kappa + \ln\Th)}{\constNumBlocks  p}$,
then 
\begin{IEEEeqnarray}{C}
\label{eq:num_blocks_bound}
    \Prob{ \exec\colon \lnot\FewBlockOpps(\exec) } 
    \leq \negl(\kappa),
    \IEEEeqnarraynumspace
\end{IEEEeqnarray}
where $\constNumBlocks  = \min\left\{ \frac{\epsilon_1^2}{36}, \frac{\epsilon_2^2}{\epsilon_2+2}\frac{\pa}{p}\right\}$.
\end{lemma}

\begin{corollary}
\label{cor:security-freshest-block}
The protocol $\protocol$ 
with the freshest block download rule
and parameters
$\rho$ such that $\SecurityConditionFull$,
$\tau = \Omega(\kappa+\ln\Th)$,
$\Tconf=\Omega\left((\kappa+\ln\Th)^2\right)$
is secure with
$\Tlive=\Omega\left((\kappa+\ln\Th)^2\right)$.
\end{corollary}

Lemma~\ref{lem:num_blocks_bound} is proved in Appendix~\ref{sec:appendix-proof-num-blocks-bound}
and
Corollary~\ref{cor:security-freshest-block} is obtained by setting $\tau = \DeltaHeader + \frac{\bwslot}{\bwtime}$ and $\bwslot$ as per Lemma~\ref{lem:num_blocks_bound}.

%% file: 6_experiments.tex
\section{Experiments}
\label{sec:experiments}

\subsection{Implementation Details}
\label{sec:impl}

We implemented our PoS LC node in $800$ lines of Golang.\footnote{The source code is available at: \url{https://github.com/yangl1996/synclc-sim}}
For all of our experiments, the slot duration $\tau$ is set to $1\,\mathrm{second}$, and the total block production rate is $0.06\,\mathrm{blocks/s}$. There is no transaction processing. Instead, nodes fill blocks with random bytes up to a size limit ($100\,\mathrm{KB}$ in our experiments).

Our implementation has a fully-featured network stack modelled after Cardano's node software \cite{cardano-network-protocols,cardano-network-spec}. Similar to Cardano, block propagation involves two subsystems: \emph{chain sync}, and \emph{block fetch}. The chain sync subsystem allows a node to advertise the header chain of the longest chain it has downloaded and validated, and to track the header chains advertised by peers. Because the header only takes a tiny fraction of space in a block, the bandwidth consumed by the chain sync subsystem is negligible. In all of our experiments, chain sync only consumed up to 1.2\% of the available bandwidth.

The block fetch subsystem periodically examines the header chains learned from peers through chain sync, and sends requests to download block bodies according to a download rule. We implement the two download rules discussed in Section~\ref{sec:introduction}: `download along the longest header chain', and `download towards the freshest block'. Similar to Cardano, our block fetch logic limits the maximum number of peers to concurrently download from, an important parameter which we call the \emph{in-flight cap}. This ensures the limited network bandwidth is never spread too thin across too many concurrent downloads. Finally, chain sync and block fetch share the same TCP connection for each pair of peers. To avoid head-of-line blocking, we multiplex the two subsystems so that chain sync is never impaired by block fetch traffic.

To simulate bandwidth constraints, we build our testbed using Mininet \cite{mininet}. Each blockchain node runs in a Mininet virtual host with its own network interface, and is connected to a central switch through a link with limited bandwidth and artificial propagation delay. Specifically, we limit the bandwidth of honest nodes to $20\,\mathrm{Mbps}$, and adversarial nodes to $1\,\mathrm{Gbps}$. We set the round-trip time between any pair of nodes to $100\,\mathrm{ms}$. The testbed runs on a workstation with two Intel Xeon E5-2623 v3 CPUs and $32\,\mathrm{GB}$ of RAM.

\subsection{Demonstration of the Spamming Attack}
\label{sec:spam-demo}

In this experiment, we show that the widely-adopted `download along the longest chain' rule is vulnerable to adversarial spamming, and the `download towards the freshest block' rule mitigates this attack. There are $20$ honest nodes connected in a full mesh topology. Honest nodes equally split 67\% of the total stake, so each honest node has a block production rate of $0.002\,\mathrm{block/s}$. The adversary controls 33\% of the stake ($0.02\,\mathrm{block/s}$), and sets up $5$ attacking nodes. Each attacking node connects to all honest nodes. The adversary uses the attacking nodes to monitor the longest chains announced by honest nodes, and tries to mine equivocating spam chains (\cf Figure~\ref{fig:poslclongestheaderchainattack}). When successful, the adversary announces them and hopes honest nodes download these spam chains.

\input{eval-1-chain-growth}

\input{eval-2-chain-growth-different-cap}

\input{eval-3-propagation-delay}

Figure~\ref{fig:spam-trace} shows the time series of honest chain growth over an hour when the in-flight cap is set to $2$. Note that honest chain growth stalls after $400$ seconds when nodes download the longest chain. Since there are $5$ attacking nodes, once the adversary gets a longer chain by luck, each honest user will use all of its $2$ in-flight slots to download spam chains (from $2$ of the $5$ attacking nodes), leaving no room for honest blocks. Before any honest node finishes downloading a spam block, the adversary will have advertised another equivocating chain, keeping the honest nodes busy. Although honest nodes can still mine blocks, they cannot download blocks from each other, so each honest node effectively mines on its own fork. The resulting heavy forking causes the honest chains to grow slower than the adversary mining rate, and the adversary maintains the lead in chain length and sustains this attack (red-shaded areas in Figure~\ref{fig:spam-trace}) until the experiment ends.

In comparison, honest chain growth is unaffected when nodes download towards the freshest block. Note that although the adversary is still able to trick honest nodes into downloading spam blocks (blue-shaded areas in Figure~\ref{fig:spam-trace}), the adversary cannot \emph{sustain} the attack: when a new honest block is produced, the chain containing that fresh block will be prioritized. Before the adversary manages to produce a fresher block, all honest nodes will have caught up on the correct chain. Further experiments in Appendix~\ref{sec:vary-block-size} show that honest chain growth is unaffected with even larger block sizes.

\subsection{Impact of the In-Flight Cap}
\label{sec:inflight-cap}

We now extend the previous experiment by varying the in-flight cap between $2$ and $7$, and demonstrate the relationship between the in-flight cap and the number of attacking nodes. Figure~\ref{fig:cap-grow} shows the results. When the in-flight cap is equal to or smaller than the number of attacking nodes, downloading the longest chain is not secure. This may remind readers of the eclipse attack \cite{eclipse,eclipse1,eclipse2,eclipse3,eclipse4}: the adversary is in fact eclipsing the honest nodes in the block fetch subsystem by occupying all its in-flight slots. Meanwhile, downloading the freshest chain is always secure \emph{regardless} of the in-flight cap, because a fresh honest block can \emph{break} such eclipse.

Figure~\ref{fig:cap-grow} seems to suggest that downloading the longest chain is secure when the in-flight cap is larger than the number of attackers. 
Is it true? Should we then increase the in-flight cap to infinity? We point out that the in-flight cap ensures each in-flight download gets a sufficiently large share of the available bandwidth to complete in a reasonable amount of time. This is critical in ensuring low propagation delay for honest blocks. As an extreme example, assume that there are a large number of attacking nodes and an infinite in-flight cap. Although a node will always start downloading an honest block as soon as it receives the announcement, the bandwidth allocated to download this block will be extremely small due to the competing downloads of adversarial blocks, effectively halting the download. As a result, a finite in-flight cap is necessary, and the attacker can always attack the `download along the longest chain' rule by outnumbering the in-flight cap.

To demonstrate this effect, we remove the in-flight cap,
increase the number of attacking nodes to $10$, and measure the block propagation time. The results in Figure~\ref{fig:propagation-delay} show that the propagation time under both rules increases. This is because when the attack \emph{is} active, there are more competing flows downloading spam blocks, leaving less bandwidth for honest blocks. Still, the chain growth rate is unharmed when downloading the freshest chain, at $0.041\,\mathrm{block/s}$. This is because nodes can break away from the spam chain as soon as a new honest block is produced, regardless how bad the propagation time is under active spam. In comparison, the propagation delay when downloading the longest chain becomes much worse. In fact, the higher delay causes the chain growth rate to drop to $0.035\,\mathrm{block/s}$. 
In conclusion, removing or increasing the in-flight cap does not save the `download along the longest chain' rule, but impacts the block propagation delay of the `download towards the freshest block' rule
only slightly so that security is unaffected.

\subsection{Bandwidth Consumption}%

\input{eval-4-bandwidth-trace}

Besides block bodies, a blockchain node needs to receive other types of traffic in real time, such as unconfirmed transactions, requests from clients, and remote control data. A practical download rule must not consume all the available bandwidth at a node at all time. As explained in Section~\ref{sec:spam-demo}, under the `download towards the freshest block' rule, an honest node breaks free from the spam chains when an honest block is mined. That is, spamming stops when there is a time slot with only one honest block proposed. Intuition suggests that as long as the overall mining rate is not too high, such event should happen frequently. Indeed, the ingress traffic traces in Figure~\ref{fig:bandwidth-trace} show that periods of high network utilization only last for tens of seconds when downloading the freshest block, quickly succeeded by long windows of low utilization. In comparison, when downloading the longest chain, the period of high utilization lasts until the end of the experiment, leaving no room for honest blocks or other traffic.

%% file: eval-1-chain-growth.tex
\begin{figure}[t]
    \centering
    \begin{tikzpicture}
        \small
        \begin{axis}[
            mysimpleplot,
            xlabel={Time [$\mathrm{s}$]},
            ylabel={Local chain length},
            legend columns=2,
            xmin=0, xmax=3600,
            ymin=0, ymax=150,
            height=0.4\linewidth,
            width=\linewidth,
        ]
        
        \addlegendimage{myParula01Blue,no marks}
        \addlegendentry{Download freshest block}
        \addlegendimage{myParula02Orange,no marks}
        \addlegendentry{Download longest chain}
        \addlegendimage{area legend,draw opacity=0,fill=myParula01Blue,fill opacity=0.2,draw=none};
        \addlegendentry{Attack on freshest block};
        \addlegendimage{area legend,draw opacity=0,fill=myParula02Orange,fill opacity=0.15,draw=none};
        \addlegendentry{Attack on longest chain};

        \foreach \tStart/\tStop in {90.998647/138.998647, 395.998647/417.998647, 450.998647/510.998647, 530.998647/536.998647, 671.998647/700.998647, 762.998647/780.998647, 834.998647/929.998647, 951.998647/1018.998647, 1044.998647/1102.998647, 1117.998647/1168.998647, 1175.998647/1190.998647, 1239.998647/1243.998647, 1247.998647/1253.998647, 1262.998647/1279.998647, 1296.998647/1308.998647, 1316.998647/1321.998647, 1331.998647/1341.998647, 1365.998647/1383.998647, 1390.998647/1423.998647, 1611.998647/1625.998647, 1631.998647/1668.998647, 1732.998647/1766.998647, 1781.998647/1803.998647, 1820.998647/1833.998647, 1884.998647/1933.998647, 1967.998647/1972.998647, 1984.998647/1989.998647, 2091.998647/2101.998647, 2109.998647/2130.998647, 2215.998647/2216.998647, 2245.998647/2279.998647, 2343.998647/2365.998647, 2430.998647/2450.998647, 2468.998647/2503.998647, 2622.998647/2639.998647, 2666.998647/2678.998647, 2743.998647/2744.998647, 2782.998647/2798.998647, 2804.998647/2813.998647, 2819.998647/2844.998647, 2876.998647/2917.998647, 2936.998647/2989.998647, 3063.998647/3078.998647, 3196.998647/3224.998647, 3229.998647/3253.998647, 3275.998647/3278.998647, 3298.998647/3299.998647, 3328.998647/3333.998647, 3351.998647/3356.998647, 3435.998647/3455.998647, 3471.998647/3539.998647, 3603.998647/3615.998647} {
            \edef\temp{\noexpand\draw [fill=myParula01Blue,fill opacity=0.2,draw=none] (axis cs:\tStart,0) rectangle (axis cs:\tStop,150);}
            \temp
        }

        \foreach \tStart/\tStop in {96.999632/103.999632, 113.999632/264.999632, 401.999632/3959.999632} {
            \edef\temp{\noexpand\draw [fill=myParula02Orange,fill opacity=0.15,draw=none] (axis cs:\tStart,0) rectangle (axis cs:\tStop,150);}
            \temp
        }

        \foreach \n in {0,...,19}{
            \addplot [myParula01Blue,no marks,ultra thin] table [x=time,y=height] {experiments/attack_demo_freshest/chain_growth_victim_\n.txt};
        }

        \foreach \n in {0,...,19}{
            \addplot [myParula02Orange,no marks,ultra thin] table [x=time,y=height] {experiments/attack_demo_longest/chain_growth_victim_\n.txt};
        }
        
        \end{axis}
    \end{tikzpicture}%
    \vspace{-0.5em}%
    \caption{Traces of honest chain growth under spamming attack (\cf Figure~\ref{fig:poslclongestheaderchainattack}) when using different download rules and an in-flight cap of 2. Each curve represents one honest node. Shaded areas represent time periods when nodes are suffering from the attack and are downloading invalid blocks. PoS LC downloading longest chain stalls. PoS LC downloading freshest blocks is robust.}
    \label{fig:spam-trace}
\end{figure}
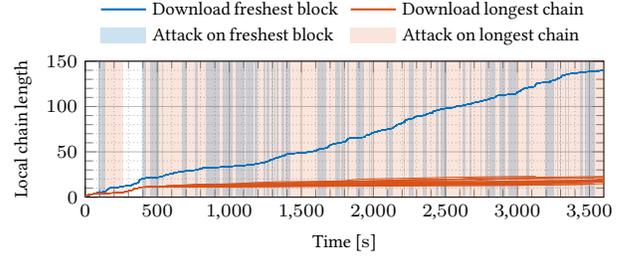

%% file: eval-2-chain-growth-different-cap.tex
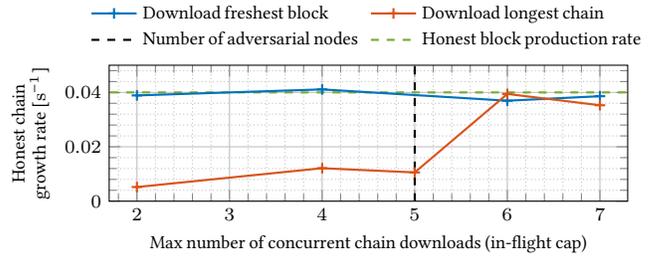
\begin{figure}[t]
    \centering
    \begin{tikzpicture}
        \small
        \begin{axis}[
            mysimpleplot,
            xlabel={Max number of concurrent chain downloads (in-flight cap)},
            ylabel={Honest chain\\growth rate [$\mathrm{s}^{-1}$]},
            legend columns=2,
            xmin=1.7, xmax=7.3,
            ymin=0, ymax=0.05,
            height=0.4\linewidth,
            width=\linewidth,
            yticklabel style={
                /pgf/number format/fixed,
                /pgf/number format/precision=2
            },
            scaled y ticks=false,
        ]
        
        \addlegendimage{myparula11}
        \addlegendentry{Download freshest block}
        \addlegendimage{myparula21}
        \addlegendentry{Download longest chain}
        
        \addlegendimage{dashed,thick}
        \addlegendentry{Number of adversarial nodes}
        \addlegendimage{dashed,myParula05Green,thick}
        \addlegendentry{Honest block production rate}

        \draw [dashed,thick] (axis cs:5,0) -- (axis cs:5,1);
        \draw [dashed,myParula05Green,thick] (axis cs:0,0.04) -- (axis cs:100,0.04);

        \addplot [myparula11] table [x=cap,y=growth] {experiments/attack_variable_cap/freshest.txt};
        
        \addplot [myparula21] table [x=cap,y=growth] {experiments/attack_variable_cap/longest.txt};

        \end{axis}
    \end{tikzpicture}%
    \vspace{-0.5em}%
    \caption{Honest chain growth rate under spamming attack (\cf Figure~\ref{fig:poslclongestheaderchainattack}) while allowing concurrent block downloads from different number of peers. With in-flight cap below the number of adversarial peers, PoS LC downloading longest chain shows performance degradation; PoS LC downloading freshest block is robust.}
    \label{fig:cap-grow}
\end{figure}

%% file: eval-3-propagation-delay.tex
\begin{figure}[t]
    \centering
    \begin{tikzpicture}
        \small
        \begin{axis}[
            mysimpleplot,
            xlabel={Propagation delay [$\mathrm{s}$]},
            ylabel={Fraction of blocks},
            legend columns=2,
            xmin=0, xmax=12,
            height=0.4\linewidth,
            width=\linewidth,
        ]
        
        \addlegendimage{myParula01Blue,no marks}
        \addlegendentry{Freshest block, 10 attackers}
        \addlegendimage{myParula02Orange,no marks}
        \addlegendentry{Longest chain, 10 attackers}
        \addlegendimage{myParula01Blue,no marks, dashed, opacity=0.5}
        \addlegendentry{Freshest chain, 5 attackers}
        \addlegendimage{myParula02Orange,no marks, dashed, opacity=0.5}
        \addlegendentry{Longest chain, 5 attackers}

        \addplot [myParula01Blue,no marks] table [x=midpoint,y=density] {experiments/propagation_delay/freshest.txt};

 \addplot [myParula02Orange,no marks] table [x=midpoint,y=density] {experiments/propagation_delay/longest.txt};        
 
  \addplot [myParula02Orange,no marks, dashed, opacity=0.5] table [x=midpoint,y=density] {experiments/propagation_delay/longest-5-attackers.txt};
  
  \addplot [myParula01Blue,no marks, dashed, opacity=0.5] table [x=midpoint,y=density] {experiments/propagation_delay/freshest-5-attacker.txt};

        \end{axis}
    \end{tikzpicture}%
    \vspace{-0.5em}%
    \caption{Empirical cumulative density function of block propagation delay under different download rules, facing different number of attackers, and an infinite in-flight cap.}
    \label{fig:propagation-delay}
\end{figure}
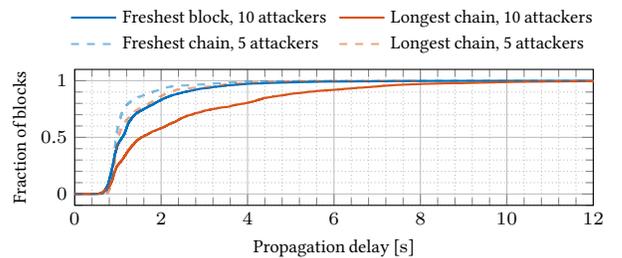

%% file: eval-4-bandwidth-trace.tex
\begin{figure}[t]
    \centering
    \begin{tikzpicture}
        \small
        \begin{axis}[
            mysimpleplot,
            xlabel={Time [$\mathrm{s}$]},
            ylabel={Download traffic [$\mathrm{Mbps}$]},
            legend columns=2,
            xmin=0, xmax=500,
            ymin=-1, ymax=21,
            height=0.4\linewidth,
            width=\linewidth,
        ]
        
        \addlegendimage{myParula01Blue,no marks}
        \addlegendentry{Download freshest block}
        \addlegendimage{myParula02Orange,no marks}
        \addlegendentry{Download longest chain}
        \addlegendimage{area legend,draw opacity=0,fill=myParula01Blue,fill opacity=0.2,draw=none};
        \addlegendentry{Attack on freshest block};
        \addlegendimage{area legend,draw opacity=0,fill=myParula02Orange,fill opacity=0.15,draw=none};
        \addlegendentry{Attack on longest chain};

        \foreach \tStart/\tStop in {-1902.00191/-1869.00191, -1597.00191/-1570.00191, -1542.00191/-1522.00191, -1487.00191/-1476.00191, -1462.00191/-1457.00191, -1321.00191/-1316.00191, -1230.00191/-1169.00191, -1158.00191/-1139.00191, -1114.00191/-1111.00191, -1088.00191/-1075.00191, -1041.00191/-1031.00191, -1020.00191/-1006.00191, -988.00191/-983.00191, -948.00191/-892.00191, -875.00191/-854.00191, -817.00191/-740.00191, -730.00191/-711.00191, -696.00191/-694.00191, -676.00191/-658.00191, -654.00191/-574.00191, -381.00191/-373.00191, -362.00191/-339.00191, -260.00191/-214.00191, -211.00191/-130.00191, -108.00191/-99.00191, -89.00191/-50.00191, -25.00191/9.99809, 98.99809/110.99809, 116.99809/127.99809, 222.99809/268.99809, 350.99809/386.99809, 437.99809/469.99809, 475.99809/516.99809, 629.99809/640.99809, 673.99809/678.99809, 750.99809/772.99809, 789.99809/807.99809, 811.99809/842.99809, 883.99809/909.99809, 943.99809/965.99809, 977.99809/1036.99809, 1070.99809/1074.99809, 1203.99809/1207.99809, 1236.99809/1265.99809, 1282.99809/1349.99809, 1358.99809/1408.99809, 1442.99809/1448.99809, 1455.99809/1542.99809} {
            \edef\temp{\noexpand\draw [fill=myParula01Blue,fill opacity=0.2,draw=none] (axis cs:\tStart,-10) rectangle (axis cs:\tStop,150);}
            \temp
        }

        \foreach \tStart/\tStop in {-1906.001542/-1892.001542, -1889.001542/-1883.001542, -1601.001542/-1587.001542, -1546.001542/-1538.001542, -1491.001542/-1479.001542, -1466.001542/-1452.001542, -1325.001542/-1316.001542, -1234.001542/-1229.001542, -1217.001542/1587.998458} {
            \edef\temp{\noexpand\draw [fill=myParula02Orange,fill opacity=0.15,draw=none] (axis cs:\tStart,-10) rectangle (axis cs:\tStop,150);}
            \temp
        }

        \addplot [myParula01Blue,no marks,thin] table [x=time,y=mbps] {experiments/bandwidth_usage/freshest-cap-4.txt};

        \addplot [myParula02Orange,no marks,thin] table [x=time,y=mbps] {experiments/bandwidth_usage/longest-cap-4.txt};

        \end{axis}
    \end{tikzpicture}%
    \vspace{-0.5em}%
    \caption{Traces of download traffic over a $500$-second period at one of the victim nodes when using different download rules and in-flight cap of $4$. Shaded areas represent time periods when the node is suffering from the attack and downloading invalid blocks.}
    \label{fig:bandwidth-trace}
\end{figure}
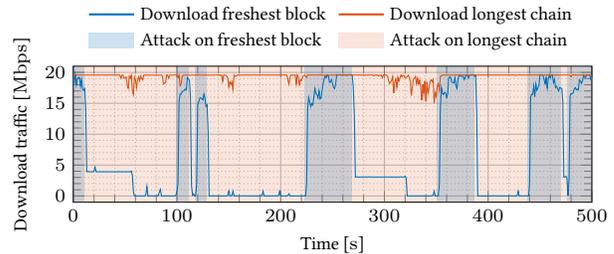

%% file: 8_equivocation_avoidance.tex
\section{Other Download Rules}
\label{sec:other_download_rules}

\subsection{Equivocation Avoidance}
\label{sec:equivocation_avoidance_rule}

We formalize a common heuristic to deal with equivocations, namely downloading only one out of many equivocations in Algorithm~\ref{algo:longest-header-chain-equivoc-avoid-rule}.
In every time slot when a node wishes to download blocks, it filters the tree consisting of block headers it has received by retaining only one leaf in the tree for each block production opportunity (determined by the proposing node and the time slot of each block). 
We strengthen the adversary by allowing it to decide,
per honest node and time slot,
which among multiple equivocating headers would be retained.
After removing equivocations, invalid chains, and chains that are already downloaded from the header tree, the node selects the longest header chain, and downloads the content for the first missing block in this chain. 

Here, we have illustrated equivocation avoidance as a modification to the longest header chain download rule. By doing so, we show that while the longest header chain download rule was insecure by itself, it can be made secure by `avoiding equivocations'.
However, equivocation avoidance could also be added to the freshest block rule to make it more efficient.
Our analysis below suggests that the freshest block rule by itself is already more efficient than the longest header chain rule with equivocation avoidance because the latter requires downloading a much larger number of blocks within one time slot (see Corollaries~\ref{cor:security-freshest-block}, \ref{thm:security_eq_avoid}), leading to longer time slots and poorer bandwidth utilization.

\begin{algorithm}[t]
    \caption[]{`Equivocation avoidance' download rule;
    replaces $\operatorname{downloadRule}$ in Algorithm~\ref{algo:pseudocode-pos-lc} (\cf Appendix~\ref{sec:appendix-helperfunctions-pseudocode})}
    \label{algo:longest-header-chain-equivoc-avoid-rule}
    \begin{algorithmic}[1]
        \small
        \On{$\DeltaHeader$ time into each time slot $t$}
            \LineComment{Before beginning block content downloads for time slot $t$,
            filter current header tree to keep at most one leaf per block production opportunity, \ie, per $(\mathsf{node},\mathsf{time})$ pair (equivocation avoidance; ties broken adversarially)}
            \State $\HeaderTree^* \gets \operatorname{oneLeafPerProductionOpportunity}(\protocol.\HeaderTree)$
        \EndOn
        \RealFunction{$\operatorname{avoidEquivocations}(\HeaderTree, \TxsMap)$}
            \LineComment{Ignore chains with invalid content in any block}
            \State $\Tree' \gets \{ \Chain \in \HeaderTree^* \mid \forall \Chain' \preceq \Chain\colon \TxsMap[\Chain'] \neq \INVALID \}$
            \LineComment{Ignore downloaded chains}
            \State $\Tree' \gets \{ \Chain \in \Tree' \mid \TxsMap[\Chain] = \NIL \}$
            \LineComment{Select the longest chain}
            \State $\Chain \gets \argmax_{\Chain'\in\Tree'} \len{\Chain'}$   \label{loc:pseudocode-longest-headear-chain-decision-line}
            \LineComment{Find the first not downloaded block on that chain
            (if non-existent: $\bot$)}
            \State $\Chain' \gets \argmin_{\Chain''\preceq\Chain\colon  \TxsMap[\Chain'']=\NIL} \len{\Chain''}$   
            \State \Return{$\Chain'$}
        \EndRealFunction
    \end{algorithmic}
\end{algorithm}

\subsection{Analysis}
\label{sec:equivocation_avoidance_analysis}
We use the general framework developed in Section~\ref{sec:analysis} to prove security of PoS LC under the equivocation avoidance download rule.
Recall that we only need to prove that $\grpromempty$ (Definition~\ref{def:grprom-definition}) holds with overwhelming probability.

In a uniquely successful slot, honest nodes may not immediately download towards the block from that slot. This is because there could be other chains in a node's header tree that are longer (recall that Algorithm~\ref{algo:longest-header-chain-equivoc-avoid-rule} prioritizes the longest header chain after removing equivocations and invalid prefixes).
However, we can bound the number of blocks that will be downloaded before downloading the block from the uniquely successful slot.
With equivocation avoidance, 
honest nodes retain only one leaf in their header tree per block production opportunity. So, 
honest nodes download at most one chain per block production opportunity. Since 
block production opportunities are bounded, we will show in Lemma~\ref{prop:unique_download_longest_header_chain_rule} that there can not be too many longer chains in the honest node's header tree.

Define for slots $s\leq t$,
\begin{IEEEeqnarray}{C}
    W_{s,t} \triangleq \max_{r<s \colon \Unique{r} \land (\Nint{r,s} \geq \Lint{r,t})} \Nint{r,s}.
\end{IEEEeqnarray}

\begin{definition}
\label{def:bounded_eq}
For an execution $\exec$, $\BoundedEq(\exec)$ holds iff 
\begin{IEEEeqnarray}{C}
    \forall t \leq \Th\colon W_{t-1,t-1} + \sum_{s\leq t} A_{s} W_{s,t} < \downloadLimit.
    \IEEEeqnarraynumspace
\end{IEEEeqnarray}
\end{definition}

\begin{lemma}
\label{prop:unique_download_longest_header_chain_rule}
For an execution $\exec$ and the longest header chain download rule with equivocation avoidance $\dlruleeqavoid$ (Algorithm~\ref{algo:longest-header-chain-equivoc-avoid-rule}),
\begin{IEEEeqnarray}{c}
\BoundedEq(\exec) \implies \grpromempty(\exec,\dlruleeqavoid) \IEEEeqnarraynumspace
\end{IEEEeqnarray}
\end{lemma}

\begin{lemma}
\label{lem:bounded_equivocations}
If $\SecurityConditionFull$ for some $\epsilon_1\in(0,1)$ and $\downloadLimit = \pa \Tb (1+\beta\rho\Tb(1+\epsilon_3))(1+\epsilon_2)$
for some $\epsilon_2, \epsilon_3 > 0$
where $\Tb = \frac{\Omega(\kappa+\ln\Th)}{\constBoundedEq  p}$, then
\begin{IEEEeqnarray}{C}
\label{eq:bounded_equivocations}
    \Prob{ \exec\colon \lnot \BoundedEq(\exec) } 
    \leq \negl(\kappa)
\end{IEEEeqnarray}
where $\constBoundedEq  = \max\left\{ \frac{\epsilon_1^2}{36}, \frac{\epsilon_2^2}{\epsilon_2+2}\frac{\pa}{p},  \frac{\pu(1-\epsilon_3)}{p}\ln\left(\frac{\pu}{1-\pu}\right), \frac{\epsilon_3^2\pu}{2p}, \frac{\epsilon_3^2\beta\rho}{(\epsilon_3+2)p} \right\}$. 
\end{lemma}

Lemma~\ref{prop:unique_download_longest_header_chain_rule} is proved in Appendix~\ref{sec:appendix-proof-equi-avoid} and Lemma~\ref{lem:bounded_equivocations} in Appendix~\ref{sec:appendix-proof-bounded-equivocations}.
Then, we obtain the following corollary of Theorem~\ref{thm:security}.

\begin{corollary}
\label{thm:security_eq_avoid}
The protocol $\protocol$ 
with the equivocation avoidance download rule
and parameters
$\rho$ such that $\SecurityConditionFull$,
$\tau=\Omega\left((\kappa+\ln\Th)^2\right)$
and
$\Tconf=\Omega\left((\kappa+\ln\Th)^2\right)$,
is secure with
$\Tlive=\Omega\left((\kappa+\ln\Th)^2\right)$.
\end{corollary}

\subsection{Blocklisting}
\label{sec:blocklisting_rule}
Another common heuristic to deal with equivocations is `blocklisting' the proposer of equivocating blocks. Blocklisting can be implemented at the level of the download rule as follows: an honest node never downloads a chain whose tip is proposed by a party for which the node has seen two block headers with the same time slot (an equivocation). 
Blocklisting is a decision that is taken unilaterally by each honest node and may be taken at different points of time by different nodes.

Note that this is only a stricter version of the equivocation avoidance rule described in Section~\ref{sec:equivocation_avoidance_rule} because in any given time slot, a block that is rejected in the equivocation avoidance rule will also be rejected in the blocklisting rule. Moreover, any chain whose tip is proposed by an honest node will not be discarded under this rule. 
Therefore, for any execution $\exec$ and the blocklisting rule $\dlruleblock$
\begin{IEEEeqnarray}{c}
    \grpromempty(\exec,\dlruleeqavoid) \implies \grpromempty(\exec,\dlruleblock). \IEEEeqnarraynumspace
\end{IEEEeqnarray}
Therefore, security of PoS LC with the `blocklisting' rule is implied by security of PoS LC with `equivocation avoidance' rule and the same parameters.

%% file: 7_parallel.tex
\section{High Throughput Under Bandwidth Constraint}
\label{sec:parallel}

\begin{figure}
    \centering
    \begin{tikzpicture}[y=0.5cm]
        \footnotesize

        \begin{scope}[shift={(0,0)}]

            \node [align=left,anchor=west] at (6.2,0) {(a)};
            
            \node [align=right,anchor=east] at (1.8,0) {Passively following $\protocol$};
            \draw [draw=myParula07Red,fill=myParula07Red!50] (2,0.4) rectangle (2.3,-0.4);
            \draw [draw=myParula05Green,fill=myParula05Green!50] (2.3,0.4) rectangle (2.4,-0.4);
            \draw [black!20,line width=0.4pt] (2,0.4) rectangle (6,-0.4);

            \draw [|-|] (2,0.7) -- (2.4,0.7) node [pos=0,above,xshift=-4pt,yshift=1pt] {$\Omega(\frac{1}{\kappa})$};
            \draw [|-|] (2.3,1.1) -- (2.4,1.1) node [pos=1,above,xshift=4pt] {$\Omega(\frac{1}{\kappa})$};

        \end{scope}
        
        \begin{scope}[shift={(0,-1)}]

            \node [align=left,anchor=west] at (6.2,0) {(b)};

            \node [align=right,anchor=east] at (1.8,0) {Actively participating in $\protocol$};
            \draw [draw=myParula07Red,fill=myParula07Red!50] (2,0.4) rectangle (4.3,-0.4);
            \draw [draw=myParula05Green,fill=myParula05Green!50] (4.3,0.4) rectangle (4.4,-0.4);
            \draw [black!20,line width=0.4pt] (2,0.4) rectangle (6,-0.4);

            \draw [|-|] (2,-0.7) -- (4.4,-0.7) node [pos=0.4,below] {$\leq\left(\frac{1+p}{2}\right)$};
            \draw [|-|] (4.3,-1.1) -- (4.4,-1.1) node [pos=1,below] {$\Omega(\frac{1}{\kappa})$};

        \end{scope}
        
        \begin{scope}[shift={(0,-4)}]

            \node [align=left,anchor=west] at (6.2,0) {(c)};

            \node [align=right,anchor=east] at (1.8,0) {Parallel composition of $\protocol$};
            \draw [draw=myParula07Red,fill=myParula07Red!50] (2,0.4) rectangle (4.3,-0.4);
            \draw [draw=myParula05Green,fill=myParula05Green!50] (4.3,0.4) rectangle (4.4,-0.4);
            
            \draw [draw=myParula07Red,fill=myParula07Red!50] (4.4,0.4) rectangle (4.7,-0.4);
            \draw [draw=myParula05Green,fill=myParula05Green!50] (4.7,0.4) rectangle (4.8,-0.4);
            
            \draw [draw=myParula07Red,fill=myParula07Red!50] (4.8,0.4) rectangle (5.1,-0.4);
            \draw [draw=myParula05Green,fill=myParula05Green!50] (5.1,0.4) rectangle (5.2,-0.4);
            
            \draw [draw=myParula07Red,fill=myParula07Red!50] (5.6,0.4) rectangle (5.9,-0.4);
            \draw [draw=myParula05Green,fill=myParula05Green!50] (5.9,0.4) rectangle (6.0,-0.4);
            
            \node at (5.4,0) {...};
            
            \draw [|-|] (2,-0.7) -- (4.4,-0.7) node [midway,below,yshift=-1pt] {Primary chain};
            \draw [|-|] (4.4,-0.7) -- (6.0,-0.7) node [midway,below,yshift=-1pt] {Secondary chains};
            
            \draw [black!20,line width=0.4pt] (2,0.4) rectangle (6,-0.4);

        \end{scope}

    \end{tikzpicture}
    \vspace{-1.5em}%
    \caption{Worst-case throughput and bandwidth consumption, as a fraction of the total bandwidth. Green portions represent bandwidth consumption that contributes to throughput, while red portions represent bandwidth consumption that is caused by the adversary and may not contribute to throughput (\eg, empty/invalid blocks, spamming).}
    \label{fig:parallelchains-network-util}
\end{figure}
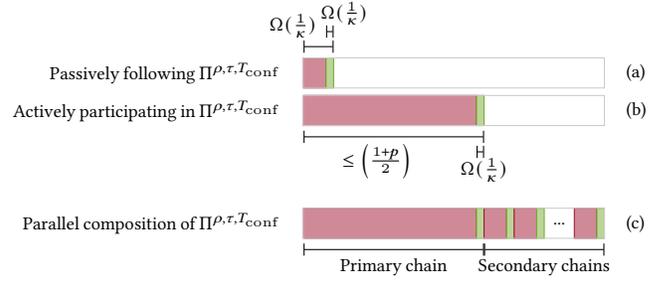

\DeclareRobustCommand{\svdots}{%
  \vcenter{%
    \offinterlineskip
    \hbox{.}
    \vskip0.25\normalbaselineskip
    \hbox{.}
    \vskip0.25\normalbaselineskip
    \hbox{.}%
  }%
}

In what follows, we use the freshest block download rule
as our running example, but the results carry over
analogously to other download rules analyzed 
using our unified framework, such as those in Section~\ref{sec:other_download_rules}.
From
Corollary~\ref{cor:security-freshest-block}, we parameterize $\protocol$ with the freshest block download rule with $\tau = \Omega(\kappa)$ for security, so the protocol gets slower as the security parameter increases.
A similar slowdown is also observed in the analysis in \cite{near-optimal-thruput}.
Thus, the throughput of $\protocol$ decreases with increasing security parameter. 
Indeed, we show in Section~\ref{sec:parallelchains-throughput} that the worst-case throughput of $\protocol$ is lower bounded by $\frac{2\pu-p}{\tau}=\frac{1}{\Omega(\kappa)}$.

The slow block production rate also means that \emph{passively following} the confirmed blocks
of a chain 
only requires downloading up to $\frac{p}{\tau} = \frac{1}{\Omega(\kappa)}$ blocks per second because the secure protocol $\protocol$ has already achieved consensus on these blocks (see Figure~\ref{fig:parallelchains-network-util}(a)).
In fact, the ratio between throughput and the bandwidth required to download the confirmed blocks is
the chain quality (fraction of honest blocks in the chain). This fraction, $\frac{2\pu-p}{p}>0$ is independent of the security parameter $\kappa$.
This suggests to invoke the idea of Parallel Chains \cite{parallel,near-optimal-thruput}:
fill the available bandwidth using multiple instances of the slow LC protocol in parallel and combine the transactions of all instances into a single ledger. By increasing the number of chains, one can compensate for the decreasing throughput of the individual chains 
as $\kappa$ is increased.

However, following the confirmed chains alone is not enough to achieve consensus on all these chains.
Note that the bandwidth consumption of a node \emph{actively participating} in $\protocol$ may be higher than what is required to download only the confirmed chain, due to spamming attacks.
By spending this additional bandwidth, the nodes actively  participating in $\protocol$ make the protocol secure, which is what allows other nodes to passively follow and download the confirmed chain with little bandwidth consumption.
However, even under spamming attacks, we show in Section~\ref{sec:parallelchains-throughput} that the worst-case bandwidth consumption is only a little more than half the available bandwidth $\bwtime$ (shown in Figure~\ref{fig:parallelchains-network-util}(b)).
This leaves nearly half the bandwidth available for a node participating in $\protocol$ to download the confirmed portions of other parallel chains.
This still allows us to increase the number of parallel chains to occupy the remaining bandwidth (shown in Figure~\ref{fig:parallelchains-network-util}(c)).
So, we modify the parallel chains construction from \cite{parallel,near-optimal-thruput} as described in the following section.

\begin{figure}
    \centering
    \begin{tikzpicture}[x=1.5cm,y=0.72cm,node distance=0.5cm]
        \small
            
        \draw [Latex-] (5.5,0) -- (0.25,0) node [above right] {\emph{Time}};
        \draw [] (3,0) ++(0,-0.2) -- ++(0,0.4) node [above] {$t-\Tconf$};
        \draw [] (5.1,0) ++(0,-0.2) -- ++(0,0.4) node [above] {$t$};
        
        \begin{scope}[blockchain,y=0.72cm,shift={(1,-1)}]

            \coordinate (G) at (0,0);
            \node [anchor=east,xshift=-0.5em] at (G) {$\protocol_1$};
            
            \node (p0) at (0.7629476504106385,0) [block] {};
            \node (p1) at (1.2946805533357357,0) [block] {};
            \node (p2) at (1.6882486139982926,0) [block] {};
            \node (p3) at (2.4032266319526814,0) [block] {};
            \node (p4) at (3.588227495694008,0) [block] {};
            
            \draw [link] (p0) -- (G);
            \draw [link] (p1) -- (p0);
            \draw [link] (p2) -- (p1);
            \draw [link] (p3) -- (p2);
            \draw [link] (p4) -- (p3);

        \end{scope}
        
        \begin{scope}[blockchain,y=0.72cm,shift={(1,-2)}]

            \coordinate (G) at (0,0);
            \node [anchor=east,xshift=-0.5em] at (G) {$\protocol_2$};
            
            \node (p0) at (1.0861971621686313,0) [block] {};
            \node (p1) at (1.734119126298118,0) [block] {};
            \node (p2) at (2.9927868283050643,0) [block,densely dotted] {};
            \node (p3) at (3.542892428120681,0) [block,densely dotted] {};
            \node (p4) at (3.882358246801729,0) [block,densely dotted] {};
            
            \draw [link] (p0) -- (G);
            \draw [link] (p1) -- (p0);
            \draw [link,densely dotted] (p2) -- (p1);
            \draw [link,densely dotted] (p3) -- (p2);
            \draw [link,densely dotted] (p4) -- (p3);

        \end{scope}
        
        \begin{scope}[blockchain,y=0.72cm,shift={(1,-3)}]

            \coordinate (G) at (0,0);
            \node [anchor=east,xshift=-0.5em] at (G) {$\protocol_3$};
            
            \node (p0) at (0.6074087323687713,0) [block] {};
            \node (p1) at (1.2380381767835142,0) [block] {};
            \node (p2) at (2.270611025448454,0) [block,densely dotted] {};
            \node (p3) at (3.467462861,0) [block,densely dotted] {};
            \node (p4) at (3.930965663355087,0) [block,densely dotted] {};
            
            \draw [link] (p0) -- (G);
            \draw [link] (p1) -- (p0);
            \draw [link,densely dotted] (p2) -- (p1);
            \draw [link,densely dotted] (p3) -- (p2);
            \draw [link,densely dotted] (p4) -- (p3);

        \end{scope}
        
        \begin{scope}[blockchain,y=0.72cm,shift={(1,-4)}]

            \coordinate (G) at (0,0);
            
            \node (pnil2) at (-0.5,0) {$\svdots$};
            \node (p0) at (1,0) {$\svdots$};
            \node (p2) at (3,0) {$\svdots$};

        \end{scope}
        
        \begin{scope}[blockchain,y=0.72cm,shift={(1,-5)}]

            \coordinate (G) at (0,0);
            \node [anchor=east,xshift=-0.5em] at (G) {$\protocol_m$};
            
            \node (p0) at (1.2612717801989293,0) [block] {};
            \node (p1) at (2.211242299416104,0) [block,densely dotted] {};
            \node (p2) at (2.6336282713016852,0) [block,densely dotted] {};
            \node (p3) at (3.2346486828890955,0) [block,densely dotted] {};
            \node (p4) at (3.749305129330001,0) [block,densely dotted] {};
            
            \draw [link] (p0) -- (G);
            \draw [link,densely dotted] (p1) -- (p0);
            \draw [link,densely dotted] (p2) -- (p1);
            \draw [link,densely dotted] (p3) -- (p2);
            \draw [link,densely dotted] (p4) -- (p3);

        \end{scope}
        
        \draw [dashed,myParula03Yellow] (0,-0.6) rectangle (5.25,-1.4);
        \draw [dashed,myParula03Yellow] (0,-1.6) rectangle (5.25,-5.4);
        \node [myParula03Yellow] at (1.5,-0.3) {\textsc{Primary Chain}};
        \node [myParula03Yellow] at (1.5,-5.7) {\textsc{Secondary Chains}};
        
        \draw [dotted,myParula04Purple,thick] (3,0) -- (3,-6);
        \draw [dotted,myParula04Purple,thick] (5.1,0) -- (5.1,-6);
        \draw [latex-latex,myParula04Purple] (0.25,-6.1) -- (3,-6.1) node [midway,below,align=center,anchor=north] {Follow secondary chains'\\confirmed blocks};
        \draw [latex-latex,myParula04Purple] (3,-6.1) -- (5.1,-6.1) node [midway,below,align=center,anchor=north] {Follow primary chain's\\freshest blocks};
    \end{tikzpicture}
    \vspace{-0.5em}%
    \caption{In the parallel chains construction using $\protocol$, each node is assigned
    one primary chain; the other $(m-1)$ chains are secondary.
    Nodes participate actively in their primary chain using,
    for example,
    the
    freshest block download rule, and follow their secondary chains
    passively by downloading confirmed blocks only.}
    \label{fig:parallelchains}
\end{figure}
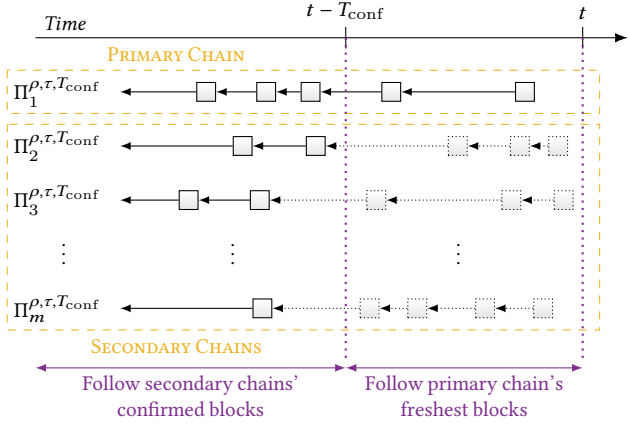

\subsection{Parallel Chains Construction}
\label{sec:parallelchains-parallel}

The protocol consists of $m$ parallel instances of $\protocol$ (see Figure~\ref{fig:parallelchains}).
For simplicity, assume that at genesis (and after the adversary has chosen
which nodes to corrupt), stakeholders are randomly partitioned into $m$
equally sized sets, and the nodes of each set get assigned
a particular instance of $\protocol$ as their \emph{primary chain}.
Nodes are responsible for maintaining consensus on their primary chain.
For this purpose, they 
download blocks as per a secure download rule 
and propose blocks on their primary chain as described in $\protocol$.
The remaining $(m-1)$ instances of $\protocol$ that are not a node's primary chain
are considered its \emph{secondary chains}.
Nodes do not participate actively in consensus building on their secondary chains,
but only download the confirmed blocks from those chains, as determined by the $\Tconf$-deep LC confirmation rule based on the block headers.
Transactions from the confirmed portion of all the chains are
first ordered
by their time slots and then by the index of the protocol instance they appear in, to then be merged into a single output ledger.
Moreover, every transaction can be included only in a single $\protocol$ instance
determined, \eg, based on the transaction's hash or sender address,
so as to avoid duplicating transactions across different $\protocol$ instances.
See Appendix~\ref{sec:appendix-parallel-chains-pseudocode} for pseudocode for the above parallel chains construction.

Each instance of $\protocol$ is secure if at most $\beta$ fraction of nodes for whom this instance is the primary chain are corrupt, and the parameters $\rho,\tau,\Tconf$ satisfy the constraints in Theorem~\ref{thm:security}.
Note that if the number of stakeholders assigned to each primary chain
is large, then 
the adversarial power in each instance of $\protocol$ is very likely
close to the
overall adversarial power, 
rendering the construction secure
against non-adaptive adversaries that corrupt at most $\beta$ fraction of the nodes.
See Appendix~\ref{sec:appendix-proofs-parallel} for a more detailed security analysis.

\subsection{Throughput and Bandwidth Consumption}%
\label{sec:parallelchains-throughput}

To quantify the throughput of $\protocol$, we first note that the longest chain in any honest node's view grows at least at the rate of uniquely successful slots, $\pu$ blocks per slot (Lemma~\ref{prop:chain_growth}).
Moreover, we can lower bound the chain quality, \ie, the fraction of blocks in the blockchain in any honest node's view, which are proposed by honest nodes. All blocks proposed by honest nodes will contain distinct and valid transactions. Therefore, the chain quality along with the chain growth rate give a lower bound on the throughput.

\begin{lemma}{(Throughput)}
\label{lem:throughput}
    There exists a constant $T_1$ such that for any honest node $i$ and time slots $t_1,t_2\geq t_1+T$ with $T\geq T_1$, $\dC_i(t_2) \setminus \dC_i(t_1)$ contains at least $\theta T(1-\epsilon_4)$ blocks proposed by honest nodes, with probability at least $1-\exp(-\alpha_4 T)$, where $\theta = 2\pu-p$.
\end{lemma}

From Lemma~\ref{lem:throughput}, the throughput of each chain is at least $\TP_1 = \frac{\theta}{\tau}$ blocks per second.\footnote{For simplicity, we consider a constant number of transactions in each block. Hence, this directly translates to throughput in transactions per second.} Note that this lower bound holds under the worst-case adversary strategy.

Next, we calculate the bandwidth consumption of passively following the confirmed blocks of a secondary chain.
Due to the security of $\protocol$
run by the nodes for whom the corresponding chain is primary,
the confirmed chain contains only valid available blocks and can be downloaded by spending little bandwidth without any interference from spamming blocks.

\begin{lemma}{(Passive Bandwidth Consumption)}
\label{lem:passive_bandwidth}
    There exists a constant $T_2$ such that for any honest nodes $i,i'$ and time slots $t_1, t_2\geq t_1+T$ such that $T\geq T_2$, $\LOG{i'}{t_2} \setminus \LOG{i}{t_1}$ contains at most $\bwpassive T(1+\epsilon_5)$ blocks, with probability at least $1-\exp(-\alpha_5 T)$, where $\bwpassive = p$.
\end{lemma}

Finally, we analyze the worst-case bandwidth consumption of active nodes in $\protocol$. 
As per the freshest block download rule (see Algorithm~\ref{algo:freshest-block-rule},
lines~\ref{loc:pseudocode-pos-lc-empty-chain}ff.), once all blocks proposed in the most recent non-empty time slot have been downloaded, the downloading node stays idle (because then $\calC=\bot$ in Algorithm~\ref{algo:pseudocode-pos-lc}, line~\ref{loc:pseudocode-pos-lc-empty-chain}). Since in every uniquely successful slot, each node downloads the freshest block within one slot (Lemma~\ref{lem:loner_download}), the node thereafter remains idle until the next block proposal. This gives a simple lower bound on the worst-case fraction of time a node's bandwidth consumption is idle.
(See Figure~\ref{fig:bandwidth-trace} for a matching observation in our experiments.)

\begin{lemma}{(Active Bandwidth Consumption)}
\label{lem:active_bandwidth}
    There exists a constant $T_3$ such that for any honest node $i$ and time slots $t_1,t_2\geq t_1+T$ with $T\geq T_3$, node $i$ does not download any blocks for at least $\idleslots T \tau (1-\epsilon_6)$ time during the interval of time slots $(t_1,t_2]$, with probability at least $1-\exp(-\alpha_6 T)$, where $\idleslots = \frac{\pu(1-p)}{p} \geq \frac{1-p}{2}$.
\end{lemma}

Lemmas~\ref{lem:throughput}, \ref{lem:passive_bandwidth} and \ref{lem:active_bandwidth} are proved in Appendix~\ref{sec:app-tp-bw-util-proofs}.
Lemma~\ref{lem:active_bandwidth} implies that a bandwidth of at least $\idleslots \cdot \bwtime$ remains unutilized by each node's primary chain. From Lemma~\ref{lem:passive_bandwidth}, each node needs to download on average $\bwpassive$ blocks per slot, or $\frac{\bwpassive}{\tau}$ blocks per second, to follow the confirmed blocks of one of the secondary chains. This allows each node to follow $m-1 = \frac{\idleslots}{\bwpassive} \bwtime\tau$ number of secondary chains.
Therefore the $m$ parallel chains have an aggregate throughput of
\begin{IEEEeqnarray}{rCl}
    \TP_m = m\,\TP_1 &=& \left(1 + \frac{\idleslots}{\bwpassive} C\tau\right) \frac{\theta}{\tau}   %
    \geq \frac{(1-p)(2\pu-p)}{2p} \, C   \nonumber\\
    &=& \frac{(1-p)\epsilon_1}{2} \, C \text{ blocks per second}
\end{IEEEeqnarray}
using $\SecurityCondition$ from Theorem~\ref{thm:security}. 
Therefore,
the aggregate throughput of the parallel chains remains within a constant fraction of the optimal throughput which is the 
bandwidth of $\bwtime$ blocks per second.
This is true
even 
if the 
number of secondary chains
is parameterized
so that the protocol
produces
an average load of only a certain
fraction of the 
bandwidth
left over by the primary chain,
so as to 
bound queuing delays due to fluctuations in bandwidth utilization.

Notice that the throughput and passive bandwidth consumption of the protocol do not change with the download rule. 
With 
`equivocation avoidance' and `blocklisting',
by 
doubling $\bwslot$ such that $\grpromemptycustom{\bwslot/2}$ holds (thereby roughly doubling the time slot duration), and ensuring that honest nodes do not download more than $\bwslot/2$ blocks in any slot, the active bandwidth utilization is explicitly bounded by half the available bandwidth. Thus, the parallel chains construction with these download rules also behaves similarly.

The worst-case throughput of a single chain and that of the parallel construction are limited by the chain quality factor $\epsilon_1 = \frac{2\pu-p}{p}$ due to the possibility of selfish mining attacks \cite{selfishmining}. Using the Conflux inclusion rule from \cite{conflux} (which is also employed in \cite{near-optimal-thruput}), this factor can be improved to $\frac{\pu}{p}$ 
(which does not vanish as we push the resilience $\beta$ closer to $1/2$).
In this rule, each block includes pointers to blocks that are not in its prefix, in order to include them in the ledger.
To adapt this rule to bandwidth constrained networks,
we modified it to ensure
that only one block from each block production opportunity is pointed to and the number of pointers in each block is limited yet enough to include honest blocks. The details of this construction are in Appendix~\ref{sec:appendix-conflux}.

Finally, in a comparison with \cite{near-optimal-thruput}, both works show a parallel chains construction that achieves throughput up to a constant fraction of the network capacity. However, our work proves this under worst-case adversarial strategies (including, \emph{inter alia}, equivocation-based spamming), while \cite{near-optimal-thruput} proves security only 
for adversaries that do not aggravate network congestion
so much that a delay upper bound is violated.
On the other hand, the security of our construction requires static corruption and honest majority among nodes in each chain (as each nodes performs consensus on one chain), whereas \cite{near-optimal-thruput} works under a global honest majority assumption (as each node participates in all chains).

%% file: appendix_pseudocode.tex
\section{Reference Algorithms}

\subsection{Helper Functions for Algorithms~\ref{algo:pseudocode-pos-lc},~\ref{algo:pseudocode-Ftree},~\ref{algo:longest-header-chain-equivoc-avoid-rule}}
\label{sec:appendix-helperfunctions-pseudocode}

\begin{itemize}
    \item $\operatorname{Hash}(\txs)$:
    
        \noindent Cryptographic hash function to produce
        a binding commitment to $\txs$
        (modelled as a random oracle)
        
    \item $\Chain' \preceq \Chain$:
    
        \noindent Relation describing that $\Chain'$ is a prefix of $\Chain$
        
    \item $\Chain\|\Chain'$:
    
        \noindent Concatenation of $\Chain$ and $\Chain'$
        
    \item $\operatorname{prefixChainsOf}(\Chain)$:
    
        \noindent Set of prefixes of $\Chain$
    
    \item $\operatorname{longestChain}(\Tree)$:
    
        \noindent Determine longest chain among set $\Tree$ of chains. Ties are broken by the adversary.
        
    \item $\Chain\trunc{\Tconf}$:
    
        \noindent Prefix of chain $\Chain$ consisting of all blocks with time slots up to $\Tconf$ less than the current time slot
        
    \item $\operatorname{txsAreSemanticallyValidWrtPrefixesOf}(\Chain, \txs)$:
    
        \noindent Verifies for each transaction in $\txs$ that
        the transaction is semantically valid
        with respect to
        and properly authorized by
        the owner of the underlying assets
        as determined by
        the transaction's prefix
        in the ledger resulting from appending
        $\txs$ to the transactions as ordered
        in $\Chain$
        (assumes that content of all blocks in $\Chain$ is known to the node)
        
    \item $\operatorname{newBlock}(\mathsf{time}\colon t,\mathsf{node}\colon P,\mathsf{txsHash}\colon \operatorname{Hash}(\txs))$:
    
        \noindent Produces a new block header with
        the given parameters
    
    \item $\operatorname{oneLeafPerProductionOpportunity}(\HeaderTree)$:
    
        \noindent Filter header tree $\HeaderTree$ to keep at most one leaf per block production opportunity, \ie, per $(\mathsf{node},\mathsf{time})$ pair (equivocation avoidance; ties broken adversarially)
\end{itemize}

\subsection{Environment $\Env$}
\label{sec:appendix-environment}

The environment $\Env$ initializes $N$ nodes and lets $\Adv$ corrupt up to $\beta N$ nodes at the beginning of the execution. Corrupted nodes are controlled by the adversary. Honest nodes run $\protocol$.
The environment maintains a mapping $\TxsMap$ from block headers to the block content (transactions). This mapping is referred to as the `idealized repository' in Section~\ref{sec:modelprotocol}.
$\Env$ also maintains for each node a queue of pending block headers
to be delivered after a delay determined by the adversary (at most $\DeltaHeader$).
Honest nodes and the adversary interact with $\Env$ via the following functions:
\begin{itemize}
    \item $\Env.\Call{broadcastHeaderChain}{\Chain}$:
    
        \noindent If called by an honest node, $\Env$ sends header chain $\Chain$ to $\Adv$. 
        Then, for each honest node $P$, on receiving $\Call{deliver}{\Chain,P}$ from $\Adv$, or when $\DeltaHeader$ time has passed since $\Chain$
        was handed to $\Env$ for broadcasting, $\Env$ triggers $P.\Call{receivedHeaderChain}{\Chain}$.
        
    \item $\Env.\Call{uploadContent}{\Chain, \txs}$:
    
        \noindent $\Env$ stores a mapping from the header chain $\Chain$ to the content $\txs$ of its last block by setting $\TxsMap[\Chain] = \txs$.
        $\Env$ only stores the content $\txs$
        if $\mathrm{Hash}(\txs) = \Chain.\mathsf{txsHash}$.
        
    \item $\Env.\Call{requestContent}{\Chain}$:
    
        \noindent If $\TxsMap[\Chain]$ is set, then let $\txs = \TxsMap[\Chain]$ (if not, $\Env$ ignores the request).
        On receiving this call from an honest node $P$ in a time slot $t$,
        if $\Env$ has triggered $P.\Call{receivedContent}{.}$ less than $\bwslot$ times in slot $t$, then $\Env$ triggers  $P.\Call{receivedContent}{\Chain, \txs}$.
        On receiving this call from $\Adv$, $\Env$ sends $(\Chain, \txs)$ to $\Adv$.
        
    \item $\Env.\Call{receivePendingTxsSemanticallyValidWrt}{\Chain}$:
    
        \noindent $\Env$ generates a set of pending
        transactions that are not included in
        the block contents of
        but semantically valid (see Appendix~\ref{sec:appendix-helperfunctions-pseudocode}) with respect to $\Chain$, and returns them.
        
    \item $\Env.\Call{outputLedger}{\Chain}$:
    
        \noindent On receiving this call from a node $P$,
        $\Env$ records $\Chain$ as $P$'s ledger
        to be externalized. This constitutes $\LOG{i}{t}$,
        for which consistency and liveness are
        required for a secure consensus protocol.
\end{itemize}

%% file: appendix_experiments.tex
\section{Supplemental Experimental Material}

\subsection{Experimental Setup Details for Figure~\ref{fig:networkdelayblockrate}}
\label{sec:appendix-experiment-details-fig1}

For this experiment, we start 17 Cardano nodes in 17 AWS data centers across the globe and connect them into a fully-connected graph. We point out that the Cardano block fetch logic includes an optimization to only download blocks that have larger heights than the locally-adopted longest chain.
As a result, a node may not eventually download every block whose header it sees.
To demonstrate network congestion in the absence of a suitable download rule, 
we modify the code to disable this optimization and ensure that every node eventually downloads all blocks.
To show congestion, we configure a variable number ($N$) of nodes to all mine blocks at the beginning of the same slot, and report the time for all 17 nodes to download all $N$ blocks.

\subsection{Chain Growth with Larger Blocks}
\label{sec:vary-block-size}
In this experiment, we look at the robustness of the `download towards the freshest block' rule when we increase the block size. The topology is the same as previous experiments, but the in-flight cap is fixed to $1$. Figure~\ref{fig:size-grow} shows that this rule maintains the chain growth rate, 
despite the increasing network load.

\input{eval-5-higher-block-growth}

%% file: eval-5-higher-block-growth.tex
\begin{figure}[t]
    \centering
    \begin{tikzpicture}
        \small
        \begin{axis}[
            mysimpleplot,
            xlabel={Block size [$\mathrm{KB}$]},
            ylabel={Honest chain\\growth rate [$\mathrm{s}^{-1}$]},
            legend columns=2,
            xmin=0.5, xmax=16.5,
            ymin=0, ymax=0.05,
            height=0.4\linewidth,
            width=\linewidth,
            yticklabel style={
                /pgf/number format/fixed,
                /pgf/number format/precision=2
            },
            scaled y ticks=false,
        ]
        
        \addplot [myparula11] table [x=kb,y=rate] {experiments/larger_blocks/growth.txt};
        
        \end{axis}
    \end{tikzpicture}%
    \caption{Honest chain growth rate under spamming attack when using different block sizes and the download freshest block rule. Despite the increasing network load (through the increasing block size),
    there is no performance deterioration when
    downloading the freshest block.}
    \label{fig:size-grow}
\end{figure}
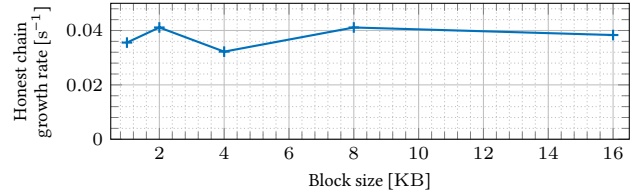

%% file: appendix_parallel.tex
\section{Parallel Chains Pseudocode}
\label{sec:appendix-parallel-chains-pseudocode}

Algorithm~\ref{algo:pseudocode-parallel-pos-lc} gives pseudocode for the parallel chains construction using our PoS LC protocol parameterized with a download rule. Note the following main differences with respect to Algorithm~\ref{algo:pseudocode-pos-lc}.
Upon initialization, each node is assigned a primary protocol instance index by the functionality $\Fparallel$.
Each node maintains a separate header tree and downloaded chain for each index.
While scheduling content downloads, primary instance blocks get the highest priority, with the same download rule that parameterizes $\protocol$. If there are no blocks left to be downloaded in the primary instance, the node picks among the confirmed longest chains of all secondary instances, the block with the oldest time slot with unknown content.
Downloading the block with the oldest time slot allows the node to construct the ledger quickly, although this priority rule does not play a critical role in the consensus security. In line~\ref{loc:pseudocode-parallel-pos-lc-construct-ledger}, the ledger is constructed by ordering the confirmed blocks of all the instances first
by their time slots and then by the index of the protocol instance they appear in.
The functionality $\Fparallel$ (Algorithm~\ref{algo:pseudocode-Ftree-parallel}) assigns the primary chain index for each node by uniformly and randomly partitioning the set of nodes across the $m$ chains.
This can be approximated in instantiations by each node selecting as its primary chain index a hash of its public key modulo $m$.
Rather than by the transaction hash, another way to shard transactions is by distributing all accounts uniformly among the protocol instances, and requiring transactions in a particular instance to have both the source and destination accounts in the same instance. Transactions with the source and destination accounts in different instances would be split into two transactions, one which burns the funds in the source account and subsequently another one which recreates funds in the destination account (while showing a receipt of burn in the source chain), each transaction in its respective protocol instance (see \cite{powsidechains,possidechains,proofofburn} and references therein for background on this technique).
Such a solution allows validation of each transaction with respect to its prefix within the same instance at the time of block production (Algorithm~\ref{algo:pseudocode-parallel-pos-lc} line~\ref{loc:pseudocode-parallel-pos-lc-mainloop}), a property sometimes referred to as \emph{predictable validity}.
An important consequence of this is that there is no ``ledger sanitization'' procedure required while constructing the ledger out of the confirmed blocks. In other words, transactions once added to the chain cannot be invalidated in the ledger because they were validated with respect to their past state while proposing and forwarding the block.
Thus, every transaction contributes to throughput.

\input{parallel_pseudocodes_modular}

\subsection{Additional Helper Functions for Algorithm~\ref{algo:pseudocode-parallel-pos-lc} (see also Appendix~\ref{sec:appendix-helperfunctions-pseudocode})}
\label{sec:appendix-helperfunctions-pseudocode-parallel}

\begin{itemize}

    \item 
        $\operatorname{sortBySlotThenIndex}(\mathcal{S})$:
        
        \noindent Arranges the chains in the set $\mathcal{S}$ in increasing order of time slots of their tip. Chains with the same time slot from different protocol instances are arranged in increasing order of the index of their protocol instance.

    \item $\operatorname{longestChain}(\Tree)$:
    
        \noindent Computes the longest chain in the tree $\Tree$, i.e. computes $\argmax_{\Chain \in \Tree} \len{\Chain}$.
        
    \item $\Env.\Call{receivePendingTxsSemanticallyValidWrt}{\Chain}$:
    
        \noindent Same as in the case of a single chain, but only includes transactions for which the source account is defined in the same chain $\Chain$.
\end{itemize}

%% file: parallel_pseudocodes_modular.tex
\begin{algorithm}[t]
    \caption{Parallel Chains PoS LC consensus protocol $\protocolpc$ (helper functions: Appendix~\ref{sec:appendix-helperfunctions-pseudocode-parallel}, $\Fparallel$: Algorithm~\ref{algo:pseudocode-Ftree-parallel}, $\protocol$: Algorithm~\ref{algo:pseudocode-pos-lc})}
    \label{algo:pseudocode-parallel-pos-lc}
    \begin{algorithmic}[1]
        \small
        \On{$\Call{init}{\mathsf{genesisHeaderChain}, \mathsf{genesisTxs}}$}
            \State $\primindex \gets \Fparallel.\Call{primaryChainIndex}{ }$
            \For{$\index=1,...,m$}
                \State $\Pi_{\index} \gets \operatorname{new}\, \protocol$ \Comment{Initialize $m$ instances of $\protocol$}
                \State $\Pi_{\index}.\Call{init}{\mathsf{genesisHeaderChain}, \mathsf{genesisTxs}}$
            \EndFor
        \EndOn
        \On{$\Call{receivedHeaderChain}{\index,\Chain}$}
            \label{loc:pseudocode-parallel-pos-lc-receiveheader}
            \State $\Pi_{\index}.\Call{receivedHeaderChain}{\Chain}$
        \EndOn
        \On{$\Call{receivedContent}{\index,\Chain,\txs}$}
            \label{loc:pseudocode-parallel-pos-lc-receivecontent}
                \State $\Pi_{\index}.\Call{receivedContent}{\Chain,\txs}$

        \EndOn
        \On{$\Call{scheduleContentDownload}{\null}$}
            \label{loc:pseudocode-parallel-pos-lc-downloadrule}
                \Comment{Called when download idle}
            \State $\Pi_{\primindex}.\Call{scheduleContentDownload}{\null}$ \Comment{First priority for primary}
            \If{no content requested by $\Pi_{\primindex}$}
            \LineComment{Download first missing block along the confirmed portion of the longest header chains in the secondary instances.}
                \State $\mathcal{S} \gets \{ \operatorname{longestChain}(\Pi_{\index}.\HeaderTree)\trunc{\Tconf} \mid \index \in \{1,...,m\}\setminus\{\primindex\}\}$
                \State $\Chain \gets \argmin_{\Chain'' \preceq \Chain' \in \mathcal{S} \colon \TxsMap[\Chain'']=\NIL} \Chain''.\mathsf{time}$
                \State $\Env.\Call{requestContent}{\Chain}$
            \EndIf
        \EndOn
        \For{time slots $t \gets 1,...,\Th$ of duration $\tau$}
            \LineComment{Only include valid txs whose accounts belong to the primary chain}
            \State $\txs \gets \Env.\Call{receivePendingTxsSemanticallyValidWrt}{\Pi_{\primindex}.\DownloadedChain}$%
            \label{loc:pseudocode-parallel-pos-lc-mainloop}
            \LineComment{Check eligibility to produce a new block, and if so do so, see Algorithm~\ref{algo:pseudocode-Ftree-parallel}}
            \If{$\Chain' \neq \bot$ \textbf{with} $\Chain' \gets \Fparallel.\Call{extend}{\primindex,t,\Pi_{\primindex}.\DownloadedChain,\txs}$}
                \State $\Env.\Call{uploadContent}{\primindex,\Chain', \txs}$
                \State $\Env.\Call{broadcastHeaderChain}{\primindex,\Chain'}$
            \EndIf
            \While{end of current time slot $t$ not reached}
                \State $\Call{scheduleContentDownload}{ }$
            \EndWhile
            \label{loc:pseudocode-parallel-pos-lc-construct-ledger}
            \LineComment{Find the maximum time slot of all downloaded and confirmed chains}
            \State $\maxt \gets \max \{ t \mid \Pi_{\index}.\dC\trunc{\Tconf}.\mathsf{time} \geq t, \index\in\{1,...,m\}\}$
            \LineComment{Arrange confirmed and downloaded chains in increasing order of time slots, then chain index}
            \State $\mathsf{LOG} \gets \operatorname{sortBySlotThenIndex}(\{\Chain \mid \Chain \preceq \Pi_{\index}.\DownloadedChain\trunc{\Tconf},\allowbreak \Chain.\mathsf{time}\leq\maxt, \index\in\{1,...,m\}\})$
            \State $\Env.\Call{outputLedger}{\mathsf{LOG}}$
        \EndFor
    \end{algorithmic}
\end{algorithm}

\begin{algorithm}[t]
    \caption{Idealized functionality $\Fparallel$ for parallel chains
    (see also
    $\Ftree$: Algorithm~\ref{algo:pseudocode-Ftree})}
    \label{algo:pseudocode-Ftree-parallel}
    \begin{algorithmic}[1]
        \small
        \On{$\Call{init}{\mathsf{genesisHeaderChain},\mathsf{numParties}}$}
            \State $\mathcal{P}_1,...,\mathcal{P}_m \gets \text{random equi-partition of $\{1,...,\mathsf{numParties}\}$}$
            \For{$\index=1,...,m$}
                \For{$P\in\mathcal{P}_\index$}
                    \State $\primindex[P] \gets \index$
                \EndFor
                \State $\mathcal F_{\index} \gets \operatorname{new}\, \Ftree$ \Comment{Initialize $m$ instances of $\Ftree$}
                \State $\mathcal{F}_{\index}.\Call{init}{\mathsf{genesisHeaderChain},\mathsf{numParties}/m}$
            \EndFor
        \EndOn
        \On{$\Call{primaryChainIndex}{ }$} \textbf{from} party $P$
            \State \Return{$\primindex[P]$}
        \EndOn
        \On{$\Call{extend}{\index,t',\Chain,\txs}$ \textbf{from} party $P$ \textbf{at} time slot $t$}
            \If{$\primindex[P] \neq \index$}
                \State \Return{$\bot$}
            \EndIf
            \State \Return{$\mathcal{F}_{\index}.\Call{extend}{t',\Chain,\txs}$}
        \EndOn
    \end{algorithmic}
\end{algorithm}

%% file: appendix_proofs.tex
\section{Proof Details}
\label{sec:appendix-proofs}

\subsection{Proof of Theorem~\ref{thm:security}}
\label{sec:appendix-proof-security-theorem}

\begin{definition}
\label{def:pivot}
A \emph{pivot} is a slot $t$ such that
\begin{IEEEeqnarray}{rCl}
    \label{eq:pivot}
    \forall (r,s] \ni t\colon\  \left(\Lint{r,s} > \Nint{r,s} \right) \lor \left(\Nint{r,s} = 0\right).
\end{IEEEeqnarray}
The predicate $\TPivot{t}$ is true iff $t$ is a pivot.
A slot $t$ is a unique pivot slot iff $\TPivot{t}\land \Unique{t}$.
\end{definition}

\begin{definition}
For an execution $\exec$,
$\FrequentPivots(\exec)$ holds iff
\begin{IEEEeqnarray}{C}
    \forall t \leq \Th-\Tp \colon \exists t' \in (t, t+\Tp]\colon\  \TPivot{t'} \land \Unique{t'}.
\end{IEEEeqnarray}
\end{definition}

\begin{lemma}
\label{lem:combinatorial}
For all $\bwslot, \Tp \in \mathbb{N}, \rho \in \mathbb{R}^{+}$,
executions $\exec$
and download rules $\dlrule$
such that
\begin{IEEEeqnarray}{c}
\FrequentPivots(\exec) \land \grpromempty(\exec, \dlrule) \IEEEeqnarraynumspace
\end{IEEEeqnarray}
holds,
if
$\tau = \DeltaHeader + \frac{\bwslot}{\bwtime}$ and
$\Tconf=\Tp$,
then
the protocol 
$\protocol$
with download rule $\dlrule$
satisfies safety
and liveness with $\Tlive = 2\Tp$
in $\exec$.
\end{lemma}

\begin{lemma}
\label{lem:frequent_pivots_prob_bound}
If $\SecurityCondition$ for some $\epsilon_1\in(0,1)$
and $\Tp = \frac{\Omega((\kappa + \ln\Th)^2)}{\constFrequentPivots p}$,
then
\begin{IEEEeqnarray}{C}
    \Prob{ \exec\colon \lnot \FrequentPivots(\exec) }
    \leq \negl(\kappa)
\IEEEeqnarraynumspace
\end{IEEEeqnarray}
where 
$\constFrequentPivots$
is a constant that depends on $\epsilon_1$ and $\rho$.
\end{lemma}

Lemma~\ref{lem:combinatorial} is proved in Appendix~\ref{sec:appendix-combinatorial} and Lemma~\ref{lem:frequent_pivots_prob_bound}
in Appendix~\ref{sec:analysis-probabilistic}.

\begin{proof}[Proof of Theorem~\ref{thm:security}]
Using Lemma~\ref{lem:combinatorial}, safety and liveness hold except with probability
\begin{IEEEeqnarray}{c}
    \Prob{\exec \colon \lnot\FrequentPivots(\exec) \lor \lnot\grpromempty(\exec,\dlrule)}.
    \IEEEeqnarraynumspace
\end{IEEEeqnarray}
This probability is negligible as per a union bound with Lemma~\ref{lem:frequent_pivots_prob_bound} and the assumption about the download rule.

\end{proof}

\subsection{Proof of Lemma~\ref{lem:combinatorial}}
\label{sec:appendix-combinatorial}

\begin{lemma}
\label{lem:fresh_block_base_case}
Suppose that 
for a download rule $\dlrule$ and execution $\exec$,
$\grpromempty(\exec,\dlrule)$ holds.
Let $t^*$ be a time slot such that $\TPivot{t^*} \land \Unique{t^*}$. Let $b^*$ be the block proposed in slot $t^*$. 
Then $b^*\in\dC_i(t)$ for all $i$ and all $t \geq t^*$. 
\end{lemma}

\begin{proof}
For contradiction, suppose that $s\geq t^*$ is the first slot such that $b^*\not\in\dC_i(s)$ for some $i$. Let $\calC'=\dC_i(s)$ such that $b^*\not\in\calC'$. 
Let $h'$ be the last block corresponding to a uniquely successful slot on $\calC'$.
Let $h'$ be proposed in the slot $r$. 
Clearly, $r\leq s$.

The block $h'$ extends $\dC_{i'}(r-1)$ for some $i'$ since honest nodes propose blocks on their longest downloaded chain. Since $h'\in\calC'$ and $b^*\notin\calC'$, this means that $b^*\notin\dC_{i'}(r-1)$. If $r>t^*$, this is a contradiction because we assumed that $s$ is the first slot such that $s \geq t^*$ and $b^*\not\in\dC_i(s)$ for some $i$. Since $\Unique{t^*}$, $r\neq t^*$. So, we conclude that $r < t^*$. All blocks in $\calC'$ extending $h'$ are from successful slots that are not uniquely successful, \ie, they are adversarial slots. So,
\begin{IEEEeqnarray}{C}
\label{eq:adv_chain}
    \len{\calC'} \leq \len{h'} + \Nint{r,s}
\end{IEEEeqnarray}

From Lemma~\ref{prop:chain_growth},
\begin{IEEEeqnarray}{C}
\label{eq:hon_chain}
    L_{\min}(s) \geq L_{\min}(r) + \Lint{r,s}.
\end{IEEEeqnarray}
Note that $L_{\min}(s)\leq L_i(s)$ $\forall i$ and $\len{\calC'}=L_i(s)$ for some $i$. 
Also note that $h'$ is from a uniquely successful slot $r$ and $\grpromempty$ holds, so $L_{\min}(r)\geq \len{h'}$. Using the above observations with \eqref{eq:adv_chain} and \eqref{eq:hon_chain}, we get
\begin{IEEEeqnarray}{C}
\label{eq:pivot_contradiction}
    \Lint{r,s} \leq \Nint{r,s}
\end{IEEEeqnarray}
where $r<t^*$ and $s\geq t^*$. Since $\TPivot{t^*}$, this is a contradiction.
\end{proof}

Lemma~\ref{lem:fresh_block_base_case} shows that the block from every unique pivot slot 
stays in all honest nodes' downloaded longest chains thereafter.
Therefore, under $\FrequentPivots$, every interval of $\Tp$ slots brings at least one 
such block.
To conclude with the proof of Lemma~\ref{lem:combinatorial}, one needs to show that the occurrence of 
such blocks
leads to safety and liveness. This is done 
in Lemma~\ref{lem:combinatorial}.

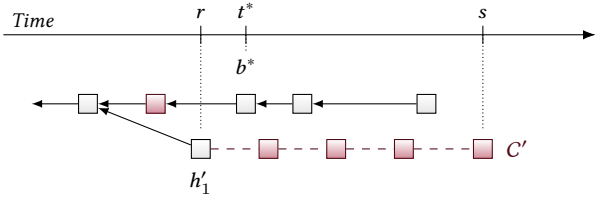
\begin{figure}%
    \centering
    \begin{tikzpicture}
        \small
        
        \begin{scope}[blockchain]

            \coordinate (G) at (0.5,0);
            
            \node (p0) at (1,0) [block] {};
            \node (padv) at (1.6,0) [block-adv2] {};
            \node (p1) at (2.4,0) [block] {};
            \node (p2) at (2.9,0) [block] {};
            \node (p3) at (4,0) [block] {};
           
            \node [yshift=0.5cm] (p1label) at (2.4,0) {$b^*$};

            \node (p1a1) at (2.0,-1) [block] {};
            \node (p1a2) at (2.6,-1) [block-adv2] {};
            \node (p1a23) at (3.2,-1) [block-adv2] {};
            \node (p1a3) at (3.8,-1) [block-adv2] {};
            \node (p1a4) at (4.5,-1) [block-adv2] {};
            
            \draw [link] (p0) -- (G);
            \draw [link] (padv) -- (p0);
            \draw [link] (p1) -- (padv);
            \draw [link] (p2) -- (p1);
            \draw [link] (p3) -- (p2);

            \draw [link] (p1a1) -- (p0);
            
            \draw [hiddenlink-adv2] (p1a1) -- (p1a2) -- (p1a23) -- (p1a3) -- (p1a4);
            
            \begin{scope}[yshift=-1em]
                \draw [Latex-] (5.5,2) -- (0.25,2) node [above right] {\emph{Time}};
                
                \draw [] (2.4,2) ++(0,-0.2) -- ++(0,0.4) node [above] {$t^*$};
                \draw [] (2.0,2) ++(0,-0.2) -- ++(0,0.4) node [above] {$r$};
                \draw [] (4.5,2) ++(0,-0.2) -- ++(0,0.4) node [above] {$s$};
                
                \draw [densely dotted] (2.4,2) -- (p1label);
                \draw [densely dotted,shorten >=0.5em] (2.0,2) -- (p1a1);
                \draw [densely dotted,shorten >=0.5em] (4.5,2) -- (p1a4);
            \end{scope}
            
            \node [yshift=-0.2cm, anchor=north,align=center] (p1a1label) at (p1a1) {$h'_1$};
            \node [xshift=0.2cm,anchor=west,align=left,label-adv2] (p1a4label) at (p1a4) {$\calC'$};

        \end{scope}
    \end{tikzpicture}
    \caption[]{An illustration of one example of the blocks and time slots defined in the proof of Lemma~\ref{lem:fresh_block_base_case}. The block $b^*$ is proposed in the unique pivot slot $t^*$. At the end of slot $s\geq t^*$, the chain $\calC'\not\ni b^*$ is the longest chain in some node's view. The last block from a uniquely successful slot in $\calC'$ is $h_1'$ proposed in the slot $r<t^*$.
    Red
    (\tikz[blockchain,baseline={([yshift=-0.6ex]current bounding box.center)}]{ \node [block-adv2] at (0,0) {}; })
    and gray
    (\tikz[blockchain,baseline={([yshift=-0.6ex]current bounding box.center)}]{ \node [block] at (0,0) {}; })
    blocks are proposed by adversarial and honest nodes, respectively.
    A red dashed link
    (\tikz[blockchain,baseline={([yshift=-.5ex]current bounding box.center)}]{ \draw [hiddenlink-adv2] (0,0) -- (0.25,0); })
    indicates that the block is withheld and released later.
    Note that in this example, $\Nint{r,s}=4>3=\Lint{r,s}$, which is in contradiction to $\TPivot{t^*}$.}
    \label{fig:deconf_examples}
\end{figure}

\begin{proof}[Proof of Lemma~\ref{lem:combinatorial}]
Let $\Tconf = \Tp$. 
First, we prove safety by contradiction. 
Suppose that for some honest nodes $\node{i},\node{j}$ and $t'\geq t$ that $\dC_i(t)\trunc{\Tconf}\not\preceq \dC_j(t')\trunc{\Tconf}$.
We can assume that $t\geq\Tp$ because otherwise $\dC_i(t)\trunc{\Tconf}=\emptyset$ and therefore $\dC_i(t)\trunc{\Tconf} \preceq \dC_j(t')\trunc{\Tconf}$ for all $t'$.

Consider all the uniquely successful slots $t_1,...,t_m \in (t-\Tp,t]$ with block $b_j$ proposed in slot $t_j$. 
Suppose that $b_j\in\dC_i(t)$ and $b_j\in\dC_j(t')$. Then $\dC_i(t)$ and $\dC_j(t')$ match up to $b_j$. Since $t_j>t-\Tp$, $\dC_i(t)\trunc{\Tconf} \preceq \dC_j(t')$. Also, $t'\geq t$, therefore $\dC_i(t)\trunc{\Tconf} \preceq \dC_j(t')\trunc{\Tconf}$ which is a contradiction to our assumption.
Therefore, for each $j=1,...,m$, either $b_j\not\in\dC_i(t)$ or $b_j\not\in\dC_j(t')$. This means that for all $j=1,...,m$, $b_j$ is not a \greatblock. Due to $\FewBlockOpps$ and Lemma~\ref{lem:fresh_block_base_case}, this also means that there are no unique pivot slots in the interval $(t-\Tp,t]$, which is a contradiction to $\FrequentPivots$.

We next prove liveness.
Assume a transaction $\mathsf{tx}$ is received by all honest nodes before time $t$. 
We know that there exists a unique pivot slot $t^*$ in the interval $(t,t+\Tp]$. The honest block $b^*$ from $t^*$ or its prefix must contain $\mathsf{tx}$ since $\mathsf{tx}$ is seen by all honest nodes at time $t < t^*$.
Moreover, $b^*$ is also a \greatblock,
\ie, $b^*\in\dC_i(t')$ for all honest nodes $\node{i}$ and $t'\geq t^*$.
Therefore, $\mathsf{tx}\in\LOG{i}{t'}$ for all $t'\geq t^*+\Tconf$, which is at most $t+2\Tp$.
\end{proof}

\subsection{Proof of Lemma~\ref{lem:frequent_pivots_prob_bound}}
\label{sec:analysis-probabilistic}

\subsubsection{Preliminaries}

\begin{definition}[Pivot condition]
The predicate $\PivotCondition_{(r,s]}$ holds iff $\Lint{r,s} > \Nint{r,s}$.
\end{definition}
Note that 
$\TPivot{t}$ holds iff
$\forall (r, s] \ni t$, $\PivotCondition_{(r,s]} \lor \left(\Nint{r,s} = 0\right)$ holds.

\begin{definition}[Weak Pivot]
Time slot $t$ 
satisfies $\WeakPivot{t}$ iff
\begin{IEEEeqnarray}{C}
    \forall (r, s] \ni t, s-r < w\colon\ 
    \PivotCondition_{(r,s]}  \lor \left(\Nint{r,s} = 0\right).
\IEEEeqnarraynumspace\end{IEEEeqnarray}
\end{definition}

\begin{proposition}
\label{prop:exp-decay-not-pivot-condition-some-long-interval}
If $\SecurityCondition$ for some $\epsilon_1\in(0,1)$,
\begin{IEEEeqnarray}{C}
    \forall (r,s] \colon\ 
    \Prob{ \lnot \PivotCondition_{(r,s]} } \leq 2\exp\left( - \constPivotRace  p (s-r) \right),
\IEEEeqnarraynumspace\end{IEEEeqnarray}
with $\constPivotRace = \eta \epsilon_1^2$
and $\eta=1/36$.
\end{proposition}

\begin{proof}

By a simple Chernoff bound for $\epsilon>0$, 
\begin{IEEEeqnarray}{C}
    \Prob{ \Bint{r,s} \geq p (s-r)(1+\epsilon)} \leq \exp\left(-\frac{\epsilon^2p(s-r)}{2+\epsilon}\right).
\end{IEEEeqnarray}

Also, by a Chernoff bound for $\epsilon \in (0,1)$,
\begin{IEEEeqnarray}{C}
    \Prob{ \Lint{r,s} \leq \pu (s-r)(1-\epsilon)} \leq \exp\left(-\frac{\epsilon^2 \pu(s-r)}{2}\right)
\end{IEEEeqnarray}

By choosing $\epsilon$ such that $\frac{1+\epsilon}{1-\epsilon} = 1+\epsilon_1$, we obtain that
\begin{IEEEeqnarray*}{rCl}
    \Lint{r,s} &>& \pu(s-r)(1-\epsilon) \\
    &=& \frac12 p(1+\epsilon_1) (s-r)(1-\epsilon) \\
    &=& \frac12 p (s-r) (1+\epsilon) > \frac12 \Bint{r,s} \\
    \implies \Lint{r,s} &>& \Nint{r,s},
\end{IEEEeqnarray*}
except with probability
\begin{IEEEeqnarray}{C}
    \exp\left(-\frac{\epsilon^2p(s-r)}{2+\epsilon}\right) + \exp\left(-\frac{\epsilon^2 \pu(s-r)}{2}\right)
\end{IEEEeqnarray}
From $\frac{1+\epsilon}{1-\epsilon} = 1+\epsilon_1$, we get $\epsilon = \frac{\epsilon_1}{\epsilon_1+2} \geq \frac{\epsilon_1}{3}$. Further using $\pu>\frac{p}{2}$, this probability is bounded by
\begin{IEEEeqnarray}{C}
    2\exp\left(\frac{\epsilon_1^2 p (s-r)}{36}\right)
\end{IEEEeqnarray}

\end{proof}

\begin{proposition}
\label{prop:exp-decay-not-pivot-condition-any-long-interval}
If $\SecurityCondition$,
then
for an execution horizon $\Th$
and $w > \frac{2 \ln(\sqrt2\Th)}{\constPivotRace  p}$,
\begin{IEEEeqnarray}{C}
    \Prob{ \exists (r,s], s-r \geq w\colon \lnot \PivotCondition_{(r,s]} } \nonumber\\ 
    \leq 2\Th^2 \exp\left( - \constPivotRace pw \right).
\IEEEeqnarraynumspace
\end{IEEEeqnarray}
\end{proposition}

\begin{proof}
Using a union bound
and Proposition~\ref{prop:exp-decay-not-pivot-condition-some-long-interval},
\begin{IEEEeqnarray*}{rCl}
    && \Prob{ \exists (r,s], s-r \geq w\colon \lnot \PivotCondition_{(r,s]} }   \\
    &\leq& \sum_{(r,s], s-r \geq w} \Prob{ \lnot \PivotCondition_{(r,s]} }   \\
    &\leq& 2\Th^2 \exp(-\constPivotRace  p w).
\end{IEEEeqnarray*}
\end{proof}

\begin{proposition}
\label{prop:all-w-pivots-are-pivots}
If $\SecurityCondition$,
then
for a time horizon $\Th$
and $w > \frac{2 \ln(\sqrt2\Th)}{\constPivotRace  p}$,
\begin{IEEEeqnarray}{C}
    \Prob{ \exists t\colon \WeakPivot{t} \land \lnot \TPivot{t} } \nonumber \\
    \leq 2\Th^2 \exp(-\constPivotRace  p w).
\IEEEeqnarraynumspace
\end{IEEEeqnarray}
\end{proposition}

\begin{proof}
If some $t$ is a weak pivot (with $w \geq \frac{2 \ln(\sqrt2\Th)}{\constPivotRace  p}$)
and $t$ is not a pivot, then
$\exists (r,s] \ni t$
with $s-r \geq w$
such that\\$\lnot \PivotCondition_{(r,s]}$.
But the probability for this is bounded
accordingly by Proposition~\ref{prop:exp-decay-not-pivot-condition-any-long-interval}.
\end{proof}

\begin{proposition}
\label{prop:fixed-t-is-w-pivot-with-good-probability}
If $\SecurityCondition$,
then
for time horizon $\Th$,
\begin{IEEEeqnarray}{C}
    \forall t\colon\quad
    \Prob{ \WeakPivot{t} \mid \Unique{t} } \geq p_1
\IEEEeqnarraynumspace\end{IEEEeqnarray}
where $p_1 = \frac12 (1-\pa)^{2v-1}>0$ and $\frac{w}{2} > v = \frac{1}{\constPivotRace p}\ln\left(\frac{4(1+e^{-\constPivotRace p})}{(1-e^{-\constPivotRace p})^2}\right)$.
\end{proposition}

\begin{proof}
For $v < w/2$ to be determined later, 
consider the events
\begin{IEEEeqnarray}{rCll}
    E_1 &\triangleq& \{& \Nint{t-v,t+v} = 0\}, \\
    E_2 &\triangleq& \{& \forall (r,s] \ni t, s-r < w, (r,s] \notin (t-v,t+v] \colon \nonumber \\
    &&&\PivotCondition_{(r,s]} \}.
    \IEEEeqnarraynumspace
\end{IEEEeqnarray}
Note that, $E_1 \cap E_2 \subseteq \{\WeakPivot{t}\}$ and $\Prob{E_1 \mid \Unique{t}} = (1-\pa)^{2v-1}$.

For bounding $\Prob{\lnot E_2}$, we will use a union bound by carefully counting the number of intervals $(r,s] \ni t$ such that $s-r < w$ and $(r,s] \notin (t-v,t+v]$.
Let $u=s-r$.
For $u \leq v$, note that $(r,s] \ni t$ implies that $(r,s] \in (t-v,t+v]$.
One can check that for $v+1 \leq u \leq 2v$, there are $2(u-v)-1$ intervals $(r,s] \ni t$ such that $(r,s] \notin (t-v,t+v]$. For $2v+1 \leq u < w$, all intervals $(r,s]\ni t$ are such that $(r,s] \notin (t-v,t+v]$, and there are $u$ such intervals.
Therefore, from Proposition~\ref{prop:exp-decay-not-pivot-condition-some-long-interval} and a union bound,
\begin{IEEEeqnarray}{rCl}
    \Prob{\lnot E_2} &\leq& \sum_{u=v+1}^{w-1} \sum_{\substack{(r,s]\ni t \colon \\ s-r=u \land \\ (r,s] \notin (t-v,t+v]} } \Prob{\lnot \PivotCondition{(r,s]}} \nonumber \\
    &\leq& \sum_{u=v+1}^{2v} (2(u-v)-1) 2e^{-\constPivotRace  pu} + \sum_{u=2v+1}^{w-1} u 2e^{-\constPivotRace  pu} \nonumber \\
    &\leq& \sum_{k=1}^{v} 2(2j-1)e^{-\constPivotRace  p(v+j)} + \sum_{u=2v+1}^{w-1} 2u e^{-\constPivotRace  pu} \nonumber \\
    &\leq& \sum_{k=1}^{v} 2(2j-1)e^{-\constPivotRace  p(v+j)} + \sum_{u=2v+1}^{\infty} 2u e^{-\constPivotRace  pu} \nonumber \\
    &=& \frac{2e^{-\constPivotRace p(v+1)}(1-(2v+1)e^{-\constPivotRace pv})}{1-e^{-\constPivotRace p}}  \nonumber \\
    && + \frac{4e^{-\constPivotRace p(v+2)}(1-e^{-\constPivotRace pv})}{(1-e^{-\constPivotRace p})^2}  \nonumber \\ 
    && + \frac{2(2v+1)e^{-\constPivotRace p(2v+1)}}{1-e^{-\constPivotRace p}} + \frac{2e^{-\constPivotRace p(2v+2)}}{(1-e^{-\constPivotRace p})^2} \nonumber \\
    &=& \frac{2e^{-\constPivotRace p(v+1)}}{1-e^{-\constPivotRace p}} + \frac{4e^{-\constPivotRace p(v+2)}-2e^{-\constPivotRace p(2v+2)}}{(1-e^{-\constPivotRace p})^2} \nonumber \\
    &\leq& \frac{2e^{-\constPivotRace p(v+1)}}{1-e^{-\constPivotRace p}} \left(1 + \frac{2e^{-\constPivotRace p}} {1-e^{-\constPivotRace p}} \right) \nonumber \\
    &\leq& \frac{2e^{-\constPivotRace pv}(1+e^{-\constPivotRace p})}{(1-e^{-\constPivotRace p})^2}
\end{IEEEeqnarray}

We may choose $v = \frac{1}{\constPivotRace p}\ln\left(\frac{4(1+e^{-\constPivotRace p})}{(1-e^{-\constPivotRace p})^2}\right)$, so that $\Prob{\lnot E_2} \leq \frac12$.

It is easy to see that $\Prob{E_2 \mid E_1 \cap \{\Unique{t}\}} \geq \Prob{E_2 \mid E_1} \geq \Prob{ E_2 }$.
\begin{IEEEeqnarray*}{rCl}
    \Prob{ \WeakPivot{t} \mid \Unique{t}}
    &\geq& \Prob{ E_1 \cap E_2 \mid \Unique{t}} \\
    &\geq& \Prob{ E_1 \mid \Unique{t}} \Prob{ E_2 }   \\
    &\geq& \frac12 (1-\pa)^{2v-1}.
\end{IEEEeqnarray*}
for the given choice of $v$.
\end{proof}

\begin{proposition}
\label{prop:fixed-t-within-T-is-w-pivot-with-good-probability}
If $\SecurityCondition$,
then
for horizon $\Th$ and
$w > \frac{2}{\constPivotRace p}\ln\left(\frac{4(1+e^{-\constPivotRace p})}{(1-e^{-\constPivotRace p})^2}\right)$,
\begin{IEEEeqnarray}{rl}
    \forall t\colon &
    \Prob{ \exists t' \in (t, t+\Tp]\colon \WeakPivot{t'} \land \Unique{t'} } \nonumber \\ &\geq
    1 - \exp(-\constPivotInterval  \Tp/w), 
\IEEEeqnarraynumspace\end{IEEEeqnarray}
with $\constPivotInterval  = \frac{p_1\pu}{2}$.
\end{proposition}

\begin{proof}
Let $k$ be the largest integer such that $\Tp\geq 2wk$. For $i=0,...,(k-1)$, define $t_i=t+(2i+1)w$ and
\begin{IEEEeqnarray}{rCl}
    E_i &\triangleq& \{\WeakPivot{t_i} \land \Unique{t_i}\} \\
    E &\triangleq& \{ \exists t' \in (t, t+\Tp]\colon \WeakPivot{t'} \land \Unique{t'} \}.
\end{IEEEeqnarray}

Thus, we have $\bigcup_{i=0}^{k-1} E_i \subseteq E$,
and by construction $E_i$ are independent.
Hence,
\begin{IEEEeqnarray}{rCl}
    \Prob{ E}
    &\geq& \Prob{\bigcup_{i=0}^{k-1} E_i}
    = 1 - \Prob{\bigcap_{i=0}^{k-1} \lnot E_i}   \nonumber\\
    &\geq& 1 - (1-p_1\pu)^k \nonumber \\
    &\geq& 1 - \exp( - p_1\pu k) \nonumber \\
    &=& 1 - \exp( - p_1\pu \Tp/2w),
\IEEEeqnarraynumspace\end{IEEEeqnarray}
where we have used Proposition~\ref{prop:fixed-t-is-w-pivot-with-good-probability}.
\end{proof}

\begin{proposition}
\label{prop:forall-t-within-T-is-w-pivot-with-good-probability}
If $\SecurityCondition$,
then
for horizon $\Th$,
$w > \frac{2}{\constPivotRace p}\ln\left(\frac{4(1+e^{-\constPivotRace p})}{(1-e^{-\constPivotRace p})^2}\right)$ and $\Tp > \frac{w\ln(\Th)}{\constPivotInterval }$,
\begin{IEEEeqnarray}{C}
    \Prob{ \forall t\colon \exists t' \in (t, t+\Tp]\colon \WeakPivot{t'} \land \Unique{t'} } \nonumber\\
    \geq 1 - \Th\exp(-\constPivotInterval  \Tp/w).
\IEEEeqnarraynumspace\end{IEEEeqnarray}
\end{proposition}

\begin{proof}
By a union bound over all $\Th$ possible time slots,
and using Proposition~\ref{prop:fixed-t-within-T-is-w-pivot-with-good-probability}.
\end{proof}

\subsubsection{Proof of Lemma~\ref{lem:frequent_pivots_prob_bound}}

\begin{proof}
Finally, to prove Lemma~\ref{lem:frequent_pivots_prob_bound}, let
\begin{IEEEeqnarray*}{rCl}
    E_1 &\triangleq&
        \{ \forall t\colon \exists t' \in (t, t+\Tp]\colon \WeakPivot{t'}\land \Unique{t'} \}   \\
    E_2 &\triangleq&
        \{ \forall t\colon \WeakPivot{t} \Rightarrow \TPivot{t} \}   \\
    E &\triangleq&
        \{ \forall t\colon \exists t' \in (t, t+\Tp]\colon \TPivot{t'}\land \Unique{t'} \}.
\end{IEEEeqnarray*}
Note that $E_1 \cap E_2 \subseteq E$.
Then we apply a union bound
on the probabilities from Propositions~\ref{prop:forall-t-within-T-is-w-pivot-with-good-probability} and~\ref{prop:all-w-pivots-are-pivots}.
\begin{IEEEeqnarray}{rCl}
    \Prob{ \lnot E }
    &\leq& \Prob{ \lnot E_1 } + \Prob{ \lnot E_2 }
    \leq 2\Th^2e^{-\constPivotRace pw}+ \Th e^{-\constPivotInterval  \Tp/w}.
\end{IEEEeqnarray}

Let $\kappa' = \kappa + \ln\Th$. 
Pick $w$ such that $w = \frac{2 \ln(\sqrt2\Th)+\Omega(\kappa)}{\constPivotRace  p}$. This ensures that the probability $2\Th^2 e^{-\constPivotRace  p w}$ corresponding to having more adversarial than honest slots in some interval of size at least $w$,
is $\negl(\kappa)$.

Finally, we pick $\gamma$ so that the probability $\Th e^{-\constPivotInterval  \Tp/w}$ corresponding to not finding a pivot slot in some interval of $\Tp$ slots, is $\negl(\kappa)$. 
Therefore we get
$\gamma \geq \frac{\ln(\Th)+\Omega(\kappa)}{\constPivotInterval }w$.
Combining these, we have
$\gamma \geq \frac{\Omega((\ln(\Th)+\kappa)^2)}{\constPivotRace \constPivotInterval  p}$.
Choose $\constFrequentPivots = \constPivotRace \constPivotInterval $
\end{proof}

\subsection{Proof of Lemma~\ref{lem:num_blocks_bound}}
\label{sec:appendix-proof-num-blocks-bound}

\begin{proof}
Define the event $F_t$ as
\begin{IEEEeqnarray}{C}
    \max_{r<t \colon \Unique{r} \land (\Nint{r,t} \geq \Lint{r,t})} \Nint{r,t} \geq \bwslot.
\end{IEEEeqnarray}
This event can be equivalently expressed as
\begin{IEEEeqnarray}{C}
    \exists r < t \colon
        \Unique{r}
        \land (\Nint{r,t} \geq \Lint{r,t})
        \land (\Nint{r,t} \geq \bwslot).
        \IEEEeqnarraynumspace
\end{IEEEeqnarray}

The event $\{\lnot\FewBlockOpps\}$ can be expressed as $\bigcup_{t\leq \Th} F_t$.
Then for some fixed $\Td$,
\begin{IEEEeqnarray}{rCl}
    \Prob{F_t} &\leq& \Prob{\bigcup_{r=0}^{t-1} \left\{\Nint{r,t} \geq \Lint{r,t} \land \Nint{r,t} \geq \bwslot \right\}} \nonumber \\
    &\leq& \sum_{r=0}^{t-\Td} \Prob{ \Nint{r,t} \geq \Lint{r,t} } + \sum_{r=t-\Td}^{t-1} \Prob{\Nint{r,t} \geq \bwslot} \nonumber \\
    &\leq& \sum_{k=\Td}^{\infty} 2 \exp\left(-\constPivotRace pk\right) + \Td \exp\left(-\frac{\epsilon_2^2}{2+\epsilon_2} p_A \Td\right) \nonumber \\
    &=& \frac{2\exp(-\constPivotRace  p\Td)}{1-\exp(-\constPivotRace  p)} + \Td \exp\left(-\frac{\epsilon_2^2}{2+\epsilon_2} p_A \Td\right) \nonumber \\
    &\leq& 2\Td \exp\left(-\constNumBlocks  p \Td\right),
\end{IEEEeqnarray}
for $\Td \geq \frac{2}{1-\exp(-\constPivotRace  p)}$ and $\constNumBlocks  = \min\left\{ \constPivotRace , \frac{\epsilon_2^2}{\epsilon_2+2}\frac{\pa}{p}\right\}$. By using a union bound
over the execution horizon $\Th$, we get
\begin{IEEEeqnarray}{C}
    \Prob{\lnot \FewBlockOpps} \leq
    2 \Th \Td \exp(-\constNumBlocks  p \Td)
    \leq 2 \Th^2 \exp(-\constNumBlocks  p \Td)
    \IEEEeqnarraynumspace
\end{IEEEeqnarray}
We then set $\Td = \frac{2\ln(\sqrt{2}\Th) + \Omega(\kappa)}{\constNumBlocks  p}$ to make this probability $\negl(\kappa)$.
\end{proof}

\subsection{Proof of Lemma~\ref{prop:unique_download_longest_header_chain_rule}}
\label{sec:appendix-proof-equi-avoid}

\begin{proof}
Let $t_1,...,t_m$ be the uniquely successful slots in $(0,\Th]$.
Let $b_j$ be the 
block from slot $t_j$ for some $1\leq j\leq m$.

For induction, assume that $\grprom{(0,t_j-1]}$ holds. Using this, we will show that $\grprom{(0,t_{j+1}-1]}$ holds. For the base case, this is true for $j=1$ since $t_1$ is the first uniquely successful slot by definition.
Suppose that there is a chain $\Chain'$ 
in the header tree of an honest node in slot $t_j$
such that $\len{\Chain'}\geq\len{b_j}$. Note that the tip of $\Chain'$ can not be a unique block because unique blocks have increasing heights as per Lemma~\ref{prop:chain_growth}. Therefore the tip of $\Chain'$ is from an adversarial slot. Consider such a chain $\Chain'$ ending in a block from an adversarial slot $s_j \leq t_j$. Let $r_j$ be the last uniquely successful slot such that the block $b_j'$ from that slot is in $\Chain'$. Then,
\begin{IEEEeqnarray}{C}
    \len{\Chain'} \leq \len{b_j'} + \Nint{r_j,s_j}.
\end{IEEEeqnarray}
From the assumption of $\grprom{(0,t_j-1]}$ and part~(\ref{item:inc_height}) of  Lemma~\ref{prop:chain_growth}, 
\begin{IEEEeqnarray}{C}
\label{eq:cgeq2_eqavoid}
    \len{b_j} \geq \len{b_j'} + \Lint{r_j,t_j}.
\end{IEEEeqnarray}
Since $\len{\Chain'}\geq\len{b_j}$, this would mean that $\Nint{r_j,s_j} \geq \Lint{r_j,s_j}$.
As a block from a uniquely successful slot, $b_j'$ was downloaded by all honest nodes within slot $r_j$. Therefore, there are at most $\Nint{r_j,s_j}$ blocks on the chain $\Chain'$ that are yet to be downloaded. Therefore the number of blocks to be downloaded by each honest node on $\Chain'$ is at most 
\begin{IEEEeqnarray}{C}
    \max_{r_j<s_j \colon \Unique{r_j} \land \left(\Nint{r_j,s_j} \geq \Lint{r_j,t_j}\right)} \Nint{r_j,s_j} = W_{s_j,t_j}.
\end{IEEEeqnarray}
Next, we count the number of such chains $\Chain'$ with distinct block production opportunities at the tip. Due to the equivocation avoidance policy, the adversary can make honest nodes download at most one chain per adversarial block production opportunity in slots $s_j \leq t_j$. The total number of blocks to be downloaded in all these chains combined is $\sum_{s_j<t_j} A_{s_j} W_{s_j,t_j}$.

Finally, from the proof of Lemma~\ref{lem:loner_download}, we note that the prefix of $b_j$ has at most $W_{t_j-1,t_j-1}$ blocks that need to be downloaded by any honest node. Therefore, the total number of blocks that any honest node needs to download before downloading $b_j$ is at most
\begin{IEEEeqnarray}{C}
    W_{t_j-1,t_j-1} + \sum_{s_j \leq t_j} A_{r_j} W_{s_j,t_j}.
\end{IEEEeqnarray}
From the definition of $\BoundedEq$, this is less than $\bwslot$. Therefore, every honest node can download $b_j$ within the time slot $t_j$. This completes the induction step by showing that $\grprom{(0,t_{j+1}-1]}$.
For $j=m$, we conclude with $\grpromempty$ as required.
\end{proof}

\subsection{Proof of Lemma~\ref{lem:bounded_equivocations}}
\label{sec:appendix-proof-bounded-equivocations}

\begin{proof}
From Lemma~\ref{lem:num_blocks_bound}, we already know that for $N = \pa T(1+\epsilon_2)$ and $\Td > \max\left\{ \frac{2}{1-\exp(-\constPivotRace p)}, \frac{2\ln(\sqrt{2}\Th)}{\constNumBlocks  p} \right\}$, we have
\begin{IEEEeqnarray}{C}
    \Prob{\lnot(\forall t\leq \Th: W_{t,t} < N)} \leq 2\Th^2 e^{-\constNumBlocks  p \Td}.
\end{IEEEeqnarray}
It is easy to see that for any given sample path (i.e. realization of $\exec$) and any $s \leq t$, $W_{s,t} \leq W_{t,t}$. Next, we can show that there exists some $\Tb$ such that $W_{s,t} = 0$ for all $s < t-\Tb$ and for all $t$, so that we have the following with overwhelming probability:
\begin{IEEEeqnarray}{rCl}
    W_{t-1,t-1} + \sum_{s \leq t} A_s W_{s,t} &\leq& N + N \Tb \beta \rho (1+\epsilon).
\end{IEEEeqnarray}
This is because in any $\Tb$ slots, there are at most $\Tb \beta \rho (1+\epsilon)$ adversarial block production opportunities with probability at least $1-\Th\exp\left(-\frac{\epsilon^2\beta\rho\Tb}{\epsilon+2}\right)$ (through a Chernoff bound and union bound).

To show that $W_{s,t} = 0$ for all $s < t-\Tb$ for a fixed $t$, 
\begin{IEEEeqnarray}{rl}
    &\Prob{ \exists s < t-\Tb: W_{s,t} > 0 } \\
    \leq& \Prob{ \exists s < t-\Tb, \exists r<s: \Nint{r,s} \geq \Lint{r,t} } \\
    \leq& \Prob{ \exists r < t-\Tb: \Nint{r,t-\Tb} \geq \Lint{r,t} } \\
    \leq& \Prob{ \exists r < t-\Tb: \Nint{r,t-\Tb} \geq \Lint{r,t-\Tb} + \Lint{t-\Tb,t}} \IEEEeqnarraynumspace \\
    \label{eq:random_walk_prob_reference_eqavoid}
    \leq& \Prob{ \exists r < t-\Tb: \Nint{r,t-\Tb} \geq \Lint{r,t-\Tb} + L} \IEEEeqnarraynumspace \\
    & + \Prob{ \Lint{t-\Tb,t} < L} \IEEEeqnarraynumspace
\end{IEEEeqnarray}
where we choose $L=\pu\Tb(1-\epsilon)$.
The second term is bounded by a Chernoff bound
\begin{IEEEeqnarray}{rl}
    \Prob{ \Lint{t-\Tb,t} < L} \leq&  \exp\left(-\frac{\epsilon^2\pu\Tb}{2}\right).
\end{IEEEeqnarray}

For calculating the first term, let 
\begin{IEEEeqnarray*}{c}
    X_n = L + \Lint{t-\Tb-n, t-\Tb} - \Nint{t-\Tb-n, t-\Tb}
\end{IEEEeqnarray*}
for $n \geq 0$ be a random walk. Let $p_{l} = \Prob{\exists n : X_n \leq 0 \mid X_0 = l}$, i.e. the probability that the random walk ever hits $0$ after starting from $l$.
We can observe that $p_1 = 1-\pu + \pu p_2$. We can also note that due to the translation invariance of the random walk,
\begin{IEEEeqnarray*}{rCl}
    p_2 &=& \Prob{\exists n : X_n \leq 1 \mid X_0 = 2} \Prob{\exists n > n_1 : X_{n_1} \leq 0 \mid X_{n_1} = 1} \\ &=& \Prob{\exists n : X_n \leq 0 \mid X_0 = 1}^2 = p_1^2.
\end{IEEEeqnarray*}
Therefore, we obtain $p_1 = \frac{1-\pu}{\pu}$ by solving $p_1 = 1-\pu + \pu p_1^2$. Finally, we note using the same logic as above that $p_L = p_1^L = \left(\frac{1-\pu}{\pu}\right)^L$ which is the required probability in the first term in \eqref{eq:random_walk_prob_reference_eqavoid}.

Therefore, we have
\begin{IEEEeqnarray}{rCl}
    \Prob{ \exists s \leq t-\Tb: W_{s,t} > 0 } &\leq& \left(\frac{1-\pu}{\pu}\right)^{\pu\Tb(1-\epsilon)} + \exp\left(-\frac{\epsilon^2\pu\Tb}{2}\right)
    \IEEEeqnarraynumspace
\end{IEEEeqnarray}
Finally, by a union bound over the required probabilities, we have for $\downloadLimit = \pa T (1+\beta\rho\Tb(1+\epsilon))(1+\epsilon_2)$,
\begin{IEEEeqnarray}{l}
    \Prob{\lnot \BoundedEq} \leq \Th\left(\frac{1-\pu}{\pu}\right)^{\pu\Tb(1-\epsilon)} \\ 
    + \Th\exp\left(-\frac{\epsilon^2\pu\Tb}{2}\right) + \Th\exp\left(-\frac{\epsilon^2\beta\rho\Tb}{\epsilon+2}\right) + 2\Th^2\exp(-\constNumBlocks  pT) \IEEEeqnarraynumspace \\
    \leq 5\Th^2 \exp(-\constBoundedEq  p\Tb).
\end{IEEEeqnarray}
Here, we choose $\Tb = T$ and
\begin{IEEEeqnarray}{c}
    \constBoundedEq  = \max\left\{ \constNumBlocks ,  \frac{\pu(1-\epsilon)}{p}\ln\left(\frac{\pu}{1-\pu}\right), \frac{\epsilon^2\pu}{2p}, \frac{\epsilon^2\beta\rho}{(\epsilon+2)p} \right\}.
\end{IEEEeqnarray}

Finally, we set $\Tb = \frac{2\ln(\sqrt{5}\Th) + \Omega(\kappa)}{\constBoundedEq  p}$ so that the required probability is negligible.

\end{proof}

\subsection{Proofs for Throughput and Bandwidth Consumption}
\label{sec:app-tp-bw-util-proofs}

\subsubsection{Proof of Lemma~\ref{lem:throughput}}

\begin{proof}

Due to Lemma~\ref{prop:chain_growth}, in any interval of slots $(t_1,t_2]$, the downloaded longest chain of every honest node grows by at least $\Lint{t_1,t_2}$ (even though all blocks on the chain may not be honest).
Therefore, corresponding to the interval $(t_1,t_2]$ with $t_2 \geq t_1 + T$, at least $\pu T(1-\epsilon)$ blocks are added to every node's downloaded longest chain with probability
\begin{IEEEeqnarray}{l}
    \Prob{\Lint{t_1,t_2} \geq \pu T (1-\epsilon)} \nonumber \\
    \geq \Prob{\Lint{t_1,t_2} \geq \pu (t_2-t_1) (1-\epsilon)} \geq 1 - \exp\left(\frac{\epsilon^2}{2}\pu T\right). \IEEEeqnarraynumspace
\end{IEEEeqnarray}
Now let $N=\pu T(1-\epsilon)$. Consider any $N$ consecutive blocks in a valid blockchain. Let $t_1'$ and $t_2'$ be the time slots corresponding to the first and last blocks respectively in this set, and let $T'=t_2'-t_1'$. From the above probability bound, we have $T' \leq T = \frac{N}{\pu(1-\epsilon)}$. 
Also, with probability at least $1-\exp\left(-\frac{\epsilon'^2}{2+\epsilon'}\pa T'\right)$, there are at most $\pa T'(1+\epsilon')$ adversarial slots in $(t_1',t_2']$, hence there are at most $\pa T'(1+\epsilon')$ adversarial blocks in the $N$ consecutive blocks.

Therefore, corresponding to every interval $(t_1,t_2]$, there are at least $\pu T(1-\epsilon) - \pa T(1+\epsilon') = (\pu-\pa)T(1-\epsilon_4)$ honest blocks in any node's downloaded longest chain with probability at least $1-\exp(-\alpha_4 T)$ for some constant $\alpha_4$. Finally, we note that $\theta = \pu-\pa = 2\pu - p$.
\end{proof}

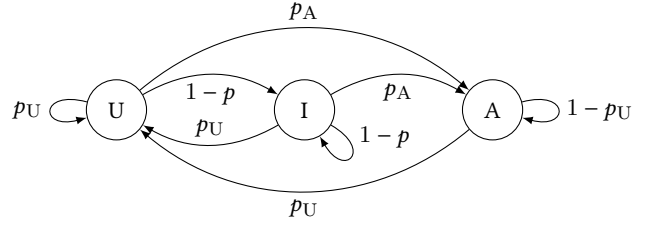
\begin{figure}
    \centering
    \begin{tikzpicture}[x=2.5cm,y=2cm]

        \node [circle,draw,minimum width=0.8cm,minimum height=0.8cm] (nU) at (-1,0) {U};
        \node [circle,draw,minimum width=0.8cm,minimum height=0.8cm] (nI) at (0,0) {I};
        \node [circle,draw,minimum width=0.8cm,minimum height=0.8cm] (nA) at (1,0) {A};
        
        \draw [-latex] (nU) to [bend left=40] node [midway,above] {$\pa$} (nA);
        \draw [-latex] (nA) to [bend left=40] node [midway,below] {$\pu$} (nU);
        
        \draw [-latex] (nU) to [bend left] node [midway,below] {$1-p$} (nI);
        \draw [-latex] (nI) to [bend left] node [midway,above] {$\pu$} (nU);
        
        \draw [-latex] (nI) to [bend left] node [midway,below] {$\pa$} (nA);
        
        \draw [-latex] (nI) to [out=330,in=300,looseness=8] node [midway,above right] {$1-p$} (nI);
        \draw [-latex] (nU) to [out=165,in=195,looseness=8] node [left] {$\pu$} (nU);
        \draw [-latex] (nA) to [out=15,in=-15,looseness=8] node [right] {$1-\pu$} (nA);

    \end{tikzpicture}%
    \vspace{-0.5em}%
    \caption{An upper bound on the bandwidth utilization of our protocol can be calculated from the stationary distribution of this Markov chain}
    \label{fig:markov_chain}
\end{figure}

\subsubsection{Proof of Lemma~\ref{lem:passive_bandwidth}}

\begin{proof}
Consider time slots $t_1$ and $t_2 \geq t_1 + T$. Due to the safety of $\protocol$, we know that $\LOG{i}{t_1} \preceq \LOG{i'}{t_2}$ for any honest nodes $i,i'$. The last block in $\LOG{i}{t_1}$ must have a time slot $t_1' \geq t_1 - 2\Tconf$ because between $t_1-2\Tconf$ and $t_1-\Tconf$, there is at least one unique pivot slot which contributes a block to $\LOG{i}{t_1}$. Therefore $\LOG{i'}{t_2} \setminus \LOG{i}{t_1}$ contains only blocks with time slots in the interval $(t_1',t_2']$ where $t_2'=t_2-\Tconf$. Note that blocks in the confirmed chain must have increasing time slots, so their number is limited by the number of slots with block proposal, \ie $\Bint{t_1',t_2'}$. The average number of slots with block proposal in the interval $(t_1',t_2']$ is $p(t_2'-t_1') \leq p(t_2-t_1+\Tconf) = p(T+\Tconf)$.
Then by a Chernoff bound,
\begin{IEEEeqnarray}{rCl}
    \Prob{\Bint{t_1',t_2'} > pT(1+\epsilon_5)} 
    &\leq& \exp\left( -\alpha_5 T \right)
\end{IEEEeqnarray}
for sufficiently large $T>\Tconf$ and some constant $\alpha_5$.
\end{proof}

\subsubsection{Proof of Lemma~\ref{lem:active_bandwidth}}

\begin{proof}
Consider the Markov chain shown in Figure~\ref{fig:markov_chain} with three states---$\mathrm{U}$ corresponding to a uniquely successful slot, $\mathrm{I}$ corresponding to a slot without a block proposal such that the most recent block proposal was a uniquely successful slot, and $\mathrm{A}$ corresponding to adversarial slots or slots without block proposals such that the most recent block proposal was an adversarial slot.

The stationary distribution of this Markov chain is
\begin{IEEEeqnarray}{C}
    \pi_{\mathrm{U}} = \pu,
    \quad
    \pi_{\mathrm{I}} = \frac{\pu(1-p)}{p},
    \quad
    \pi_{\mathrm{A}} = \frac{\pa}{p}.
\end{IEEEeqnarray}

Note that in time slots corresponding to the $\mathrm{I}$ (idle) state, there are no fresh blocks to be downloaded because the most recent block proposal was a unique honest block which was downloaded within $1$ slot. Therefore, on average, in $\idleslots$ fraction of time slots, every honest node's bandwidth remains idle, where 
\begin{IEEEeqnarray}{rCl}
    \idleslots \geq \pi_{\textrm{I}}
    &=& \frac{\pu(1-p)}{p} \nonumber \\
    &=& \frac{1}{2}(1-p)(1+\epsilon_1)\nonumber \\
    &\geq& \left(\frac{1-p}{2}\right).
\end{IEEEeqnarray}
(For $\epsilon_1$, see the proof of Theorem~\ref{thm:security}.)
Finally, by a Chernoff bound, the probability that for a given $t_1,t_2$, there are at least $\idleslots T(1-\epsilon_6)$ slots in the $\mathrm{I}$ state in the interval $(t_1,t_2]$ is at least $1-\exp\left(-\frac{\epsilon_6^2}{2}\idleslots T\right)$.
\end{proof}

%% file: appendix_proofs_parallel.tex
\section{Security of Parallel Chains}
\label{sec:appendix-proofs-parallel}
The below security theorem holds for any download rule which satisfies the requirement in Theorem~\ref{thm:security}, and in addition leaves a fraction $\idleslots \in (0,1)$ of the total bandwidth unutilized (\cf Lemma~\ref{lem:active_bandwidth}). The latter requirement can be easily achieved for any download rule for any desired $\idleslots \in (0,1)$ by increasing the time slot duration by a factor of $\frac{1}{\idleslots}$ and only downloading blocks in the first $\idleslots$ fraction of the time slot.

Also note that the below theorem holds under a static corruption adversary (\ie, the adversary decides which nodes to corrupt before the randomness of the protocol is drawn).

\begin{theorem}
    For all $K \in \mathbb{N}$ and download rules $\dlrule$ such that
    \begin{IEEEeqnarray}{c}
        \Prob{\exec \colon \lnot \grpromempty(\exec,\dlrule)} \leq \negl(\kappa), \IEEEeqnarraynumspace
    \end{IEEEeqnarray}
    if $\SecurityConditionFull$ for some $\epsilon_1 \in (0,1)$, $\tau=\Omega(\kappa+\ln\Th)$, $\Tconf=\Omega((\kappa+\ln\Th)^2)$, 
    Lemma~\ref{lem:active_bandwidth} holds for some $\idleslots \in (0,1)$,
    and $m = 1 + \frac{\idleslots}{\bwpassive}C\tau(1-\epsilon_7)$,
    then the protocol $\protocolpc$ with the download rule $\dlrule$
    is secure
    with parameter $\Tlive = \Omega((\kappa+\ln\Th)^2)$.
\end{theorem}
\begin{proof}
    Consider a particular protocol instance $\Pi_{\index}$. Define $\dC_{i,\index}$ to be the longest downloaded chain of node $i$ for protocol instance $\Pi_{\index}$. 
    From Theorem~\ref{thm:security}, for the given $\rho$, $\tau$ and $\Tconf=\Tp$,
    each protocol instance $\Pi_{\index}$ satisfies safety and liveness with respect to the ledger defined by $\dC_{i,\index}(t)\trunc{\Tconf}$ and for nodes $i$ for which $\Pi_{\index}$ is the primary chain, expect with probability $\negl(\kappa)$. By a union bound, safety and liveness for each protocol instance holds over $m=\poly(\kappa)$ protocol instances as well.
    Due to safety of $\Pi_{\index}$, $\dC_{i,\index}(t)\trunc{\Tconf} \preceq \dC_{j,\index}(t')\trunc{\Tconf}$ or \\ $\dC_{j,\index}(t')\trunc{\Tconf} \preceq \dC_{i,\index}(t)\trunc{\Tconf}$ for all time slots $t,t'$ and all honest nodes $i,j$
    for which $\Pi_{\index}$ is the primary chain.
    However, this holds even if $\Pi_{\index}$ is not the primary chain for node $i$ or $j$ because such nodes receive all block headers, determine the longest header chain based on them, and then download its confirmed prefix.
    More concretely, an adversary that pushes an inconsistent longest header chain to a node $j$ for which $\Pi_{\index}$ is a secondary chain, can also do so with headers and contents for a node $j'$ for which $\Pi_{\index}$ is the primary chain, thus causing a safety violation, which contradicts the earlier observation.
    Since all nodes have consistent confirmed chains (\ie $\dC_{i,\index}(t)\trunc{\Tconf} \preceq \dC_{j,\index}(t')\trunc{\Tconf}$ or $\dC_{j,\index}(t')\trunc{\Tconf} \preceq \dC_{i,\index}(t)\trunc{\Tconf}$) for each protocol instance and the combined ledger is derived 
    by ordering the blocks in all confirmed chains deterministically by their time slot,
    this implies safety of $\protocolpc$ (\ie, $\forall \text{ honest } i, j: \forall t, t': \LOG{i}{t} \preceq \LOG{j}{t'} \lor \LOG{j}{t'} \preceq \LOG{i}{t}$).
    To show liveness, we first show that confirmed secondary chain blocks are downloaded with bounded delay.
    From Lemma~\ref{lem:active_bandwidth}, in any interval of $\tilde{T}$ slots, the bandwidth of each node is not requested for downloads related to the primary chain but available to download secondary chain blocks in at least $\idleslots \tilde{T}(1-\epsilon_5)$ slots. Further, from Lemma~\ref{lem:passive_bandwidth}, in any interval of $\tilde{T}$ slots, the confirmed secondary chains grow by at most $\bwpassive \tilde{T}(1+\epsilon_6)$ blocks. These events happen with probability at least $1-\negl(\kappa)$ over a time horizon $\Th$ with $\tilde{T}=\Omega(\kappa+\ln\Th)$. By a union bound over $m=\poly(\kappa)$ number of chains, these hold with at least $1-\negl(\kappa)$ probability over all chains. Therefore, in $\tilde{T}$ slots, all confirmed blocks in $m-1$ secondary chains can be downloaded, where $m-1 = \frac{\idleslots \tilde{T}(1-\epsilon_5)}{\bwpassive \tilde{T}(1+\epsilon_6)}C\tau
    = \frac{\idleslots}{\bwpassive}C\tau(1-\epsilon_7)$ 
    for some $\epsilon_7$.
    Finally, note that liveness of each protocol instance guarantees liveness of the parallel chains construction. As per the transaction distribution rule described in Appendix~\ref{sec:appendix-parallel-chains-pseudocode}, each transaction belongs to a particular protocol instance. 
    By the liveness of each protocol instance, any transaction input to all honest nodes in time slot $t$, is included in $\dC_i(t)\trunc{\Tconf}$ for $t'\geq t + \Tp + \Tconf$ (see Proof of Lemma~\ref{lem:combinatorial} in Appendix~\ref{sec:appendix-combinatorial}) and all nodes $i$ for which the corresponding protocol instance is primary.
    Moreover, all honest nodes download confirmed secondary chains within $\tilde{T}$ delay.
    Therefore, $\protocolpc$ satisfies liveness with total latency $\Tp+\Tconf+\tilde{T} = \Omega((\kappa+\ln\Th)^2)$.
\end{proof}

%% file: appendix_conflux.tex
\section{Conflux Inclusion Rule}
\label{sec:appendix-conflux}

In order to prevent the throughput from vanishing as the resilience $\beta$ approaches $1/2$, we incorporate a modified version of the block inclusion rule from Conflux \cite{conflux} (also used in \cite{parallel}). In addition to the hash of the parent block, the header of a block $b$ also contains references to (hashes of) at most $R$ blocks which have time slots earlier than $b$ and are neither in the prefix nor are referenced by any blocks in the prefix of $b$. Moreover, in each chain, at most one block from each time slot may be referred. An honest block producer chooses to include the $R$ newest (by time slot) fully downloaded blocks in their view that satisfy the above criteria. The parameter $R$ is to be determined below. Blocks containing references that do not follow the above criteria will be considered invalid. The consensus protocol still uses the longest chain rule.

Note that downloading and validating a block now requires (in addition to downloading the block itself) downloading the content of all blocks in its prefix and all blocks referenced by blocks in the prefix. Unlike \cite{parallel}, we do not consider the reference links to be transitive as this would blow up the number of referred blocks to be downloaded. The output ledger of a node $i$ in slot $t$ (i.e. $\LOG{i}{t}$) will be formed by considering its truncated longest chain (i.e. $\dC_i(t)\trunc{\Tconf}$) and inserting blocks referred by a block $b$ between the parent of $b$ and $b$, in increasing order of their time slot. This may result in some transactions becoming invalid due to conflicting transactions appearing before them in the ledger. Such transactions would be removed (\emph{sanitized}) while obtaining the ledger.

\subsection{Security}

For security of the inclusive protocol, it is enough to set the time slot size to be $\tau = \DeltaHeader + \frac{\bwslot R}{\bwtime}$ (previously $\DeltaHeader + \frac{\bwslot}{\bwtime}$) where $\bwslot$ is set according to Theorem~\ref{thm:security} (we do the analysis below for the freshest block download rule, but it can be done for the equivocation avoidance download rule as well). Since each block contains at most $R$ references, the number of blocks to be downloaded in the prefix of any honest freshest block increases at most by a factor of $R$. By setting the slot size as above, we ensure that the honest block proposed in every uniquely successful slot is downloaded (along with its prefix and references therein) within the same slot.

\subsection{Single Chain Throughput}

\begin{lemma}
\label{lem:throughput_inclusive}
    If $R=\gamma p(1-\epsilon_7)$, there exists a constant $T_4$ such that for any honest node $i$ and time slots $t_1,t_2\geq t_1+T,t\geq t_2+\Tconf$ with $T\geq T_4$, $\LOG{i}{t}$ contains at least $\tpinc T(1-\epsilon_8)$ blocks proposed by honest nodes in slots $(t_1,t_2]$, with probability at least $1-\exp(-\alpha_4 T)$, where $\tpinc = \pu$.
\end{lemma}

Lemma~\ref{lem:throughput_inclusive} indicates that the average throughput of a single instance of the inclusive protocol is at least $\frac{\tpinc}{\tau}$ blocks per second.

To prove Lemma~\ref{lem:throughput_inclusive}, we only need to show that every honest block from a uniquely successful slot is included in the longest chain of every node either directly on the chain or through a reference. This will be achieved by setting $R$ to be large enough so that in any interval of slots with $R$ block production opportunities, at least one honest block is included in the longest chain. Then such an honest block would include references to the $R$ most recent blocks which would collectively include (at least) all honest blocks from uniquely successful slots.

\begin{proof}
From the security analysis (Lemma~\ref{lem:frequent_pivots_prob_bound}), we have that $\Prob{\FrequentPivots} \geq 1-\negl(\kappa)$ where $\FrequentPivots$ is the event
\begin{equation}
\forall t \colon \exists t' \in (t, t+\Tp]\colon\  \TPivot{t'} \land \Unique{t'}.
\end{equation}
Moreover, we have shown in Lemma~\ref{lem:combinatorial} that the honest block proposed in a unique pivot slot remains in the longest downloaded chain of every honest node. This satisfies our requirement. Thus, we need to set $R=\gamma p(1+\epsilon_7)$ so that there are at most $R$ uniquely successful slots between two pivot slots, i.e.
\begin{equation}
    \forall t \colon \Prob{\Lint{t,t+\Tp} > R} \geq 1 - \exp\left(\frac{\epsilon_7^2}{\epsilon_7+2}\Tp p\right).
\end{equation}
Therefore, the lemma holds under $T_4 = \Tp$ and the conditions from Lemma~\ref{lem:frequent_pivots_prob_bound}.
\end{proof}

\subsection{Parallel Chains Throughput}

We still have honest nodes idle (not downloading any blocks) in at least $\idleslots \geq \frac{1-p}{2}$ fraction of slots. The average bandwidth required to download a confirmed chain still remains at $p$ blocks per slot. Therefore, we can increase the total throughput by constructing $m = 1 + \frac{\idleslots C}{\bwpassive/\tau}$ parallel chains resulting in aggregate throughput
\begin{IEEEeqnarray}{rClr}
    \TP_m &=& \left(1 + \frac{\idleslots}{\bwpassive} C\tau\right) \frac{\tpinc}{\tau}   \nonumber\\
    &\geq& \frac{(1-p)\pu}{2p} \, C    \qquad\text{ (for the `freshest block' rule)}   \nonumber\\
    &\geq& \frac{(1-p)}{4} \, C \text{ blocks per second}.
\end{IEEEeqnarray}
This is a constant fraction of the capacity $C$ which does not vanish as $\beta\to1/2$.